\documentclass{article}
\usepackage{PRIMEarxiv}

\usepackage[utf8]{inputenc} % allow utf-8 input
\usepackage[T1]{fontenc}    % use 8-bit T1 fonts
\usepackage{url}            % simple URL typesetting
\usepackage{booktabs}       % professional-quality tables
\usepackage{amsfonts}       % blackboard math symbols
\usepackage{nicefrac}       % compact symbols for 1/2, etc.
\usepackage{microtype}      % microtypography
\usepackage{lipsum}
\usepackage{fancyhdr}       % header
\usepackage{graphicx}       % graphics
\usepackage[table]{xcolor}
\graphicspath{{media/}}     % organize your images and other figures under media/ folder
\usepackage{makecell}
\usepackage[normalem]{ulem}
\usepackage{algpseudocode}
\usepackage[ruled]{algorithm2e} % For algorithms

\usepackage{amsfonts}
\usepackage{amsmath}
\usepackage{amssymb}
\usepackage{amsthm}
\usepackage{boxedminipage}
\usepackage{caption} 
\usepackage{enumitem}
\usepackage{graphicx}
\usepackage{multirow}
\usepackage{natbib}  
\usepackage{pifont}
\usepackage{pgfplots}
\usepackage{comment}
\usepackage{xcolor}         % colors
\usepackage{tikz}
\usepackage{thm-restate}
\setlist{nolistsep}
\usetikzlibrary{arrows}
\usetikzlibrary{patterns}
\usetikzlibrary{positioning,fit,calc,arrows.meta,backgrounds}
\usetikzlibrary{shapes.geometric}

\newtheorem{lemma}{Lemma}%
\newtheorem{definition}{Definition}%
\newtheorem{theorem}{Theorem}%
\newtheorem{proposition}{Proposition}%
\newtheorem{corollary}{Corollary}%
\newtheorem{example}{Example}

\newcommand{\Ch}{\mathrm{Ch}}

\newcommand{\zhnew}[1]{\textcolor{black}{#1}}
\newcommand{\zh}[1]{\textcolor{red}{Zhaohong says: #1}}

\newcommand{\cmark}{\ding{51}}%
\newcommand{\xmark}{\ding{55}}%
\newcommand{\gcmark}{\textcolor{green!60!black}{\cmark}}
\newcommand{\rxmark}{\textcolor{red!80!black}{\xmark}}

\pgfplotsset{compat=1.18}

\allowdisplaybreaks

\usepackage{authblk}

\newcommand{\False}{\operatorname{false}}
\newcommand{\val}{\operatorname{val}}
\newcommand{\occ}{\operatorname{occ}}

\newcommand{\LLMS}{\textsc{LocalLexMax-4-SAT}}

\definecolor{familyblue}{RGB}{221,235,247}
\definecolor{daycaregreen}{RGB}{226,239,218}
\definecolor{transitionorange}{RGB}{252,228,214}
\definecolor{mutedgray}{RGB}{245,245,245}

\tikzset{
 family/.style={draw,rounded corners,fill=familyblue,minimum width=17mm,minimum height=8mm,align=center},
 daycare/.style={draw,rounded corners,fill=daycaregreen,minimum width=18mm,minimum height=9mm,align=center},
 trans/.style={draw,rounded corners,fill=transitionorange,minimum width=20mm,minimum height=9mm,align=center},
 edge/.style={-{Latex[length=2mm]},thick},
 contract/.style={-{Latex[length=1.8mm]},semithick},
 faint/.style={draw=gray!55,dashed},
 note/.style={font=\footnotesize,align=center}
}

\title{\bf Fairness and Non-wastefulness in Matching with Siblings and Initial Enrollments: Compatibility and 
Complexity}

\author[1]{Haris Aziz}
\author[2,3]{Gergely Csáji}
\author[4,5]{Zhaohong Sun}
\author[4]{Makoto Yokoo}
\affil[1]{UNSW Sydney, Australia}
\affil[2]{HUN-REN SZTAKI, Hungary}
\affil[3]{Budapest University of Technology and Economics, Hungary}
\affil[4]{Kyushu University, Japan}
\affil[5]{CyberAgent, AI Lab, Japan}

\date{\vspace{-10mm}}

\begin{document}

\maketitle

\begin{abstract}
We study a two-sided matching problem between families and daycare centers in which a family may have multiple children with joint preferences, and some children may already be enrolled in daycare centers while seeking transfers. Our model combines and generalizes two well-studied matching frameworks: matching with couples and school choice with initial enrollments. The interaction between joint family preferences and initial enrollments creates fundamental challenges for the existence of desirable outcomes and the design of computationally efficient algorithms.

Since stable matchings need not exist, we decompose stability into individual rationality, fairness, and non-wastefulness, and study the trade-off between fairness and non-wastefulness from two complementary directions. To preserve non-wastefulness, we first consider master-list-based approaches and show that they are insufficient in the presence of initial enrollments. We then introduce \emph{resettlement}, which protects displaced families by allowing them to return to their initial assignments, and develop a greedy improvement algorithm that establishes the existence of matchings satisfying non-wastefulness together with resettlement-based fairness. To preserve fairness, we first examine a generalized cutoff approach and then develop a fairness-preserving non-wasteful algorithm that guarantees fairness together with a relaxed notion of non-wastefulness.

We further characterize the computational complexity of these solution concepts. While several fairness notions can be verified in polynomial time, deciding whether a non-wasteful matching satisfying them exists is NP-complete. Resettlement restores universal existence, but computing resettlement-based outcomes is PLS-hard, and verifying the stronger dominance-based fairness notion is coNP-complete. Our results identify both axiomatic and computational trade-offs in reconciling fairness and non-wastefulness in matching markets with siblings and initial enrollments.
\end{abstract}

% \haris{Groups are discussed in abstracrt; siblings discussed elsewhere; make it consistent.}

\section{Introduction}
The daycare matching problem is a pressing issue faced by many cities worldwide. In many regions, daycare slots are in high demand, with long waiting lists resulting from the scarcity of qualified teachers and limited physical space~\citep{KaKo24a, SYT+24a}. This creates substantial challenges for families seeking to balance work and caregiving responsibilities. Decentralized approaches to assigning children to daycare centers may lead to inefficient outcomes and fail to make effective use of available resources. To address this issue, we adopt a centralized market-design approach to daycare matching. By eliciting families’ preferences and daycare centers’ priorities, this approach aims to identify desirable matchings that balance efficiency and fairness.

The daycare matching problem is not only an important practical matching market~\citep{KMT14a,GJR+21a,STM+23a}, but also provides a theoretically rich model that combines and extends two established matching problems: matching with couples and matching with initial enrollments. First, each family may have multiple children and submit joint preferences over assignments of its children to daycare centers. This sibling feature generalizes the Hospitals/Residents problem with Couples~\citep{RoPe99a}, in which couples of doctors express joint preferences over pairs of hospital positions. Second, some children may already be enrolled in daycare centers while seeking transfers to other centers. Such initial enrollments have been studied extensively in school choice~\citep{HHKS+17a,STY18a,STY+23a}.

Combining these two features creates challenges that do not arise when either is considered in isolation. Stability is a well-established solution concept in two-sided matching markets~\citep{GaSh62a,Roth08a} and plays a central role in many real-world applications, including the National Resident Matching Program in the US~\citep{Roth84a} and school choice programs in Boston and New York~\citep{APRS05a}. In matching with couples, stability is typically defined through the absence of blocking individuals and blocking coalitions involving doctors, couples, and hospitals~\citep{KlKl05a,KlKl07a,KPR13a,PDB16a}. Even for couples of two doctors, defining a blocking coalition requires careful treatment of different cases depending on whether the couple applies to the same or different hospitals, the number of available seats, and the hospitals’ priorities~\citep{McMa10a,BMM14a,MMT17a}. For example, five specific cases are carefully analyzed in~\citep{BIS11a}.

A case-by-case definition becomes increasingly cumbersome when a family may contain more than two children, because a deviation may simultaneously involve several daycare centers, vacancies, priority comparisons, and changes to the assignments of multiple siblings. Initial enrollments introduce an additional complication: when some siblings move to new daycare centers, seats released from their initial or current assignments may become available to other siblings. Thus, simply extending the usual definitions for couples does not provide a convenient general formulation. This motivates our first research question:
\emph{How can stability be defined concisely and precisely in daycare matching with siblings and initial enrollments?}

Informally speaking, a feasible matching is stable if no family (i) receives an unacceptable assignment (blocking individual) or (ii) can obtain a strictly preferred assignment that the relevant daycare centers are willing to accept (blocking coalition). A family’s improvement may involve three forms of reallocation. Such reallocations can occur in two primary ways: by occupying unfilled slots or by displacing children with lower priority. When siblings are involved, a third form of reallocation becomes possible: a child may occupy a seat released by one of its own siblings as part of the same deviation. We refer to the latter as \emph{seat withdrawal}. By formulating blocking coalitions through daycare choice functions, we provide a unified definition that captures all three forms without enumerating the many possible configurations separately.

\begin{figure}[tb]
\centering

\begin{tikzpicture}[
    >=stealth,
    pfeil/.style={
        ->,
        shorten >=1pt,
        thick
    },
    undirected/.style={
        -,
        thick,
        dashed
    },
    onlytext/.style={
        text centered
    },
    fairbox/.style={
        draw=red!45,
        fill=red!5,
        dashed,
        rounded corners=5pt,
        thin
    },
    stabbox/.style={
        draw=blue!45,
        fill=blue!5,
        dashed,
        rounded corners=5pt,
        thin
    },
    nwbox/.style={
        draw=green!45!black,
        fill=green!5,
        dashed,
        rounded corners=5pt,
        thin
    },
    resetbox/.style={
        draw=purple!45,
        fill=purple!5,
        dashed,
        rounded corners=5pt,
        thin
    }
]

% ============================================================
% Nodes
% ============================================================

\node[onlytext] (FRM) at (-4,2) {FRM};
\node[onlytext] (FR)  at (-2,2) {FR};
\node[onlytext] (FRW) at (0,2)  {FRW};
\node[onlytext] (F)   at (2,2)  {F};
\node[onlytext] (FMR) at (4,2)  {FMR};

\node[onlytext] (WS) at (-2,0) {WS};
\node[onlytext] (S)  at (0,0)  {S};

\node[onlytext] (SFMR) at (6,0) {SFMR};

\node[onlytext] (NW)  at (2,-2) {NW};
\node[onlytext] (FNW) at (4,-2) {FNW};

% ============================================================
% Colored group boxes
% ============================================================

\begin{scope}[on background layer]

\node[
    fairbox,
    fit=(FRM)(FR)(FRW)(F)(FMR),
    inner xsep=13pt,
    inner ysep=22pt
] (fairgroup) {};

\node[
    stabbox,
    fit=(WS)(S),
    inner xsep=18pt,
    inner ysep=22pt
] (stabgroup) {};

\node[
    nwbox,
    fit=(NW)(FNW),
    inner xsep=18pt,
    inner ysep=22pt
] (nwgroup) {};

\node[
    resetbox,
    fit=(SFMR)(FMR),
    inner xsep=18pt,
    inner ysep=22pt
] (resetgroup) {};

\end{scope}

% ============================================================
% Group labels
% ============================================================

\node[
    font=\small\bfseries,
    text=red!60!black
] at ([yshift=-7pt]fairgroup.north)
{Fairness-related properties};

\node[
    font=\small\bfseries,
    text=blue!60!black
] at ([yshift=-7pt]stabgroup.south)
{Stability-related properties};

\node[
    font=\small\bfseries,
    text=green!50!black
] at ([yshift=7pt]nwgroup.south)
{Non-wastefulness-related properties};

\node[
    font=\small\bfseries,
    text=purple!60!black,
    align=center
] at ([yshift=8pt]resetgroup.east)
{Resettlement-related\\properties};

% ============================================================
% Implication relationships
% ============================================================

\draw[pfeil] (S) -- (WS);
\draw[pfeil] (S) -- (NW);
\draw[pfeil] (S) -- (F);

\draw[pfeil] (WS) -- (FR);

\draw[pfeil] (NW) -- (FNW);

\draw[pfeil] (F) -- (FMR);
\draw[pfeil] (SFMR) -- (FMR);

\draw[pfeil] (F) -- (FRW);
\draw[pfeil] (FRW) -- (FR);
\draw[pfeil] (FR) -- (FRM);

% ============================================================
% Compatibility relationships
% ============================================================

\draw[undirected] (NW) -- (SFMR);
\draw[undirected] (NW) -- (FMR);
\draw[undirected] (FNW) -- (SFMR);
\draw[undirected] (FNW) -- (F);

\end{tikzpicture}

\caption{
Relationships among the proposed properties.
Directed arrows indicate implication, while dashed edges indicate
compatibility between the corresponding properties.
S: Stability;
WS: Weak Stability;
NW: Non-wastefulness;
FNW: Fairness-preserving Non-wastefulness;
F: Fairness;
FMR: Fairness by Master List and Resettlement;
SFMR: Strong Fairness by Master List and Resettlement;
FRW: Fairness by Replacement and Withdrawal;
FR: Fairness by Replacement; and
FRM: Fairness by Replacement and Master List.
}
\label{fig:relations}

\end{figure}
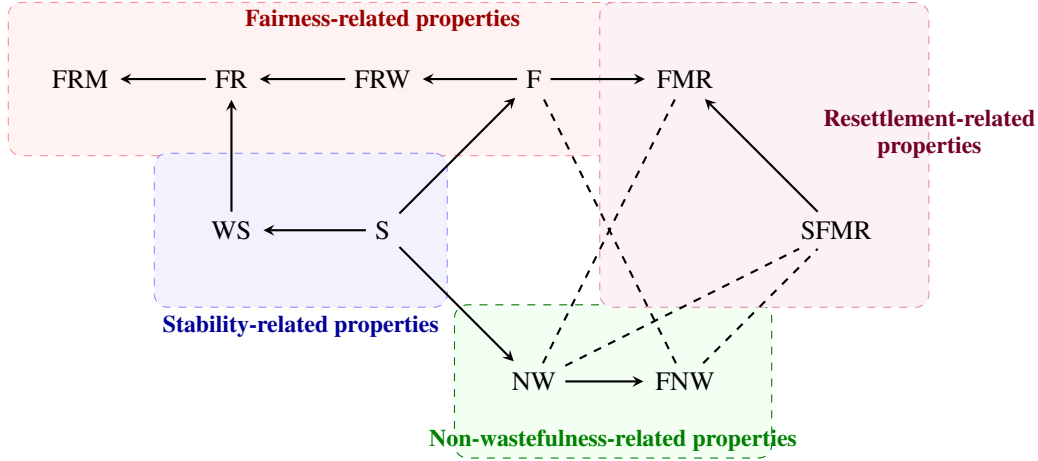

Stable matchings, however, need not exist once joint preferences of couples or siblings are introduced~\citep{Roth84a}. A natural approach is therefore to decompose stability into weaker requirements. Following the school-choice literature, we consider individual rationality, non-wastefulness, and fairness~\citep{EHYY14a,KHIY17a,AzSu21a,SoYe22a}. Individual rationality protects families from unacceptable assignments. Non-wastefulness rules out improvements that can be achieved without changing the assignments of other families. In our setting, such improvements include not only the use of seats that are currently vacant but also seats released through seat withdrawal. Fairness, in contrast, rules out profitable deviations that displace other families.

\begin{table}[tb]
\centering
\caption{
Comparison of the four proposed algorithms:
Serial Selection (SS), Greedy Improvement (GI),
Generalized Cutoff (GC), and Fairness-preserving Non-Wasteful (FN).
}
\label{tab:comparison}
\renewcommand{\arraystretch}{1.5}

\begin{tabular}{|l|c|c|c|c|}
    \hline
    \rowcolor{blue!10}
    \textbf{Properties}
    & \textbf{SS}
    & \textbf{GI}
    & \textbf{GC}
    & \textbf{FN}
    \\ \hline

    Non-wastefulness (NW)
    & \gcmark
    & \gcmark
    & \rxmark
    & \rxmark
    \\ \hline

    Fairness-preserving Non-wastefulness (FNW)
    & \gcmark
    & \gcmark
    & \rxmark
    & \gcmark
    \\ \hline

    Fairness (F)
    & \rxmark
    & \rxmark
    & \rxmark
    & \gcmark
    \\ \hline

    Fairness by Replacement and Withdrawal (FRW)
    & \rxmark
    & \rxmark
    & \rxmark
    & \gcmark
    \\ \hline

    Fairness by Replacement (FR)
    & \rxmark
    & \rxmark
    & \gcmark
    & \gcmark
    \\ \hline

    Fairness by Replacement and Master List (FRM)
    & \rxmark
    & \rxmark
    & \gcmark
    & \gcmark
    \\ \hline

    Fairness by Master List and Resettlement (FMR)
    & \rxmark
    & \gcmark
    & \rxmark
    & \gcmark
    \\ \hline

    Strong Fairness by Master List and Resettlement (SFMR)
    & \rxmark
    & \gcmark$^{\dagger}$
    & \rxmark
    & \gcmark$^{\ast}$
    \\ \hline
\end{tabular}

\vspace{1mm}

\begin{minipage}{0.96\linewidth}
\footnotesize
$^{\dagger}$ The SFMR guarantee for GI refers to its dominance-based
extension in Algorithm~\ref{alg:sfmr-greedy}.
\\
$^{\ast}$ FN is parameterized by a fairness property $A$. Each fairness entry in
the FN column refers to the corresponding choice of $A$, assuming that the
algorithm is initialized with a feasible and individually rational matching
satisfying $A$.
The SFMR entry refers to the instantiation $\mathsf A=\mathrm{SFMR}$
and does not imply that ordinary fairness (F) implies SFMR.
\end{minipage}
\end{table}

In the presence of siblings, these requirements may be incompatible, and initial enrollments further restrict which displaced families can be moved without violating their existing entitlements. Existing studies have considered couples or siblings and initial enrollments separately, particularly in hospital matching and school choice~\citep{NgVo18a,HHKS+17a,STY18a}. Some studies incorporate both features under restrictive preference domains~\citep{DMP22a,CEE+22a}, but much less is known under general joint preferences. 
This motivates our second research question:
\emph{When fairness and non-wastefulness cannot be simultaneously achieved, how can one property be relaxed while retaining the other, and which relaxations restore existence while remaining economically meaningful?}

Existence and axiomatic compatibility, however, do not by themselves resolve the computational questions raised by these solution concepts. Even when a desirable matching exists, it may be difficult to verify the relevant fairness property, decide whether such a matching exists in a given instance, or compute one efficiently. This motivates our third research question:
\emph{What is the computational complexity of verifying the proposed fairness notions, deciding the existence of matchings satisfying them together with non-wastefulness, and computing such matchings?}

The central theme of this paper is the tension between fairness and non-wastefulness created by the interaction of siblings and initial enrollments. Rather than seeking a single universally desirable solution concept, we study this trade-off from two complementary directions: preserving non-wastefulness while weakening fairness, and preserving fairness while weakening non-wastefulness. We identify which relaxations restore compatibility, develop algorithms that achieve them, and characterize the computational limits of the resulting solution concepts.

We summarize our contributions as follows.

First, we provide a unified treatment of stability, fairness, and
non-wastefulness for daycare matching with siblings and initial
enrollments. Our formulation identifies three distinct ways in which a
family may improve its assignment: claiming vacant seats, replacing
lower-priority children, and using seats released through sibling
withdrawal. By expressing these possibilities through daycare choice
functions, our framework avoids the case-by-case definitions commonly
required in matching with couples and provides a common basis for
comparing the solution concepts studied throughout the paper.

Second, we preserve non-wastefulness while weakening fairness and study
two complementary algorithmic approaches. The \emph{Serial Selection
(SS) Algorithm}, inspired by serial dictatorship, processes families
according to a master list and allows them to take available improvements
without altering the assignments of other families. This naturally
guarantees non-wastefulness, but initial enrollments make conventional
master-list-based fairness too demanding: displacing a lower-ranked
family may deprive it of an assignment protected by its initial
enrollment. This motivates the \emph{Greedy Improvement (GI) Algorithm},
which combines master-list priority with \emph{resettlement}. A
higher-ranked family may displace lower-ranked families, while the
displaced families are returned to their initial assignments whenever
such resettlement is feasible. We show constructively that GI guarantees
non-wastefulness together with fairness by master list and resettlement.
The same principle extends to our stronger dominance-based fairness
notion, under which displaced families can be selected directly according
to daycare priorities.

Third, we consider the opposite direction: preserving fairness while
weakening non-wastefulness. The \emph{Generalized Cutoff (GC) Algorithm}
extends the previous cutoff approach \citep{KaKo24a,ABB24a}: daycare cutoffs are
gradually relaxed, and families choose their most preferred assignments
available under the current cutoffs. Joint preferences among siblings,
however, create interactions across daycare centers that prevent this
natural extension from guaranteeing full fairness; we show that it
guarantees only fairness by replacement. This limitation motivates our
\emph{Fairness-preserving Non-wasteful (FN) Algorithm}. Rather than
restricting improvements through cutoffs, FN permits an improvement
whenever it does not create a prescribed fairness violation. This yields
full fairness together with a new relaxation of non-wastefulness, which
we call \emph{fairness-preserving non-wastefulness}. The underlying
principle is modular: by changing the fairness criterion used to screen
improvements, the same approach can accommodate different notions of
fairness.

% \haris{Give an intuitive overview of the algorithms studied in the paper. Why these algorithms? Can't some ingredients of these algorithms be combined to come up with a new algorithm with better properties?}

% \zhnew{
The four algorithms represent distinct ways of addressing the same
fairness–non-wastefulness tension. GC controls access through
\emph{daycare cutoffs}, FN screens improvements through a
\emph{fairness test}, SS protects existing assignments through
\emph{serial selection}, and GI permits displacement through
\emph{resettlement}. Simply combining these ingredients cannot in
general eliminate the underlying incompatibility: our impossibility
results show that non-wastefulness remains incompatible with several
successively weaker notions of fairness. Resettlement provides a way
around this incompatibility by protecting the initial assignments of
displaced families. Figure~\ref{fig:relations} summarizes the
relationships and compatibility among the resulting solution concepts,
while Table~\ref{tab:comparison} compares the guarantees of the four
algorithms.
% }

Finally, we characterize the computational limits of these solution
concepts. We distinguish among verifying whether a given matching
satisfies a fairness notion, deciding whether a matching satisfying both
fairness and non-wastefulness exists, and computing such a matching.
Fairness, fairness by replacement and withdrawal, fairness by
replacement, and fairness by replacement and master list can be verified
in polynomial time, whereas deciding the existence of a non-wasteful
matching satisfying these fairness requirements is NP-hard. Introducing
resettlement restores universal existence: a matching satisfying
non-wastefulness and fairness by master list and resettlement, or its
stronger dominance-based variant, always exists. This compatibility,
however, comes with a computational cost. Computing such matchings is
PLS-hard. Moreover, detecting a violation of the stronger dominance-based
fairness notion is NP-complete, and hence verifying whether a matching
satisfies this notion is coNP-complete. Thus, restoring existence does
not necessarily restore computational tractability.
Table~\ref{tab:complexity} summarizes these results.

\begin{table}[tb]
\centering
\caption{
Summary of the compatibility of different fairness notions
with non-wastefulness (NW), and the complexity of verifying,
deciding the existence of, and computing a matching satisfying
NW and the corresponding fairness notion.
}
\label{tab:complexity}

\renewcommand{\arraystretch}{1.5}

\begin{tabular}{l|c|c|c|c}
    \hline

\rowcolor{blue!10}\textbf{Fairness notion}
& \makecell{\textbf{Compatibility}\\\textbf{with NW}}
& \textbf{Verification}
& \textbf{Decision}
& \textbf{Computation}
\\ \hline

    Fairness (F)
    & \rxmark
    & P
    & NP-C (Theo.~\ref{theo:decision:nw-fairness})
    & NP-hard (Theo.~\ref{theo:decision:nw-fairness})
    \\ \hline

    \makecell[l]{Replacement and\\Withdrawal (FRW)}
    & \rxmark
    & P
    & NP-C (Theo.~\ref{theo:decision:fr_frw_NW})
    & NP-hard (Theo.~\ref{theo:decision:fr_frw_NW})
    \\ \hline

    Replacement (FR)
    & \rxmark
    & P
    & NP-C (Theo.~\ref{theo:decision:fr_frw_NW})
    & NP-hard (Theo.~\ref{theo:decision:fr_frw_NW})
    \\ \hline

    \makecell[l]{Replacement and\\Master List (FRM)}
    & \rxmark
    & P
    & NP-C (Theo.~\ref{thm:decision:FRM_NW})
    & NP-hard (Theo.~\ref{thm:decision:FRM_NW})
    \\ \hline

    \makecell[l]{Master List and\\Resettlement (FMR)}
    & \gcmark
    & P
    & P
    & PLS-C (Theo.~\ref{theo:computation:FMR})
    \\ \hline

    \makecell[l]{Strong Master List and\\Resettlement (SFMR)}
    & \gcmark
    & coNP-C (Theo.~\ref{theo:verification:SFMR_NW})
    & P
    & PLS-hard (Theo.~\ref{theo:computation:SFMR})
    \\ \hline
\end{tabular}

\end{table}

% \newpage
\section{Literature}

Our work lies at the intersection of matching with couples, daycare
matching with siblings, and matching markets with initial endowments. It
also relates to the literature on master lists, approximate feasibility,
and the computational complexity of stable matching.

\paragraph{Matching with couples.}
Our model is closely related to matching with couples, where couples have
joint preferences over pairs of positions; see \citet{BiKl13a} for a
comprehensive survey. A fundamental difficulty in this setting is that
stable matchings need not exist. Moreover, deciding whether a stable
matching exists is computationally hard. \citet{Ronn90a} established
NP-hardness even when each hospital has only one position.
\citet{McMa10a} showed that hardness persists when couples restrict their
preferences to pairs of positions within the same hospital, while
\citet{BIS11a} obtained hardness results under common rankings of doctors.
Thus, neither simple capacities nor substantial structure on preferences
and priorities is sufficient in general to restore tractability.

Our setting shares the central difficulty of matching with couples:
members of the same family are evaluated jointly rather than independently.
However, families in our model may contain more than two children and,
more importantly, some children may already be enrolled in daycares. These
initial enrollments introduce entitlements that are absent from the
standard matching-with-couples model and fundamentally affect the
interaction between fairness and non-wastefulness.

\paragraph{Restricted preferences and large markets.}
One approach to overcoming the non-existence of stable matchings with
couples is to impose structure on joint preferences. \citet{KlKl05a}
introduced weakly responsive preferences, under which stable matchings are
guaranteed to exist. Roughly, weak responsiveness requires that an
improvement in one member's assignment, holding the other member's
assignment fixed, be regarded as an improvement for the couple.
\citet{HaKo10a} subsequently studied bilateral substitutes in the matching
with contracts framework of \citet{HaMi05a} and showed the connection
between this condition and weak responsiveness. Such restrictions,
however, substantially constrain the complementarity between siblings'
assignments. We instead allow general joint family preferences and study
weaker desiderata when full stability cannot be achieved.

A complementary literature obtains positive results through market
largeness rather than preference restrictions. \citet{KPR13a} show that
stable matchings exist with high probability in sufficiently large markets
when the number of couples grows sufficiently slowly relative to market
size. \citet{ABH14a} substantially extend the range of market sizes for
which such asymptotic existence results hold, and more recent work further
develops this large-market perspective \citep{SYY25a}. These results
provide an important explanation for why stable outcomes may often be
observed despite worst-case non-existence. Our analysis is instead
worst-case and axiomatic: we ask which fairness and non-wastefulness
properties remain mutually compatible for arbitrary instances.

\paragraph{Algorithms for matching with couples and siblings.}
The computational difficulties of matching with couples have motivated a
variety of practical solution methods. The National Resident Matching
Program employs a heuristic related to the incremental algorithm of
\citet{RoVa90a}. \citet{BFI16a} propose an alternative approach based on
Scarf's algorithm \citep{Scar67a}, first computing a fractional solution
whose integrality yields a stable matching.

Integer programming, constraint programming, and satisfiability-based
methods have also been extensively explored
\citep{MOPU07a,BMM14a,DPB15a,MMT17a}. For example, \citet{DPB15a}
formulate the stable matching problem with couples using both SAT and
integer programming, allowing modern general-purpose solvers to search for
stable outcomes. Related optimization approaches have recently been
adapted to daycare matching markets with siblings
\citep{STM+23a,SYT+24a}. These works primarily seek stable outcomes
computationally when they exist. In contrast, our focus is on identifying
alternative normative properties that remain meaningful and achievable
when stability itself is incompatible with the market structure.

The problem has also been studied from the perspective of parameterized
complexity. \citet{MaSc11a} show that deciding whether a stable matching
exists is \textsf{W[1]}-hard when parameterized by the number of couples.
On the positive side, additional structure, including common priority
orderings, can lead to fixed-parameter tractability
\citep{BIS11a}. These results complement our complexity analysis, which
focuses on the difficulty of verifying and computing outcomes satisfying
relaxations of stability rather than stability itself.

\paragraph{Capacity modification and approximate feasibility.}
Another approach to the non-existence of stable matchings is to relax
capacity constraints rather than stability. Building on the fractional
approach based on Scarf's lemma, \citet{NgVo18a} show that every instance
can be transformed into a nearby instance admitting a stable matching by
slightly perturbing hospital capacities. Their results provide bounded
changes both to individual hospital capacities and to aggregate capacity.
More recently, \citet{CMMT24a} develop a polynomial-time algorithm for
computing near-feasible stable matchings in which each hospital's capacity
is violated by only a small amount, under restrictions on couples'
preferences.

Our approach is complementary. We retain the original daycare capacities
and instead relax the normative requirements imposed on the matching.
In particular, we decompose stability into individual rationality,
fairness, and non-wastefulness and investigate how the latter two
properties can be weakened or refined while preserving the original
feasibility constraints. This distinction is particularly relevant in
daycare markets, where physical capacities may be difficult to modify.

\paragraph{Master lists.}
Master lists provide a common ordering of agents and have been used to
impose additional structure on matching markets. \citet{IMS08a} study
stable marriage problems in which preference lists are derived from master
preference lists and characterize the computational consequences of this
structure. Common rankings have also been studied in matching with couples.
In particular, \citet{BIS11a} show that the non-existence and computational
difficulties associated with couples can persist even when hospitals share
a common ranking of doctors.

The master list in our model plays a different role. Rather than generating
the daycare priority orderings, it provides a common ordering of families
that restricts which displacements are regarded as justified. This idea
allows us to weaken fairness in a structured way while retaining
non-wastefulness. We show, however, that a master list alone is insufficient
in the presence of initial enrollments. A lower-ranked family displaced by
a higher-ranked family may lose an assignment to which it was initially
entitled. This observation motivates our notion of \emph{resettlement},
under which displaced families can be returned to their initial
assignments.

\paragraph{Initial enrollments and endowments.}
Initial assignments have a long history in matching and allocation
problems. In house allocation with existing tenants, agents may already
possess houses before participating in the allocation mechanism, and their
initial endowments must therefore be respected
\citep{AbSo99a}. Closely related issues arise in student placement and
school choice, where students may have existing placements or priority
claims \citep{BaSo99a,AbSo03b}.

Subsequent work has studied initial endowments together with increasingly
general distributional constraints. \citet{HHKS+17a} consider
strategy-proof school choice mechanisms in the presence of minimum quotas
and initial endowments. \citet{STY18a} develop efficient allocation
mechanisms that accommodate both endowments and distributional
constraints. \citet{STY+23a} further extend this line of research to
M-convex distributional constraints while retaining strategy-proofness.
Together, these studies demonstrate that initial endowments can be
incorporated into allocation mechanisms even under rich feasibility
constraints.

Our setting introduces a different source of difficulty: families have
joint preferences over the assignments of multiple children. An initial
enrollment therefore cannot always be treated independently at the child
level. A transfer by one family may simultaneously involve several
daycares, while displacement may affect the initial entitlements of
another family. The interaction between joint family preferences and
initial enrollments is precisely what generates several of the
incompatibility results studied in this paper and motivates our
resettlement-based fairness notions.

\paragraph{Local search and computational complexity.}
Finally, our results connect to the literature on local-search complexity.
The complexity class PLS was introduced by \citet{JPY88a} to formalize
search problems in which a locally improving solution can be recognized
efficiently but finding a local optimum may nevertheless be computationally
difficult. This framework is particularly relevant to our greedy
improvement procedures, which repeatedly implement welfare-improving moves
while maintaining the desired matching properties.

Our complexity results reveal a distinction between existence,
verification, and computation. Several of our fairness notions can be
verified in polynomial time even though deciding whether they can coexist
with non-wastefulness is NP-hard. Conversely, our resettlement-based
fairness notions restore universal existence, but stronger variants become
computationally difficult to verify, and computing an outcome that is both
non-wasteful and strongly fair by master list and resettlement is PLS-hard.
Thus, the computational obstacles in our model do not arise solely from
the possible non-existence of desirable outcomes: they persist even for
solution concepts for which existence is guaranteed.

% \haris{Differentiate clearly from [28] "Daycare Matching in Japan: Transfers and Siblings".} \zh{What do you mean? Compare the concepts or complexity results?}

% %-------------------------------------
% % Section: Model
% %-------------------------------------
\section{Model}

An instance of the matching problem with siblings and initial enrollment is given by $I=(C,F,D,Q,\mathcal{X},\omega,\succ_F,\succ_D)$,
where $C$, $F$, and $D$ denote the sets of children, families, and daycares, respectively. Each child $c\in C$ belongs to exactly one
family, denoted by $f_c\in F$. For each family $f\in F$, 
let $C_f=(c_1^f,c_2^f,\ldots,c_{n_f}^f)$
denote its children, listed in a fixed order, where $n_f=|C_f|$. 
Each daycare $d\in D$ is subject to a capacity $Q_d$.

We impose no bound on family size unless explicitly stated otherwise.
All our complexity results other than the local-search hardness results
hold even when each family has at most three children.

Let $\mathcal{X}\subseteq C\times D$ denote the set of available contracts. A contract $x=(c,d)\in\mathcal{X}$ represents the assignment of child $c$ to daycare $d$. 
The initial enrollment is represented by a set of contracts $\omega\subseteq\mathcal{X}$, where each child appears in at most one contract. For convenience, let $\omega(c)=d$ if $(c,d)\in\omega$, and let $\omega(c)=\emptyset$ if child $c$ is initially unassigned.

For any subset $Y\subseteq\mathcal{X}$, let $Y_c$, $Y_d$, and $Y_f$ denote the subsets of contracts involving child $c$, daycare $d$, and family $f$, respectively. Similarly, let $C(Y)$, $D(Y)$, and $F(Y)$ denote the sets of children, daycares, and families appearing in contracts in $Y$. For any set of families $G\subseteq F$, let
\[
Y_G=\bigcup_{g\in G}Y_g
\quad\text{and}\quad
\omega_G=\bigcup_{g\in G}\omega_g.
\]

A \emph{matching} $Y\subseteq\mathcal{X}$ is a set of contracts such that each child is assigned to at most one daycare, 
\[
|Y_c|\leq 1 \qquad \text{for every }c\in C.
\]
A matching $Y$ is \emph{feasible} if no daycare exceeds its capacity:
\[
|Y_d|\leq Q_d \qquad \text{for every }d\in D.
\]
We assume that the initial enrollment $\omega$ is feasible.

For a feasible matching $Y$ and a family $f$, the set $Y_f$ is called the family's
\emph{assignment}. When convenient, we represent an assignment as a tuple according to the fixed ordering of the children in $C_f$, where $\emptyset$ denotes being unmatched. For example, if $C_f=(c_1,c_2)$, the assignment $\{(c_1,d_1),(c_2,d_2)\}$
may be written as $(d_1,d_2)$, while $\{(c_1,d_1)\}$ may be written as $(d_1,\emptyset)$. When used by itself
as a family assignment, $\emptyset$ denotes the empty assignment, in
which all children of the family are unmatched.

Let $\mathcal{Y}$ denote the set of all matchings, and for each family
$f\in F$, define
\[
\mathcal{W}_f=\{Y_f\mid Y\in\mathcal{Y}\}
\]
as the set of possible assignments for family $f$. Each family $f$ has
a strict preference ordering $\succ_f$ over $\mathcal{W}_f$, and let
$\succ_F=(\succ_f)_{f\in F}$ denote the family preference profile.
We assume that each family's preference ordering over its possible assignments is explicitly represented as part of the input. Thus, although
the number of possible assignments may grow with family size, the preference list itself is part of the input, and our complexity bounds are measured with respect to the length of this explicit representation.

Preferences over matchings are induced by preferences over family
assignments. Thus, for two matchings $Y$ and $Z$,
\[
Y\succ_f Z
\quad\Longleftrightarrow\quad
Y_f\succ_f Z_f.
\]
In particular, the initial assignment of family $f$ is $\omega_f$. A
feasible matching $Y$ is acceptable to family $f$ if
$Y_f\succ_f\omega_f$ or $Y_f=\omega_f$.  We will combine this family-side acceptability requirement with daycare-side acceptability when defining individual rationality.

Each daycare $d\in D$ has a strict priority ordering $\succ_d$ over
$\mathcal{X}_d\cup\{\emptyset\}$, where $\emptyset$ denotes the outside
option of leaving a seat unfilled. We use the same symbol $\emptyset$
for these notions of non-assignment; its meaning will be clear from
context. Let
$
\succ_D=(\succ_d)_{d\in D}
$
denote the daycare priority profile. A contract $x\in\mathcal{X}_d$ is acceptable to daycare $d$ if $x\succ_d\emptyset$. We assume without loss of generality that every contract in $\mathcal{X}$ is acceptable to the corresponding daycare, since unacceptable contracts can simply be excluded from $\mathcal{X}$. For notational simplicity, we sometime write $c\succ_d c'$ instead of $(c,d)\succ_d(c',d)$. 

We impose a natural priority condition motivated by daycare assignment systems with initial enrollment: a daycare gives higher priority to children already enrolled there than to children who are not initially enrolled there~\citep{STM+23a,SYT+24a}. This assumption is important for guaranteeing compatibility between initial enrollment and fairness. Without it, individual rationality and fairness may be incompatible; for example, a daycare with a single seat may initially enroll a child who has lower priority than a newly applying child.

\begin{definition}[Priority for Initially Enrolled Children]
\label{assumption}
For every daycare $d\in D$ and any children $c,c'\in C$, if
$\omega(c)=d$ and $\omega(c')\neq d,$
then $c\succ_d c'.$
\end{definition}

% \newpage
%%%%%%%%%%%%%%%%%%%%%%%%%%%%%%%%%%%%%%%%
% % Section : Properties
%%%%%%%%%%%%%%%%%%%%%%%%%%%%%%%%%%%%%%%%
\section{Properties}

In this section, we explore how to generalize the concept of stability by accounting for siblings and initial enrollments, and we demonstrate how to decompose stability into two distinct properties:  fairness and non-wastefulness.

Broadly speaking, a feasible matching is deemed stable if 
no family: (i) is assigned an unacceptable outcome (blocking individual), or (ii) can improve its assignment by reallocating a subset of its children to more preferred daycares which are willing to accept them (blocking coalition) \citep{Roth85a}. Such reallocations can occur in two ways: by occupying unfilled slots or by displacing lower-priority children. Non-wastefulness ensures that no family can improve its allocation by merely occupying vacant slots, while fairness prohibits families from achieving better outcomes by displacing at least one child with lower priority.

The presence of siblings introduces an additional consideration. When a family changes its assignment, a child may vacate a seat at a daycare, and that seat may subsequently be used by another child in the same family. We refer to this possibility as \emph{seat withdrawal}. Depending on whether such withdrawals are taken into account when evaluating a deviation, we obtain two notions of stability.

%-------------------------------------------
% SubSection : Stability and Seat Withdrawal
%-------------------------------------------
\subsection{Seat Withdrawal and Stability}

We begin by introducing the concept of a choice function, which provides a streamlined framework for representing key concepts in matching theory. 
% \citep{HaMi05a}.

\begin{definition}[Choice Function for Daycares] \label{def:choice_func} For a given set of contracts $X \subseteq \mathcal{X}$, the choice function of a daycare $d$, denoted by $\mathrm{Ch}_d(X_d)$, selects contracts from $X_d$ sequentially, in descending order of $\succ_d$, while ensuring that the quota $Q_d$ is not exceeded. This function naturally extends to a set of daycares $D$ by aggregating their individual choices:
\[
\mathrm{Ch}_D(X) = \bigcup_{d \in D} \mathrm{Ch}_d(X_d).
\]
\end{definition}

Individual rationality is a fundamental requirement that prevents blocking by individuals. A feasible matching is \emph{individually rational} if every family is matched to an acceptable assignment and no daycare is assigned an unacceptable contract. We impose individual rationality as a rigid requirement throughout the paper.

% \begin{definition}[Individual Rationality]
% \label{def:IR}
% A matching $Y$ is \emph{individually rational} if i) $\forall f \in F, Y_f \succsim_f \omega_f$, and ii) $\forall d \in D, \forall x \in Y_d, x \succ_d \emptyset$.
% \end{definition}

\begin{definition}[Individual Rationality]
\label{def:IR}
A feasible matching $Y$ is \emph{individually rational} if
\begin{enumerate}
    \item $Y_f\succ_f\omega_f$ or $Y_f = \omega_f$ for every family $f\in F$; and
    \item $x\succ_d\emptyset$ for every daycare $d\in D$ and every
    contract $x\in Y_d$.
\end{enumerate}
\end{definition}

The primary challenge of defining stability lies in properly defining blocking coalitions, which we address next. We first consider deviations without seat withdrawal. In such a deviation, the family's current contracts remain available to the daycares when its proposed contracts are evaluated. Formally, we define weak stability as follows:
\begin{definition}[Weak Stability] \label{def:weak_stability}
% Given an individually rational matching $Y$ and a set of contracts $W \neq Y_f$, a family $f$ and a set of daycares $D(W)$ are said to \emph{strongly block} the matching $Y$ if the following conditions are satisfied:
Let $Y$ be an individually rational matching. A family $f$ and a proposed assignment $W\in\mathcal{W}_f$ with $W\neq Y_f$ form a \emph{strong blocking coalition} against $Y$ if
\begin{enumerate}
\item $W \succ_f Y_f$,
\item $W \subseteq \mathrm{Ch}_D(Y \cup W)$.
\end{enumerate}
% An individually rational matching is said to be \emph{weakly stable} if no such strong blocking coalition exists.
An individually rational matching $Y$ is \emph{weakly stable} if it admits no strong blocking coalition. 
\end{definition}

The concept in Definition~\ref{def:weak_stability} generalizes the classical notion of stability, reducing to it when siblings are not considered. 
The first condition requires that the family strictly prefer the proposed assignment. The second requires every proposed contract to be selected by the corresponding daycare when the proposed contracts are added to the current matching. Importantly, the family's current contracts remain in the choice set. Thus, a child who is currently assigned to a daycare may prevent a sibling from being accepted by that daycare, even if the child prefers to be transferred to a different daycare.
Even under this weaker notion, existence is not guaranteed in general~\citep{Roth84a}.

% \begin{proposition}[\citep{Roth84a}]
%     \label{prop:non-existence:weak-stability}
%     A weakly stable matching may fail to exist exist even if no child is initially enrolled. 
% \end{proposition}
The following example illustrates why allowing a family to withdraw its current contracts before its proposed assignment is evaluated can be important. The idea of seat withdrawal has been considered in recent work \citep{STM+23a, SYT+24a}.

\begin{example}[Instance of Seat Withdrawal]
\label{example:instance}
Consider a family $f$ with two children $C_f=(c_1,c_2)$ and two daycares
$D=\{d_1,d_2\}$, each with capacity one: $Q_{d_1}=Q_{d_2}=1.$
Suppose that $\mathcal{X} =
%\{(c_1,d_1),(c_1,d_2),(c_2,d_1),(c_2,d_2)\}.
\{(c_1,d_1),(c_1,d_2),(c_2,d_2)\}.$
Since all contracts involve children from the same family $f$, we have
$\mathcal{X}=\mathcal{X}_f$.

Suppose the family most prefers assigning $c_1$ to $d_1$ and $c_2$ to
$d_2$, and its next-best relevant assignment assigns only $c_1$ to $d_2$.
For example, its preferences may satisfy
$
\{(c_1,d_1),(c_2,d_2)\}
\succ_f
\{(c_1,d_2)\}
\succ_f
\emptyset,
$
with all other assignments ranked below $\emptyset$.
Each daycare ranks $c_1$ above $c_2$: $
c_1\succ_{d_1}c_2$ and $
c_1\succ_{d_2}c_2.$

Consider current assignment
$Y=\{(c_1,d_2)\}$ and a new assignment $W=\{(c_1,d_1),(c_2,d_2)\}.$
Under Definition~\ref{def:weak_stability}, the family cannot block $Y$ with
its preferred assignment
$W.$
Indeed,
$
\mathrm{Ch}_D(Y\cup W)
=
\{(c_1,d_1),(c_1,d_2)\},
$
because each daycare ranks $c_1$ above $c_2$. Hence,
$
W\not\subseteq\mathrm{Ch}_D(Y\cup W).
$
Nevertheless, this outcome appears undesirable. If $c_1$ first gives up the seat at $d_2$ and moves to $d_1$, the vacated seat at $d_2$ becomes available to $c_2$. The family can therefore obtain
$
W=\{(c_1,d_1),(c_2,d_2)\},
$
its most preferred assignment, without displacing any child outside the family.
\end{example}

%%%%%%%%%%%%%%%%%%%%%%%%%%%%%%%%%%%%%%%%%%%%%%%
% % Figure : seat withdrawal
%%%%%%%%%%%%%%%%%%%%%%%%%%%%%%%%%%%%%%%%%%%%%%%
\begin{figure}[tb]
\centering

\begin{tikzpicture}[
    >=stealth,
    child/.style={
        circle,
        draw=blue!70!black,
        fill=blue!10,
        minimum size=7mm,
        inner sep=0pt
    },
    daycare/.style={
        circle,
        draw=orange!80!black,
        fill=orange!15,
        minimum size=7mm,
        inner sep=0pt
    },
    contractA/.style={
        ->,
        very thick,
        draw=blue!85
    },
    contractB/.style={
        ->,
        very thick,
        draw=green!65!black
    },
    contractC/.style={
        ->,
        very thick,
        draw=red!80
    },
    withdrawn/.style={
        ->,
        very thick,
        dashed,
        draw=blue!85
    },
    note/.style={
        align=center,
        font=\small
    }
]

% ============================================================
% Current matching Y
% ============================================================

\node[font=\bfseries] at (-3.35,1.9)
    {Current assignment};

\node[note] at (-3.35,1.45)
    {\(Y=\{(c_1,d_2)\}\)};

\node[child] (c1L) at (-4.2,0.5) {\(c_1\)};
\node[child] (c2L) at (-4.2,-0.9) {\(c_2\)};

\node[daycare] (d1L) at (-2.5,0.5) {\(d_1\)};
\node[daycare] (d2L) at (-2.5,-0.9) {\(d_2\)};

% Current contract (c_1,d_2)
\draw[contractA] (c1L) -- (d2L);

% ============================================================
% Transition
% ============================================================

\node[
    note,
    text=green!60!black
] (step1) at (-0.35,0.9)
    {\(c_1\) moves to \(d_1\)};

\node[
    note,
    text=blue!85
] (step2) at (-0.35,0)
    {seat at \(d_2\) is released};

\node[
    note,
    text=red!80
] (step3) at (-0.35,-0.9)
    {\(c_2\) takes \(d_2\)};

\draw[->, thick] (step1) -- (step2);
\draw[->, thick] (step2) -- (step3);

% ============================================================
% New assignment W
% ============================================================

\node[font=\bfseries] at (2.4,1.9)
    {New assignment};

\node[note] at (2.4,1.45)
    {\(W=\{(c_1,d_1),(c_2,d_2)\}\)};

\node[child] (c1R) at (1.5,0.5) {\(c_1\)};
\node[child] (c2R) at (1.5,-0.9) {\(c_2\)};

\node[daycare] (d1R) at (3.2,0.5) {\(d_1\)};
\node[daycare] (d2R) at (3.2,-0.9) {\(d_2\)};

% Withdrawn contract (c_1,d_2)
\draw[withdrawn] (c1R) -- (d2R);

% New contract (c_1,d_1)
\draw[contractB] (c1R) -- (d1R);

% New contract (c_2,d_2)
\draw[contractC] (c2R) -- (d2R);

\end{tikzpicture}

\caption{
Illustration of seat withdrawal.
Although \(f\) cannot block \(Y\) with \(W\) under weak stability,
\(c_1\) can first move from \(d_2\) to \(d_1\), thereby releasing the
seat at \(d_2\), which can then be assigned to \(c_2\).
Thus, the family can obtain \(W\) without displacing any child outside
the family.
}
\label{fig:weak-stability-example}
\end{figure}
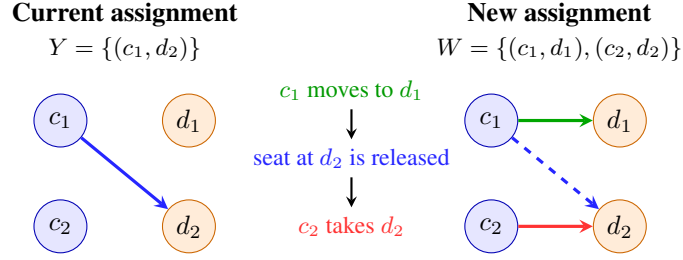

Motivated by this observation, we define stability by allowing a family to withdraw all of its current contracts before its
proposed assignment is evaluated.

% The following stability concept extends weak stability by taking \emph{seat withdrawal} into account.

\begin{definition}[Stability]  
\label{def:stability}  
Let $Y$ be an individually rational matching. A family $f$ and a proposed assignment $W\in\mathcal{W}_f$ with $W\neq Y_f$ form a \emph{blocking coalition} against $Y$ if
% Given an individually rational matching $Y$ and a set of contracts $W \neq Y_f$, a family $f$ and a set of daycares $D(W)$ are said to \emph{block} the matching $Y$ if the following conditions hold:  
\begin{enumerate}  
    \item $W \succ_f Y_f$;
    \item $W \subseteq Ch_D(Y \setminus Y_f \cup W)$.
\end{enumerate}  
An individually rational matching is \emph{stable} if no such blocking coalition exists.  
\end{definition}  

% The subtle difference between Definition~\ref{def:weak_stability} and Definition~\ref{def:stability} lies in the second condition, where, in Definition~\ref{def:stability}, the contracts in $W$ are selected by the relevant daycares when the input is $Y \setminus Y_f \cup W$, which represents the current matching excluding the family's original assignment $Y_f$, with the newly proposed assignment $W$ added. 

The difference between weak stability and stability lies entirely in the set of contracts presented to the daycare choice functions. Under weak stability, the proposed assignment is evaluated together with the family's current assignment. Under stability, the family's current contracts are first withdrawn. Consequently, seats vacated by members of the deviating family may be used by their siblings.

\begin{proposition}
    \label{prop:connection:stability}
    Stability implies weak stability.
\end{proposition}

% \begin{proof}
% Consider any strong blocking coalition between family $f$ and $D(W)$ for a given set of contracts $W$. Since $W \subseteq \mathrm{Ch}_D(Y \cup W)$, it follows that $W \subseteq \mathrm{Ch}_D(Y \setminus Y_f \cup W)$. If the contracts in $Y_f$ are removed first, the selection of $W$ remains unaffected. This is because each daycare selects contracts sequentially based on its priority only. As a result, a strong blocking coalition implies a blocking coalition. Consequently, the absence of blocking coalitions (stability) ensures the absence of strong blocking coalitions (weak stability). This completes the proof of Proposition~\ref{prop:connection:stability}.
% \end{proof}

\begin{proof}
Suppose that family $f$ and assignment $W$ form a strong blocking coalition against $Y$. Then
\[
W\subseteq\mathrm{Ch}_D(Y\cup W).
\]
Removing contracts in $Y_f$ cannot cause a contract in $W$ that was
previously selected to become rejected, because each daycare chooses the highest-priority contracts up to its capacity. Hence,
\[
W\subseteq
\mathrm{Ch}_D\bigl((Y\setminus Y_f)\cup W\bigr).
\]
Since $W\succ_fY_f$, family $f$ and $W$ also form a blocking coalition against $Y$. Therefore, if $Y$ admits no blocking coalition, it admits no strong blocking coalition. Thus, stability implies weak stability.
\end{proof}

%-------------------------------------------
% SubSection : Non-wastefulness and Fairness
%-------------------------------------------
\subsection{Non-wastefulness and Fairness}

In this subsection, we decompose stability in Definition~\ref{def:stability} into two distinct properties: non-wastefulness and fairness. Both properties are defined based on blocking coalitions, but distinguish whether a profitable deviation causes any existing contract to be rejected. Note that weak stability implies neither non-wastefulness nor fairness.

Intuitively, non-wastefulness requires that no family can obtain a strictly preferred assignment by claiming unused seats without altering the assignments of other families. 
We use the term \emph{unused seats} to include not only seats that are vacant under the current matching, but also seats that become available through the withdrawal of contracts belonging to the same family. 
Formally, we define non-wastefulness as follows.

\begin{definition}[Non-wastefulness]
\label{def:nonwastefulness}
Given an individually rational matching $Y$ and a new assignment $W$,
a family $f$ is said to \emph{claim unused seats} if the following
conditions hold:
\begin{description}
    \item[(1) blocking coalition:]
    $W \succ_f Y_f$ and
    $W \subseteq Ch_D(Y \setminus Y_f \cup W)$.
    
    \item[(2) no rejection:]
    % $Y \setminus Y_f \cup W = \mathrm{Ch}_D(Y \setminus Y_f \cup W).$
     \[Y \setminus Y_f \cup W = \mathrm{Ch}_D(Y \setminus Y_f \cup W).\]
\end{description}
An individually rational matching is said to be \emph{non-wasteful}
if no family can claim unused seats.
\end{definition}

Condition~(2) of Definition~\ref{def:nonwastefulness} requires that no
contract be rejected when family $f$ changes its assignment from $Y_f$
to $W$. Thus, the assignments of all other families remain unchanged.
Note that Condition~(2) implies
\[
W \subseteq Ch_D(Y \setminus Y_f \cup W),
\]
and hence implies the daycare-acceptance requirement in Condition~(1). We nevertheless retain this requirement explicitly in Condition~(1)
to maintain a common formulation with stability and the fairness notions
introduced below.

In contrast, fairness concerns profitable deviations that cause at least one existing contract to be rejected. Such a deviation may involve replacing lower-priority children, possibly together with the use of vacant seats or seats withdrawn by children in the same family. The difference from stability in Definition~\ref{def:stability} is therefore the additional requirement that at least one contract be rejected.

\begin{definition}[Fairness]
\label{def:fairness}
Given an individually rational matching $Y$ and a new assignment $W$, a family $f$ is said to have \emph{justified envy} if the following conditions hold:
\begin{description}
    \item[(1) blocking coalition:]
    $W \succ_f Y_f$ and
    $W \subseteq Ch_D(Y \setminus Y_f \cup W)$.
    
    \item[(2) at least one rejection:]
    \[
    |Y \setminus Y_f \cup W|
    >
    |\mathrm{Ch}_D(Y \setminus Y_f \cup W)|.
    \]
\end{description}

Let
$\mathrm{Re} = (Y \setminus Y_f \cup W)
\setminus \mathrm{Ch}_D(Y \setminus Y_f \cup W)$
denote the set of rejected contracts. Family $f$ is said to have \emph{justified envy toward} the families in $F(\mathrm{Re})$, i.e., the families whose contracts are rejected as a result of the deviation.

An individually rational matching is said to be \emph{fair} if no family has justified envy.
\end{definition}

Non-wastefulness and fairness partition blocking coalitions according to whether the deviation causes a rejection. If no contract is rejected, the deviating family can claim unused seats; if at least one contract is rejected, the family has justified envy. This immediately gives the following relationship.

\begin{proposition}
\label{prop:stability:equivalence}
An individually rational matching is stable if and only if it is both non-wasteful and fair.
\end{proposition}

\begin{proof}
Suppose first that $Y$ is stable. Then no family forms a blocking coalition against $Y$. Since both a wasteful improvement and justified envy require a blocking coalition, $Y$ is both non-wasteful and fair.

Conversely, suppose that $Y$ is both non-wasteful and fair. If $Y$ were not stable, there would exist a family $f$ and an assignment $W$ forming a blocking coalition against $Y$. Consider $Y \setminus Y_f \cup W.$
Either no contract is rejected by the daycare choice functions, in which case
\[
Y \setminus Y_f \cup W
=
\mathrm{Ch}_D(Y \setminus Y_f \cup W),
\]
and $f$ can wastefully improve, contradicting non-wastefulness; or at least one contract is rejected, in which case
\[
|Y \setminus Y_f \cup W|
>
|\mathrm{Ch}_D(Y \setminus Y_f \cup W)|,
\]
and $f$ has justified envy, contradicting fairness. Hence, no blocking coalition exists, and $Y$ is stable.
\end{proof}

%-------------------------------------------
% SubsubSection : Computational Complexity of F and NW
%-------------------------------------------
\subsubsection{Computational Complexity of Problems Related to F}
\label{subsection:complexity:F}
We next examine the computational complexity of non-wastefulness and
fairness. While verifying whether a given matching satisfies both
properties is computationally tractable, deciding whether such a matching
exists is NP-complete.

\begin{theorem}
\label{theo:verification:nw-fairness}
Given a feasible matching $Y$, whether $Y$ satisfies both
non-wastefulness and fairness can be verified in polynomial time.
\end{theorem}
\begin{proof}
We first check whether $Y$ is individually rational. For each family
$f\in F$, we verify that $Y_f$ is weakly preferred to its initial
assignment $\omega_f$, and for each contract $(c,d)\in Y$, we verify
that it is acceptable to daycare $d$. Since family preferences and
daycare priorities are explicitly represented, these conditions can be
checked in polynomial time. If $Y$ is not individually rational, it does
not satisfy the two properties.
Suppose henceforth that $Y$ is individually rational. By
Proposition~\ref{prop:stability:equivalence}, $Y$ is both non-wasteful
and fair if and only if it is stable. It therefore suffices to determine
whether $Y$ admits a blocking coalition.
For each family $f\in F$, we examine every assignment
$W\succ_f Y_f$ in its preference ordering. For each such $W$, let
\[
Z=(Y\setminus Y_f)\cup W.
\]
We compute $Ch_D(Z)$ by applying the daycare choice function at each
daycare. By Definition~\ref{def:stability}, family $f$ and assignment
$W$ form a blocking coalition against $Y$ if and only if
\[
W\subseteq Ch_D(Z).
\]
Since the daycare priority orders are explicitly represented,
$Ch_D(Z)$ can be computed in polynomial time. Moreover, family
preference orderings are explicitly represented, so the total number of
assignments examined is polynomial in the input size. Hence, we can
determine in polynomial time whether any blocking coalition exists.
Therefore, whether $Y$ satisfies both non-wastefulness and fairness can
be verified in polynomial time.
\end{proof}

\begin{theorem}
\label{theo:decision:nw-fairness}
Given an instance, deciding whether there exists a matching that is both
non-wasteful and fair is NP-complete. Consequently, computing such a
matching, while correctly reporting when none exists, is NP-hard.
\end{theorem}

\begin{proof}
We first establish NP-hardness by a reduction from the
Hospitals/Residents problem with Couples (HRC). An instance of HRC
consists of a set of residents and a set of hospitals. Some residents
form couples and express joint preferences over pairs of hospitals,
while each hospital has a capacity and a strict priority ordering over
residents. A couple is either jointly unassigned or jointly assigned to
an acceptable pair of hospitals.

The stability notion used in the HRC literature is the MM-stability
notion introduced by \citet{McMa10a} and subsequently employed in
\citet{BMM14a,MMT17a,CMMT24a}. For completeness, we recall the
definition below.

\begin{definition}[MM-stability]
\label{def:stability:McMa10a}
A matching $M$ is \emph{MM-unstable} if at least one of the following
holds:
\begin{enumerate}
    \item A single resident $r_i$ and a hospital $h_j$ block $M$ as in
    the classical Hospitals/Residents problem.

    \item A coupled resident $r_i$, whose partner is $r_j$, and a
    hospital $h_k$ block $M$: the couple $(r_i,r_j)$ prefers $(h_k,M(r_j))$
    to $(M(r_i),M(r_j))$,
    and $h_k$ is either undersubscribed in $M$ or prefers $r_i$ to some
    member of $M(h_k)\setminus\{r_j\}$.

    \item A couple $(r_i,r_j)$ and hospitals $h_k$ and $h_l$, where
    $h_k\neq M(r_i)$ and $h_l\neq M(r_j)$, block $M$: the couple
    prefers $(h_k,h_l)$ to $(M(r_i),M(r_j))$, and one of the following
    holds:
    \begin{enumerate}
        \item $h_k\neq h_l$, and $h_k$ is either undersubscribed in $M$
        or prefers $r_i$ to some member of $M(h_k)$, while $h_l$ is
        either undersubscribed in $M$ or prefers $r_j$ to some member
        of $M(h_l)$;

        \item $h_k=h_l$, and $h_k$ has at least two free posts, i.e.,
        \[
            c_k-|M(h_k)|\geq 2;
        \]

        \item $h_k=h_l$, $h_k$ has exactly one free post, i.e.,
        \[
            c_k-|M(h_k)|=1,
        \]
        and $h_k$ prefers at least one of $r_i,r_j$ to some member of
        $M(h_k)$;

        \item $h_k=h_l$, $h_k$ is full, $h_k$ prefers $r_i$ to some
        $r_s\in M(h_k)$, and $h_k$ prefers $r_j$ to some
        \[
            r_t\in M(h_k)\setminus\{r_s\}.
        \]
    \end{enumerate}
\end{enumerate}
A matching is \emph{MM-stable} if it is not MM-unstable.
\end{definition}

For positive integers $\alpha$ and $\beta$, $(\alpha,\beta)$-HRC denotes
the restriction of HRC in which each single resident finds at most
$\alpha$ hospitals acceptable, each couple finds at most $\alpha$ pairs
of hospitals acceptable, and each hospital finds at most $\beta$
residents acceptable. Deciding whether an MM-stable matching exists in
$(2,2)$-HRC is NP-complete, even when there are no single residents and
every hospital has capacity one~\citep{BMM14a}.

Given such an instance, we construct an instance of our model by
interpreting each couple as a family with two children and each hospital
as a daycare with the same capacity and priority ordering. The joint
preference ordering of each couple over acceptable hospital pairs becomes
the preference ordering of the corresponding family over the associated
assignments. All assignments corresponding to acceptable hospital pairs
are ranked above $\emptyset$, while all other assignments are ranked
below $\emptyset$. There is no initial enrollment. Clearly, this
construction can be performed in polynomial time.

We next establish the correspondence between MM-stability and our
stability notion. We first show that MM-stability and weak stability
coincide on the restricted instances used in the NP-hardness construction
of \citet{BMM14a}. Besides containing no single residents and having
unit-capacity hospitals, these instances have the property that whenever
a couple has two acceptable hospital pairs, the two pairs involve
disjoint hospitals.

Consider a couple corresponding to a family $f$ and a strictly preferred
acceptable assignment $W$. Case~(1) of
Definition~\ref{def:stability:McMa10a} cannot arise because the
restricted instances contain no single residents.

Case~(2) cannot arise either. If the couple is currently matched, then a
strictly preferred acceptable pair cannot retain one member at the same
hospital while moving only the other member, because any two acceptable
pairs of the couple involve disjoint hospitals. If the couple is
currently unmatched, our restricted HRC instances allow a couple to be
either jointly unassigned or jointly assigned to an acceptable pair of
hospitals, so a deviation assigning only one member is not an acceptable
alternative.

It remains to consider Case~(3). Since every hospital has capacity one,
Cases~(3b)--(3d), in which both members apply to the same hospital,
cannot yield a blocking coalition. Indeed, a unit-capacity hospital
cannot have two free posts; if it has one free post, it has no currently
assigned resident to replace; and if it is full, it cannot replace two
distinct currently assigned residents. Hence only Case~(3a) is relevant.

Let the proposed assignment be
\[
W=\{(r_i,h_k),(r_j,h_l)\},
\qquad h_k\neq h_l.
\]
Under Case~(3a) of MM-stability, the couple blocks $M$ precisely when it
strictly prefers $W$ to its current assignment and, at each proposed
hospital, the hospital is either vacant or prefers the corresponding
incoming resident to its current assignee.

Since every hospital has capacity one, the latter condition is equivalent
to requiring both proposed contracts to be selected when they are added
to the current matching. Under the correspondence between residents and
children and between hospitals and daycares, this is exactly
\[
W\subseteq Ch_D(Y\cup W).
\]
Therefore, a couple forms an MM-blocking coalition if and only if the
corresponding family forms a strong blocking coalition in the sense of
Definition~\ref{def:weak_stability}. Hence,
\[
\text{MM-stability}
\quad\Longleftrightarrow\quad
\text{weak stability}
\]
on the instances produced by the reduction.

We next show that weak stability and stability also coincide on these
instances. Consider a family $f$ and an acceptable assignment
$W\succ_fY_f$. If $Y_f=\emptyset$, then removing $Y_f$ clearly has no
effect. Otherwise, $Y_f$ and $W$ correspond to two distinct acceptable
hospital pairs of the same couple. By the disjointness property,
\[
D(Y_f)\cap D(W)=\emptyset.
\]
Thus, removing the family's current assignment does not change the set
of contracts considered at any daycare involved in $W$. Since daycare
choice is evaluated independently at each daycare,
\[
W\subseteq Ch_D(Y\cup W)
\quad\Longleftrightarrow\quad
W\subseteq Ch_D((Y\setminus Y_f)\cup W).
\]
Therefore,
\[
\text{weak stability}
\quad\Longleftrightarrow\quad
\text{stability}
\]
on the instances produced by the reduction.

Combining the two equivalences, the restricted HRC instance admits an
MM-stable matching if and only if the corresponding instance of our
model admits a stable matching. By
Proposition~\ref{prop:stability:equivalence}, an individually rational
matching is stable if and only if it is both non-wasteful and fair.
Hence, the original HRC instance admits an MM-stable matching if and
only if the constructed instance admits an individually rational
matching that is both non-wasteful and fair. This proves NP-hardness.

It remains to show membership in NP. A matching $Y$ serves as a
polynomial-size certificate. Feasibility and individual rationality can
be checked in polynomial time, and by
Theorem~\ref{theo:verification:nw-fairness}, whether $Y$ is both
non-wasteful and fair can also be verified in polynomial time.
Therefore, the existence problem belongs to NP and is thus NP-complete.

Finally, if a polynomial-time algorithm could compute a matching that is
both non-wasteful and fair, while correctly reporting when none exists,
then it would also decide the NP-complete existence problem in polynomial
time. Hence, the corresponding computation problem is NP-hard.
\end{proof}

% \newpage
%%%%%%%%%%%%%%%%%%%%%%%%%%%%%%%%%%%%%%%%
% % Section : Weakening Fairness
%%%%%%%%%%%%%%%%%%%%%%%%%%%%%%%%%%%%%%%%
\section{Weakening Fairness}

It is well known that weakly stable matchings may fail to exist in matching
markets with couples~\citep{Roth84a}. Since stability is stronger than weak
stability in our model, a stable matching may also fail to exist. This
motivates us to investigate whether one of the two components of stability,
fairness and non-wastefulness, can be preserved while the other is relaxed.
In this section, we focus on weakening fairness while retaining
non-wastefulness.

%-------------------------------------------
% SubSection : Replacement and Withdrawal
%-------------------------------------------
\subsection{Replacement and Withdrawal}

\begin{figure*}[t]
\centering

\resizebox{0.99\textwidth}{!}{%
\begin{tikzpicture}[
    >=Latex,
    panel/.style={
        draw,
        rounded corners=3pt,
        minimum width=5.0cm,
        minimum height=3.2cm,
        align=center,
        inner sep=7pt,
        line width=0.7pt
    },
    flow/.style={
        ->,
        line width=0.8pt
    }
]

% ------------------------------------------------
% Fairness (F)
% ------------------------------------------------
\node[
    panel,
    fill=blue!7,
    draw=blue!55!black
] (F) at (0,0) {
    {\bfseries Fairness (F)}\\[2mm]
    A new seat may come from\\[2mm]
    $\left\{
    \begin{array}{l}
    \text{a previously vacant seat,}\\[1mm]
    \text{the withdrawal of a sibling,}\\[1mm]
    \text{the replacement of another child.}
    \end{array}
    \right.$
};

% ------------------------------------------------
% FRW
% ------------------------------------------------
\node[
    panel,
    fill=yellow!10,
    draw=yellow!55!black
] (FRW) at (6.0,0) {
    {\bfseries FRW}\\[2mm]
    A new seat must come from\\[2mm]
    $\left\{
    \begin{array}{l}
    \text{\color{gray}\sout{a previously vacant seat,}}\\[1mm]
    \text{the withdrawal of a sibling,}\\[1mm]
    \text{the replacement of another child.}
    \end{array}
    \right.$
};

% ------------------------------------------------
% FR
% ------------------------------------------------
\node[
    panel,
    fill=red!6,
    draw=red!60!black
] (FR) at (12.0,0) {
    {\bfseries FR}\\[2mm]
    A new seat must come from\\[2mm]
    $\left\{
    \begin{array}{l}
    \text{\color{gray}\sout{a previously vacant seat,}}\\[1mm]
    \text{\color{gray}\sout{the withdrawal of a sibling,}}\\[1mm]
    \text{the replacement of another child.}
    \end{array}
    \right.$
};

\end{tikzpicture}%
}

\caption{
Comparison of fairness (F), fairness by replacement and withdrawal
(FRW), and fairness by replacement (FR).
Fairness allows a profitable deviation to use previously vacant seats,
seats released through sibling withdrawal, or seats obtained through
replacement.
FRW excludes deviations that rely on previously vacant seats, while FR
further excludes deviations that rely on sibling withdrawal.
}
\label{fig:F-FRW-FR}
\end{figure*}

Recall that a family may improve its assignment in three ways:
(i) by occupying previously unfilled seats,
(ii) by replacing lower-priority children, and
(iii) by using seats withdrawn by children in the same family.

The fairness notion in Definition~\ref{def:fairness} may be considered demanding because a deviation can constitute justified envy even when some of the newly obtained seats are vacant or are made available by seat withdrawal, as long as the deviation causes at least one existing contract
to be rejected. Such seats do not require the displacement of another family.

Motivated by this observation, we introduce two weaker notions of fairness. The first excludes improvements that rely on previously vacant seats, while still allowing seats released through seat withdrawal. The second is more restrictive and requires every newly obtained seat to arise through the replacement of a contract belonging to another family.

We first introduce fairness by replacement and withdrawal. Intuitively, whenever a family proposes a new assignment, the corresponding seat must already be occupied under the current assignment: it must become available either because a contract of another family is rejected or because a child of the deviating family withdraws from that daycare. Thus, improvements that rely on previously vacant seats are excluded.

\begin{definition}[Fairness by Replacement and Withdrawal (FRW)]
\label{def:fairness_withdraw}
Given an individually rational matching $Y$ and an alternative assignment
$W$, a family $f$ is said to have \emph{justified envy by replacement and
withdrawal} if the following conditions hold:
\begin{description}
    \item[(1) blocking coalition:]
    $W \succ_f Y_f$ and
    $W \subseteq Ch_D(Y \setminus Y_f \cup W)$.

    \item[(2) at least one rejection:]
    $|Y \setminus Y_f \cup W| >
    |\mathrm{Ch}_D(Y \setminus Y_f \cup W)|.$

    \item[(3) no use of vacant seats:]
    Let $\mathrm{Re} = (Y\setminus Y_f \cup W)
    \setminus Ch_D(Y\setminus Y_f\cup W).$
    For each contract $(c,d)\in W\setminus Y_f$, there exists a distinct contract $(c',d)$ such that
    $(c',d) \in \mathrm{Re}\cup(Y_f\setminus W).$
\end{description}
An individually rational matching is said to be
\emph{fair by replacement and withdrawal} if no family has justified envy
by replacement and withdrawal.
\end{definition}

The only difference between Definition~\ref{def:fairness} and
Definition~\ref{def:fairness_withdraw} is Condition~(3). For every newly
proposed contract $(c,d)\in W\setminus Y_f$, there must be a distinct
contract at the same daycare that either (i) is rejected as a result of the
deviation, or (ii) belongs to the deviating family's current assignment and
is withdrawn under $W$. Hence, every newly obtained seat is supported by
either replacement or seat withdrawal, and no previously vacant seat can
be used.

We next introduce an even weaker fairness notion by excluding sibling
withdrawal as well. Under fairness by replacement, every newly obtained
contract must be supported by a contract of another family that is actually
rejected as a result of the deviation.

\begin{definition}[Fairness by Replacement (FR)]
\label{def:fairness_replacement}
Given an individually rational matching $Y$ and a new assignment $W$, a
family $f$ is said to have \emph{justified envy by replacement} if the
following conditions hold:
\begin{description}
    \item[(1) blocking coalition:]
    $W \succ_f Y_f$ and
    $W \subseteq Ch_D(Y \setminus Y_f \cup W)$.

    \item[(2) at least one rejection:]
    $|Y \setminus Y_f \cup W| >
    |\mathrm{Ch}_D(Y \setminus Y_f \cup W)|.$

    \item[(3) replacement only:]
    For each $(c,d)\in W\setminus Y_f$, there exists a distinct contract
    $(c',d)$ such that
    \[
    (c',d)\in Y\setminus Y_f
    \quad\text{and}\quad
    (c',d)\notin Ch_D(Y\setminus Y_f\cup W).
    \]
\end{description}
An individually rational matching is said to be \emph{fair by replacement}
if no family has justified envy by replacement.
\end{definition}

The three fairness notions are ordered by implication.

\begin{proposition}
\label{prop:connection:F&FRW&FR}
Fairness implies fairness by replacement and withdrawal, and fairness by
replacement and withdrawal implies fairness by replacement.
\end{proposition}

\begin{proof}
Any justified envy by replacement and withdrawal satisfies Conditions~(1)
and~(2) of Definition~\ref{def:fairness}, and therefore constitutes
justified envy. Hence, if a matching is fair, it is also fair by replacement
and withdrawal.

Similarly, any justified envy by replacement satisfies Conditions~(1)
and~(2) of Definition~\ref{def:fairness_withdraw}. Moreover,
Condition~(3) of Definition~\ref{def:fairness_replacement} implies
Condition~(3) of Definition~\ref{def:fairness_withdraw}: every newly
proposed contract is supported by a rejected contract from another family,
which is one of the two types of supporting contracts permitted under
replacement and withdrawal. Therefore, justified envy by replacement
implies justified envy by replacement and withdrawal. Consequently,
fairness by replacement and withdrawal implies fairness by replacement.
\end{proof}

Although fairness by replacement is substantially weaker than fairness, it
is still incompatible with non-wastefulness, even in the absence of initial
enrollment.

\begin{theorem}
\label{theo:non_existence:FR&NW}
A matching that satisfies both non-wastefulness and fairness by replacement
may not exist, even when no child is initially enrolled.
\end{theorem}

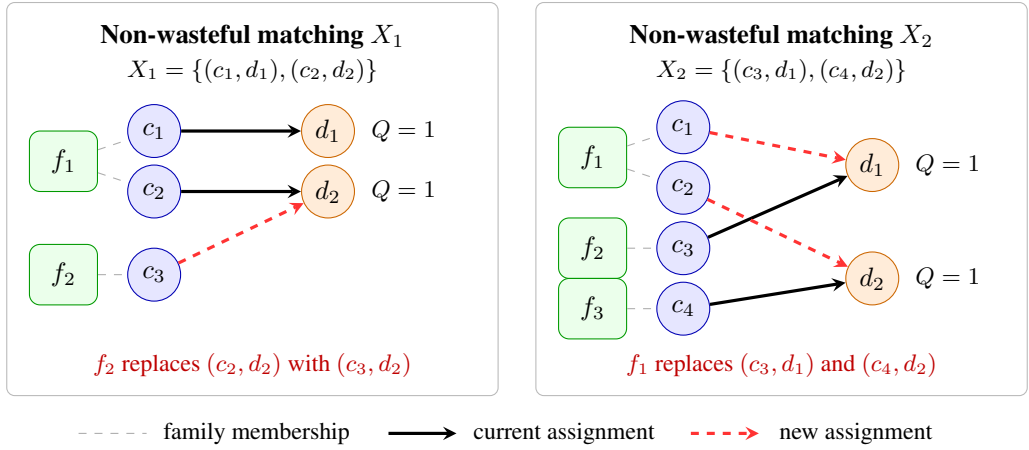
\begin{figure}[tb]
\centering

\begin{tikzpicture}[
    >=stealth,
    child/.style={
        circle,
        draw=blue!70!black,
        fill=blue!10,
        minimum size=7mm,
        inner sep=0pt
    },
    daycare/.style={
        circle,
        draw=orange!80!black,
        fill=orange!15,
        minimum size=7mm,
        inner sep=0pt
    },
    family/.style={
        rounded corners,
        draw=green!60!black,
        fill=green!8,
        minimum width=9mm,
        minimum height=8mm,
        align=center
    },
    member/.style={
        -,
        thin,
        dashed,
        draw=gray!65
    },
    assign/.style={
        ->,
        very thick,
        draw=black
    },
    replace/.style={
        ->,
        very thick,
        dashed,
        draw=red!80
    },
    panel/.style={
        draw=gray!55,
        rounded corners=3pt
    },
    note/.style={
        align=center,
        font=\small
    }
]

% ============================================================
% Matching X1
% ============================================================

\node[
    panel,
    minimum width=6.5cm,
    minimum height=5.2cm
] (X1box) at (-3.5,-1.7) {};

\node[
    font=\bfseries
] at (-3.5,0.45)
    {Non-wasteful matching \(X_1\)};

\node[note] at (-3.5,0.0)
    {$X_1=\{(c_1,d_1),(c_2,d_2)\}$};

% ------------------------------------------------------------
% Families and children
% ------------------------------------------------------------

% f1 has children c1 and c2
\node[family] (f1L) at (-6.0,-1.2) {$f_1$};

\node[child] (c1L) at (-4.8,-0.8) {$c_1$};
\node[child] (c2L) at (-4.8,-1.6) {$c_2$};

% f2 has child c3
\node[family] (f2L) at (-6.0,-2.7) {$f_2$};
\node[child] (c3L) at (-4.8,-2.7) {$c_3$};

% Membership
\draw[member] (f1L) -- (c1L);
\draw[member] (f1L) -- (c2L);
\draw[member] (f2L) -- (c3L);

% ------------------------------------------------------------
% Daycares
% ------------------------------------------------------------

\node[daycare] (d1L) at (-2.5,-0.8) {$d_1$};
\node[daycare] (d2L) at (-2.5,-1.6) {$d_2$};

\node[font=\small] at (-1.5,-0.8)
    {\(Q=1\)};
\node[font=\small] at (-1.5,-1.6)
    {\(Q=1\)};

% ------------------------------------------------------------
% Current matching
% ------------------------------------------------------------

\draw[assign] (c1L) -- (d1L);
\draw[assign] (c2L) -- (d2L);

% ------------------------------------------------------------
% Replacement
% ------------------------------------------------------------

\draw[replace] (c3L) -- (d2L);

\node[
    note,
    text=red!75!black
] at (-3.5,-3.9)
    {$f_2$ replaces $(c_2,d_2)$ with $(c_3,d_2)$};

% ============================================================
% Matching X2
% ============================================================

\node[
    panel,
    minimum width=6.5cm,
    minimum height=5.2cm
] (X2box) at (3.5,-1.7) {};

\node[
    font=\bfseries
] at (3.5,0.45)
    {Non-wasteful matching \(X_2\)};

\node[note] at (3.5,0.0)
    {$X_2=\{(c_3,d_1),(c_4,d_2)\}$};

% ------------------------------------------------------------
% Families and children
% ------------------------------------------------------------

% f1 has children c1 and c2
\node[family] (f1R) at (1.0,-1.15) {$f_1$};

\node[child] (c1R) at (2.2,-0.75) {$c_1$};
\node[child] (c2R) at (2.2,-1.55) {$c_2$};

% f2 has child c3
\node[family] (f2R) at (1.0,-2.35) {$f_2$};
\node[child] (c3R) at (2.2,-2.35) {$c_3$};

% f3 has child c4
\node[family] (f3R) at (1.0,-3.15) {$f_3$};
\node[child] (c4R) at (2.2,-3.15) {$c_4$};

% Membership
\draw[member] (f1R) -- (c1R);
\draw[member] (f1R) -- (c2R);
\draw[member] (f2R) -- (c3R);
\draw[member] (f3R) -- (c4R);

% ------------------------------------------------------------
% Daycares
% ------------------------------------------------------------

\node[daycare] (d1R) at (4.7,-1.25) {$d_1$};
\node[daycare] (d2R) at (4.7,-2.75) {$d_2$};

\node[font=\small] at (5.7,-1.25)
    {\(Q=1\)};
\node[font=\small] at (5.7,-2.75)
    {\(Q=1\)};

% ------------------------------------------------------------
% Current matching
% ------------------------------------------------------------

\draw[assign] (c3R) -- (d1R);
\draw[assign] (c4R) -- (d2R);

% ------------------------------------------------------------
% Replacements by f1
% ------------------------------------------------------------

\draw[replace] (c1R) -- (d1R);
\draw[replace] (c2R) -- (d2R);

\node[
    note,
    text=red!75!black
] at (3.5,-3.9)
    {$f_1$ replaces $(c_3,d_1)$ and $(c_4,d_2)$};

% ============================================================
% Legend
% ============================================================

\draw[member] (-5.8,-4.8) -- (-4.9,-4.8);
\node[anchor=west, font=\small] at (-4.8,-4.8)
    {family membership};

\draw[assign] (-1.7,-4.8) -- (-0.8,-4.8);
\node[anchor=west, font=\small] at (-0.7,-4.8)
    {current assignment};

\draw[replace] (2.3,-4.8) -- (3.2,-4.8);
\node[anchor=west, font=\small] at (3.35,-4.8)
    {new assignment};

\end{tikzpicture}

\caption{
The two non-wasteful matchings in the instance of
Theorem~\ref{theo:non_existence:FR&NW}.
In \(X_1\), \(f_2\) has justified envy by replacing \(c_2\) with \(c_3\)
at \(d_2\), while in \(X_2\), \(f_1\) can replace \(c_3\) and \(c_4\)
with \(c_1\) and \(c_2\), respectively.
}
\label{fig:FR-NW-nonexistence}
\end{figure}

\begin{proof}
Consider three families. Family $f_1$ has two children
$C_{f_1}= (c_1,c_2)$, while families $f_2$ and $f_3$ have one child each,
$C_{f_2}=(c_3)$ and $C_{f_3}=(c_4)$. There are two daycares,
$d_1$ and $d_2$, each with capacity one. The acceptable assignments and
priority orders are as follows:
\begin{align*}
&\succ_{f_1}: \{(c_1,d_1),(c_2,d_2)\}, && \succ_{d_1}:
c_1\succ_{d_1}c_3, 
\\
&\succ_{f_2}: (c_3,d_1)\succ_{f_2}(c_3,d_2), && \succ_{d_2}:
c_3\succ_{d_2}c_2\succ_{d_2}c_4, 
\\
&\succ_{f_3}:
(c_4,d_2).
% \\[1mm]
% &\succ_{d_1}:
% c_1\succ_{d_1}c_3,\\
% &\succ_{d_2}:
% c_3\succ_{d_2}c_2\succ_{d_2}c_4.
\end{align*}
All assignments not listed above are unacceptable to the corresponding
family, and no child is initially enrolled.

We first identify the non-wasteful matchings. If $f_1$ is assigned its
acceptable assignment, feasibility uniquely gives
\[
X_1=\{(c_1,d_1),(c_2,d_2)\}.
\]
Otherwise, $f_1$ is unmatched. If $c_3$ is assigned to $d_2$, then $d_1$
is vacant and $f_2$ can improve by moving $c_3$ from $d_2$ to $d_1$,
so the matching is wasteful. If $c_3$ is assigned to $d_1$ while $d_2$
is vacant, then $f_3$ can improve by assigning $c_4$ to $d_2$, so the
matching is again wasteful. Consequently, the only non-wasteful matching
in this case is
\[
X_2=\{(c_3,d_1),(c_4,d_2)\}.
\]
Thus, $X_1$ and $X_2$ are the only non-wasteful matchings.

Consider first $X_1$. Family $f_2$ has justified envy by replacement through
\[
W=\{(c_3,d_2)\}.
\]
Family $f_2$ prefers $d_2$ to being unmatched, and $c_3$ has higher
priority than $c_2$ at $d_2$. Hence, $(c_3,d_2)$ replaces
$(c_2,d_2)$, satisfying all conditions of
Definition~\ref{def:fairness_replacement}. Therefore, $X_1$ is not fair by
replacement.

Now consider $X_2$. Family $f_1$ has justified envy by replacement through
\[
W=\{(c_1,d_1),(c_2,d_2)\}.
\]
At $d_1$, child $c_1$ has higher priority than $c_3$, while at $d_2$,
child $c_2$ has higher priority than $c_4$. Thus, the two proposed
contracts replace $(c_3,d_1)$ and $(c_4,d_2)$, respectively. Hence,
$X_2$ is also not fair by replacement.

Therefore, every non-wasteful matching violates fairness by replacement.
A matching satisfying both properties need not exist.
\end{proof}

%-------------------------------------------
% SubsubSection : Computational Complexity of FR and FRM
%-------------------------------------------
\subsubsection{Computational Complexity of Problems Related to FRM}
\label{section:complexity_FRW}

We next examine the computational complexity of fairness by replacement and fairness by replacement and withdrawal. While violations of either fairness notion can be detected in polynomial time, deciding whether there exists a non-wasteful matching satisfying either notion is NP-complete, and computing such a matching is NP-hard.

\begin{theorem}
\label{theo:check:envy:replacement}
Given an individually rational matching $Y$, whether $Y$ satisfies
fairness by replacement or fairness by replacement and withdrawal can be
verified in time $ O\left(|\mathcal{X}| \sum_{f\in F}|\mathcal{W}_f|
\right),$ where $|\mathcal{W}_f|$ denotes the length of family $f$'s preference
ordering.
\end{theorem}

\begin{proof}
We verify the absence of justified envy for each family $f\in F$.
Fix a family $f$. We examine each assignment $W\succ_f Y_f$ in its preference ordering. Since the preference ordering
is explicitly represented, at most $|\mathcal W_f|$ assignments need to
be examined.

Fix an assignment $W$ and let $Z=(Y\setminus Y_f)\cup W.
$ We show that all conditions for justified envy by replacement and justified
envy by replacement and withdrawal can be checked in
$O(|\mathcal{X}|)$ time.

First, we compute the contracts selected by the daycare choice functions, $Ch_D(Z)=\bigcup_{d\in D}Ch_d(Z_d).$
For each daycare $d$, we scan its priority list and select, among the contracts contained in $Z_d$, the highest-priority contracts up to its capacity $Q_d$. Since the daycare priority lists are explicitly represented, scanning all priority lists takes
$O\left(\sum_{d\in D}|\mathcal{X}_d|\right)
= O(|\mathcal{X}|)$ time.
We then define the set of rejected contracts as
$\mathrm{Re} = Z\setminus Ch_D(Z).$
Condition~(1) can be checked by verifying
$W\subseteq Ch_D(Z),$ while Condition~(2) can be checked by verifying $\mathrm{Re}\neq\emptyset.$
Thus, Conditions~(1) and~(2) can be checked within the same
$O(|\mathcal{X}|)$ computation.

It remains to check Condition~(3). For each daycare $d$, let
$N_d = (W\setminus Y_f)_d$ denote the set of newly proposed contracts at $d$.
For justified envy by replacement, let
$R_d = \mathrm{Re}_d\cap(Y\setminus Y_f)$
denote the rejected contracts at $d$ belonging to families other than $f$. Condition~(3) requires that every contract in $N_d$ correspond to a distinct rejected contract in $R_d$. Since both the new and rejected contracts must involve the same daycare, such a correspondence exists if and only if $|N_d|\leq |R_d|$ for every $d\in D.$
These cardinalities can be computed by scanning the relevant contracts,
and hence Condition~(3) can be checked in $O(|\mathcal{X}|)$ time.

For justified envy by replacement and withdrawal, a newly proposed
contract may instead use a seat released by another child of the same family. Let $H_d = (Y_f\setminus W)_d$
denote the contracts withdrawn by family $f$ from daycare $d$.
Condition~(3) is therefore satisfied if and only if
$|N_d| \leq |R_d|+|H_d|$ for every $d\in D.$
Again, all relevant cardinalities can be computed in
$O(|\mathcal{X}|)$ time.

Therefore, for each assignment $W$, we can determine whether it witnesses
justified envy by replacement or justified envy by replacement and
withdrawal in $O(|\mathcal{X}|)$ time. Since at most
$|\mathcal{W}_f|$ assignments are examined, the total running time is $O\bigl(|\mathcal{X}|\cdot|\mathcal{W}_f|\bigr).$
Hence, both problems can be checked in polynomial time.
\end{proof}

\begin{theorem}
\label{theo:decision:fr_frw_NW}
Deciding whether there exists an individually rational matching that is
both non-wasteful and fair by replacement is NP-complete. The same result
holds for fairness by replacement and withdrawal. Consequently, computing
a matching satisfying either pair of properties, while correctly reporting when none exists, is NP-hard.
\end{theorem}

\begin{proof}
We use the same reduction as in
Theorem~\ref{thm:decision:FRM_NW}. We show that, on every instance
produced by that reduction, fairness by replacement (FR), fairness by
replacement and withdrawal (FRW), and fairness by replacement and master
list (FRM) are equivalent.

We first show that justified envy by replacement and justified envy by
replacement and master list coincide on the constructed instances.
By construction, the priority ordering of every daycare is consistent
with the master list, except that an initially enrolled child is ranked
above all non-incumbent children. Hence, a non-incumbent child cannot
replace an initially enrolled child.

Consider any justified envy by replacement of a family $f$ through an
assignment $W$. Every contract rejected as a result of this deviation
therefore belongs to a non-incumbent child. For any rejected contract
belonging to a family $f'$, the replacing child of $f$ has higher daycare
priority than the rejected child of $f'$. Since both children are
non-incumbents, master-list consistency implies $f\rhd f'.$ Thus, the master-list condition in Definition~\ref{def:fairness_master} is automatically satisfied. Consequently, justified envy by replacement is equivalent to justified envy by replacement and master list on the constructed instances.

We next show that justified envy by replacement and withdrawal is
equivalent to justified envy by replacement. By construction, every
multi-child family applies only to assignments in which all of its
children attend the same daycare, and any two distinct acceptable
assignments of such a family involve different daycares. All remaining
families are singletons. Therefore, whenever a family changes from its
current assignment $Y_f$ to a distinct assignment $W$, no contract in
$Y_f\setminus W$ is at the same daycare as a newly proposed contract in
$W\setminus Y_f$. Hence, a withdrawn contract can never be used to
satisfy Condition~(3) of
Definition~\ref{def:fairness_withdraw}. On these instances, justified
envy by replacement and withdrawal therefore reduces to justified envy
by replacement.

It follows that, on every instance produced by the reduction, the three
fairness requirements are equivalent:
\[
\mathrm{FRW}
\iff
\mathrm{FR}
\iff
\mathrm{FRM}.
\]
Therefore, the NP-hardness established in
Theorem~\ref{thm:decision:FRM_NW} applies directly to deciding the
existence of an individually rational matching satisfying
non-wastefulness together with FR or FRW.

It remains to establish membership in NP. A matching serves as a
polynomial-size certificate. Feasibility and individual rationality can
be checked in polynomial time, and whether a family can claim unused
seats can be checked by enumerating its explicitly represented
assignments and applying the daycare choice functions. Hence,
non-wastefulness can be verified in polynomial time. By
Theorem~\ref{theo:check:envy:replacement}, justified envy by replacement
and justified envy by replacement and withdrawal can also be detected in
polynomial time. Thus, both existence problems belong to NP and are
therefore NP-complete.

We finally consider the problem of computing a matching satisfying
non-wastefulness together with FR or FRW. Since a polynomial-time algorithm for computing such a matching (or reporting that none exists) would also solve the corresponding NP-complete existence problem, the computation problems are NP-hard. 
\end{proof}

% \newpage

%-------------------------------------------
% SubSection : Master List
%-------------------------------------------
\subsection{Master List}

In matching problems with distributional constraints, a common approach to
mitigating the incompatibility between fairness and non-wastefulness is to
introduce a master list that restricts which instances of envy are considered
justified. Motivated by the incompatibility between fairness by replacement
and non-wastefulness established in
Theorem~\ref{theo:non_existence:FR&NW}, we consider a weaker fairness notion
that incorporates a strict master list $\rhd$ over families.

Intuitively, a family should be allowed to claim priority over another family only if it is ranked higher in the master list. 
Accordingly, justified envy by replacement is recognized only when the envying family is ranked above every family whose contract is rejected as a result of the deviation.

\begin{definition}[Fairness by Replacement and Master List (FRM)]
\label{def:fairness_master}
Given an individually rational matching $Y$, a new assignment $W$, and a
strict master list $\rhd$ over families $F$, a family $f$ is said to have
\emph{justified envy by replacement and master list} if the following
conditions hold:
\begin{description}
    \item[(1) blocking coalition:]
    $W\succ_f Y_f$ and
    $W\subseteq Ch_D(Y\setminus Y_f\cup W)$.

    \item[(2) at least one rejection:]
    \[
    |Y\setminus Y_f\cup W|
    >
    |Ch_D(Y\setminus Y_f\cup W)|.
    \]

    \item[(3) replacement only:]
    For each $(c,d)\in W\setminus Y_f$, there exists a distinct contract
    $(c',d)$ such that
    \[
    (c',d)\in Y\setminus Y_f
    \quad\text{and}\quad
    (c',d)\notin Ch_D(Y\setminus Y_f\cup W).
    \]

    \item[(4) respect for the master list:]
    For every $f'\in F(\mathrm{Re})$, we have $f\rhd f'$, where
    \[
    \mathrm{Re}
    =
    (Y\setminus Y_f\cup W)
    \setminus
    Ch_D(Y\setminus Y_f\cup W)
    \]
    is the set of rejected contracts.
\end{description}
An individually rational matching is said to be
\emph{fair by replacement and master list} if no family has justified envy
by replacement and master list.
\end{definition}
Compared with fairness by replacement in
Definition~\ref{def:fairness_replacement}, Condition~(4) imposes an additional requirement for justified envy: every family whose contract is rejected must be ranked below the envying family in the master list. Therefore, justified envy by replacement and master list is a more restrictive notion of justified envy, which leads to a weaker fairness requirement.

\begin{proposition}
\label{prop:fr_implies_frm}
Fairness by replacement implies fairness by replacement and master list.
% That is,
% \[
% \mathrm{FR}\ \Longrightarrow\ \mathrm{FRM}.
% \]
\end{proposition}

\begin{proof}
Any justified envy by replacement and master list satisfies all the
conditions for justified envy by replacement, with the additional
master-list requirement in Condition~(4) of
Definition~\ref{def:fairness_master}. Hence, if a matching admits no
justified envy by replacement, it admits no justified envy by replacement
and master list.
\end{proof}

\paragraph{Master-list consistency.}
We next consider settings in which daycare priorities are aligned with the
master list. Daycare priorities are
\emph{master-list consistent up to initial enrollments} if, at every daycare
$d$, all children initially enrolled at $d$ are ranked above all children
who are not initially enrolled at $d$, while the remaining children are
ordered according to their families' positions in $\rhd$. The relative
priority of children belonging to the same family may be determined
arbitrarily. This condition strengthens
Definition~\ref{assumption}, which requires priority for initially enrolled
children.

As the following theorem shows, even under this strong alignment between
daycare priorities and the master list, initial enrollments can make
non-wastefulness and fairness by replacement and master list incompatible.

\begin{theorem}
\label{theorem:non-existence:fairness_master&NW}
In the presence of initial enrollments, a matching satisfying both
non-wastefulness and fairness by replacement and master list may fail to
exist, even when daycare priorities are master-list consistent up to initial
enrollments.
\end{theorem}

\begin{figure}[tb]
\centering

\begin{tikzpicture}[
    >=stealth,
    family/.style={
        rounded corners=2pt,
        draw=green!60!black,
        fill=green!10,
        minimum width=8mm,
        minimum height=6mm,
        inner sep=1pt
    },
    child/.style={
        circle,
        draw=blue!70!black,
        fill=blue!10,
        minimum size=6.5mm,
        inner sep=0pt
    },
    daycare/.style={
        circle,
        draw=orange!80!black,
        fill=orange!15,
        minimum size=7mm,
        inner sep=0pt
    },
    member/.style={
        -,
        thin,
        draw=gray!55
    },
    assign/.style={
        ->,
        thick,
        draw=black
    },
    replace/.style={
        ->,
        very thick,
        dashed,
        draw=red!80
    },
    improve/.style={
        ->,
        very thick,
        dashed,
        draw=blue!80
    },
    panel/.style={
        draw=gray!55,
        rounded corners=3pt
    },
    note/.style={
        align=center,
        font=\small
    },
    violation/.style={
        align=center,
        font=\small\bfseries,
        text=red!75!black
    },
    wasteful/.style={
        align=center,
        font=\small\bfseries,
        text=blue!75!black
    }
]

% ============================================================
% CASE 1
% ============================================================

\node[
    font=\bfseries,
    text=black!70!black
] at (0,4.65)
{\(c_4\) remains at \(d_1\)};

% ============================================================
% X1
% ============================================================

\node[
    panel,
    minimum width=8.0cm,
    minimum height=5.2cm
] (P1) at (-4.25,1.45) {};

% Title bar
\node[
    fill=red!10,
    draw=red!35,
    rounded corners=3pt,
    minimum width=8.0cm,
    minimum height=8mm,
    font=\bfseries
] at (-4.25,3.65)
{(a) \(X_1=\{(c_4,d_1),(c_1,d_2),(c_5,d_2)\}\)};

% Families
\node[family] (f1a) at (-7.4,2.3) {\(f_1\)};
\node[family] (f3a) at (-7.4,1.2) {\(f_3\)};
\node[family] (f4a) at (-7.4,0.1) {\(f_4\)};

% Children
\node[child] (c1a) at (-6.1,2.3) {\(c_1\)};
\node[child] (c4a) at (-6.1,1.2) {\(c_4\)};
\node[child] (c5a) at (-6.1,0.1) {\(c_5\)};

% Family membership
\draw[member] (f1a) -- (c1a);
\draw[member] (f3a) -- (c4a);
\draw[member] (f4a) -- (c5a);

% Daycares
\node[daycare] (d1a) at (-3.7,2.0) {\(d_1\)};
\node[daycare] (d2a) at (-3.7,0.4) {\(d_2\)};

\node[font=\small] at (-2.7,2.0)
    {\(Q=1\)};
\node[font=\small] at (-2.7,0.4)
    {\(Q=2\)};

% Current matching
\draw[assign] (c4a) -- (d1a);
\draw[assign] (c1a) -- (d2a);
\draw[assign] (c5a) -- (d2a);

% FRM replacement: straight line
\draw[replace] (c4a) -- (d2a);

% Explanation
\node[
    note,
    text=red!75!black
] at (-1.35,1.25)
{
\(f_3\rhd f_4\)\\[1mm]
\(c_4\succ_{d_2}c_5\)
};

\node[violation] at (-4.25,-0.75)
{\(c_4\) can replace \(c_5\) at \(d_2\): violating FRM};

% ============================================================
% X2
% ============================================================

\node[
    panel,
    minimum width=8.0cm,
    minimum height=5.2cm
] (P2) at (4.25,1.45) {};

% Title bar
\node[
    fill=red!10,
    draw=red!35,
    rounded corners=3pt,
    minimum width=8.0cm,
    minimum height=8mm,
    font=\bfseries
] at (4.25,3.65)
{(b) \(X_2=\{(c_4,d_1),(c_2,d_2),(c_3,d_2)\}\)};

% Families
\node[family] (f1b) at (1.0,2.5) {\(f_1\)};
\node[family] (f2b) at (1.0,1.3) {\(f_2\)};
\node[family] (f3b) at (1.0,0.0) {\(f_3\)};

% Children
\node[child] (c1b) at (2.3,2.5) {\(c_1\)};
\node[child] (c2b) at (2.3,1.6) {\(c_2\)};
\node[child] (c3b) at (2.3,0.9) {\(c_3\)};
\node[child] (c4b) at (2.3,0.0) {\(c_4\)};

% Family membership
\draw[member] (f1b) -- (c1b);
\draw[member] (f2b) -- (c2b);
\draw[member] (f2b) -- (c3b);
\draw[member] (f3b) -- (c4b);

% Daycares
\node[daycare] (d1b) at (4.9,2.2) {\(d_1\)};
\node[daycare] (d2b) at (4.9,0.5) {\(d_2\)};

\node[font=\small] at (5.9,2.2)
    {\(Q=1\)};
\node[font=\small] at (5.9,0.5)
    {\(Q=2\)};

% Current matching
\draw[assign] (c4b) -- (d1b);
\draw[assign] (c2b) -- (d2b);
\draw[assign] (c3b) -- (d2b);

% FRM replacement
\draw[replace] (c1b) -- (d2b);

% Explanation
\node[
    note,
    text=red!75!black
] at (7.0,1.35)
{
\(f_1\rhd f_2\)\\[1mm]
\(c_1\succ_{d_2}c_2,c_3\)
};

\node[violation] at (4.25,-0.75)
{\(c_1\) can replace a child of \(f_2\) at \(d_2\): violating FRM};

% ============================================================
% CASE 2
% ============================================================

\node[
    font=\bfseries,
    text=black!70!black
] at (0,-1.75)
{\(c_4\) is assigned to \(d_2\)};

% ============================================================
% X3
% ============================================================

\node[
    panel,
    minimum width=8.0cm,
    minimum height=5.2cm
] (P3) at (-4.25,-4.75) {};

% Title bar
\node[
    fill=blue!10,
    draw=blue!35,
    rounded corners=3pt,
    minimum width=8.0cm,
    minimum height=8mm,
    font=\bfseries
] at (-4.25,-2.55)
{(c) \(X_3=\{(c_1,d_2),(c_4,d_2)\}\)};

% Families
\node[family] (f1c) at (-7.4,-4) {\(f_1\)};
\node[family] (f3c) at (-7.4,-5.8) {\(f_3\)};

% Children
\node[child] (c1c) at (-6.1,-4) {\(c_1\)};
\node[child] (c4c) at (-6.1,-5.8) {\(c_4\)};

% Family membership
\draw[member] (f1c) -- (c1c);
\draw[member] (f3c) -- (c4c);

% Daycares
\node[daycare] (d1c) at (-3.7,-4) {\(d_1\)};
\node[daycare] (d2c) at (-3.7,-5.8) {\(d_2\)};

\node[font=\small] at (-2.7,-4)
    {\(Q=1\)};
\node[font=\small] at (-2.7,-5.8)
    {\(Q=2\)};

\node[font=\small] at (5.9,-4)
    {\(Q=1\)};
\node[font=\small] at (5.9,-5.8)
    {\(Q=2\)};

% Current matching
\draw[assign] (c1c) -- (d2c);
\draw[assign] (c4c) -- (d2c);

% NW improvement
\draw[improve] (c1c) -- (d1c);

% Explanation
\node[
    note,
    text=blue!75!black
] at (-1.3,-4.8)
{
\(d_1\) is vacant\\[1mm]
\(d_1\succ_{f_1}d_2\)
};

\node[wasteful] at (-4,-7)
{\(c_1\) can move to \(d_1\): wasteful};

% ============================================================
% X4
% ============================================================

\node[
    panel,
    minimum width=8.0cm,
    minimum height=5.2cm
] (P4) at (4.25,-4.75) {};

% Title bar
\node[
    fill=red!10,
    draw=red!35,
    rounded corners=3pt,
    minimum width=8.0cm,
    minimum height=8mm,
    font=\bfseries
] at (4.25,-2.55)
{(d) \(X_4=\{(c_1,d_1),(c_4,d_2),(c_5,d_2)\}\)};

% Families
\node[family] (f1d) at (1.0,-3.4) {\(f_1\)};
\node[family] (f2d) at (1.0,-4.5) {\(f_2\)};
\node[family] (f3d) at (1.0,-5.7) {\(f_3\)};
\node[family] (f4d) at (1.0,-6.4) {\(f_4\)};

% Children
\node[child] (c1d) at (2.3,-3.4) {\(c_1\)};
\node[child] (c2d) at (2.3,-4.15) {\(c_2\)};
\node[child] (c3d) at (2.3,-4.85) {\(c_3\)};
\node[child] (c4d) at (2.3,-5.7) {\(c_4\)};
\node[child] (c5d) at (2.3,-6.4) {\(c_5\)};

% Family membership
\draw[member] (f1d) -- (c1d);
\draw[member] (f2d) -- (c2d);
\draw[member] (f2d) -- (c3d);
\draw[member] (f3d) -- (c4d);
\draw[member] (f4d) -- (c5d);

% Daycares
\node[daycare] (d1d) at (4.9,-4.0) {\(d_1\)};
\node[daycare] (d2d) at (4.9,-5.8) {\(d_2\)};

% Current matching
\draw[assign] (c1d) -- (d1d);
\draw[assign] (c4d) -- (d2d);
\draw[assign] (c5d) -- (d2d);

% FRM replacement
\draw[replace] (c2d) -- (d2d);
\draw[replace] (c3d) -- (d2d);

% Explanation
\node[
    note,
    text=red!75!black
] at (7.0,-4.9)
{
\(f_2\rhd f_3,f_4\)\\[1mm]
\(c_2,c_3\succ_{d_2}c_4,c_5\)
};

\node[violation] at (4.25,-7.0)
{\(f_2\) replaces \(c_4,c_5\) at \(d_2\): violating FRM};

% ============================================================
% Legend
% ============================================================

\draw[member] (-8,-8.0) -- (-7,-8.0);
\node[anchor=west, font=\small] at (-6.95,-8.0)
{family membership};

\draw[assign] (-3.9,-8.0) -- (-2.9,-8.0);
\node[anchor=west, font=\small] at (-2.87,-8.0)
{current assignment};

\draw[replace] (0.5,-8.0) -- (1.5,-8.0);
\node[anchor=west, font=\small] at (1.55,-8.0)
{FRM replacement};

\draw[improve] (4.6,-8.0) -- (5.6,-8.0);
\node[anchor=west, font=\small] at (5.68,-8.0)
{NW improvement};

\end{tikzpicture}

\caption{
The four relevant matchings in the proof of Theorem~\ref{theorem:non-existence:fairness_master&NW}.
Red indicates a violation of fairness by replacement and master list
(FRM), while blue indicates a violation of non-wastefulness (NW).
}
\label{fig:FRM-NW-nonexistence}
\end{figure}
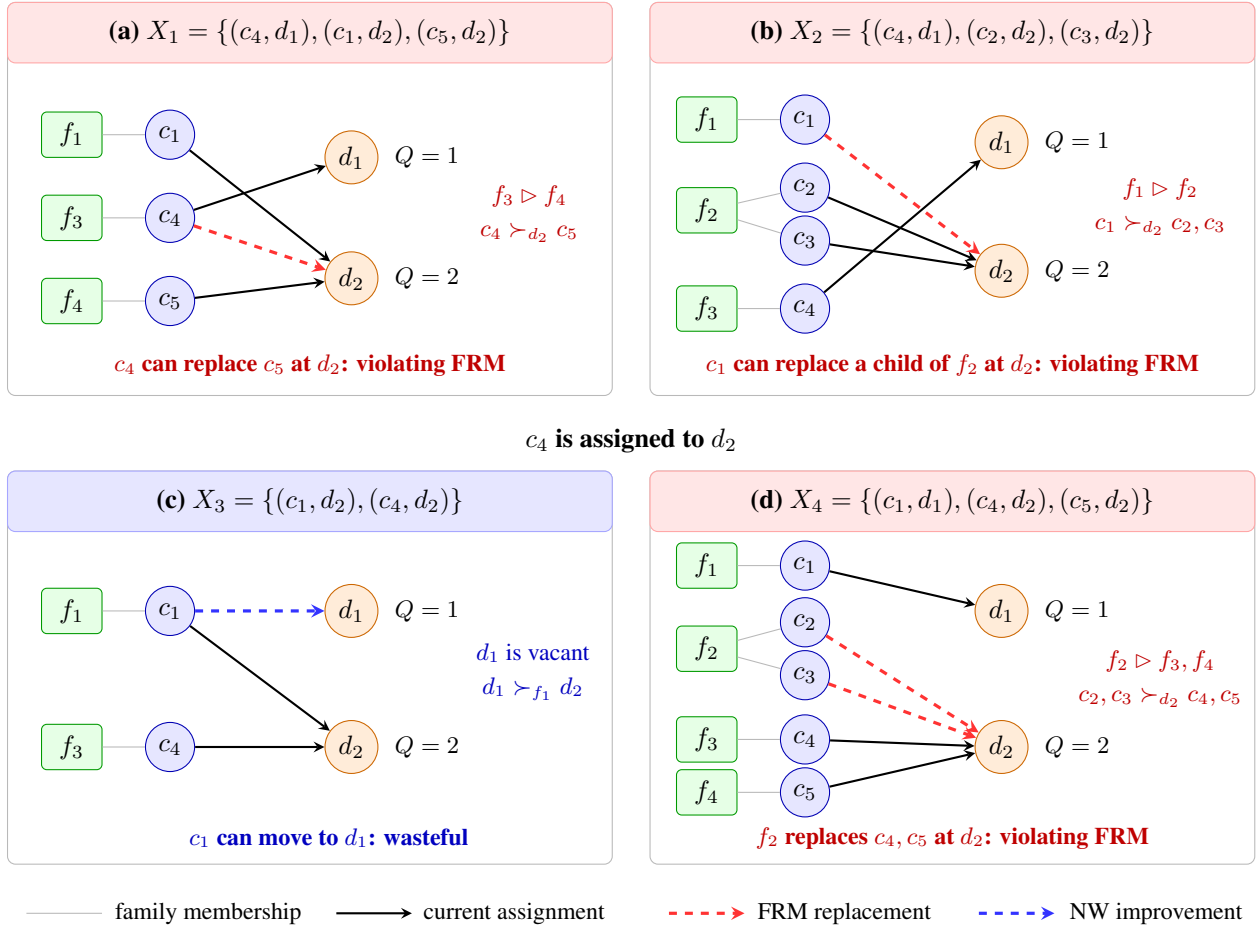

\begin{proof}
Consider two daycares $D=\{d_1,d_2\}$ and four families:
$f_1$ with $C_{f_1}=(c_1)$,
$f_2$ with $C_{f_2}=(c_2,c_3)$,
$f_3$ with $C_{f_3}=(c_4)$, and
$f_4$ with $C_{f_4}=(c_5)$.
Daycare $d_1$ has capacity one, with child $c_4$ initially enrolled there,
whereas $d_2$ has capacity two and is initially empty. The family
preferences, daycare priorities, and master list are given by
\begin{align*}
&\succ_{f_1}:
(c_1,d_1)\succ_{f_1}(c_1,d_2),
&&
\succ_{d_1}:
c_4\succ_{d_1}c_1,
\\
&\succ_{f_2}:
\{(c_2,d_2),(c_3,d_2)\},
&&
\succ_{d_2}:
c_1\succ_{d_2}c_2\succ_{d_2}c_3
\succ_{d_2}c_4\succ_{d_2}c_5,
\\
&\succ_{f_3}:
(c_4,d_2)\succ_{f_3}(c_4,d_1),
&&
\rhd:
f_1\rhd f_2\rhd f_3\rhd f_4,
\\
&\succ_{f_4}:
(c_5,d_2).
\end{align*}

The daycare priorities are master-list consistent up to initial enrollments.
At $d_1$, the initially enrolled child $c_4$ has the highest priority, while
the priorities of all non-incumbent children follow the master list. The
priority ordering at $d_2$ also follows the master list.

Individual rationality requires $c_4$ to receive an assignment weakly
preferred by $f_3$ to its initial assignment at $d_1$. Hence, $c_4$ must be
assigned either to $d_1$ or to $d_2$. We consider these two cases.

First, suppose that $c_4$ remains at $d_1$. Non-wastefulness leaves only the
following two possibilities:
\[
X_1
=
\{(c_4,d_1),(c_1,d_2),(c_5,d_2)\}
\qquad\text{and}\qquad
X_2
=
\{(c_4,d_1),(c_2,d_2),(c_3,d_2)\}.
\]
At $X_1$, family $f_3$ has justified envy by replacement and master list
through $W=\{(c_4,d_2)\}$. Child $c_4$ can replace $c_5$ at $d_2$, and
$f_3\rhd f_4$. Hence, $X_1$ is not fair by replacement and master list.

At $X_2$, family $f_1$ has justified envy by replacement and master list
through $W=\{(c_1,d_2)\}$. Child $c_1$ has higher priority at $d_2$ than
both children of $f_2$, and $f_1\rhd f_2$. Hence, $X_2$ is not fair by
replacement and master list.

Next, suppose that $c_4$ is assigned to $d_2$. It suffices to consider the
following two possibilities:
\[
X_3
=
\{(c_1,d_2),(c_4,d_2)\}
\qquad\text{and}\qquad
X_4
=
\{(c_1,d_1),(c_4,d_2),(c_5,d_2)\}.
\]
Matching $X_3$ is wasteful, since family $f_1$ can move $c_1$ from $d_2$
to its more preferred daycare $d_1$, which has a vacant seat.

At $X_4$, family $f_2$ has justified envy by replacement and master list
through
\[
W=\{(c_2,d_2),(c_3,d_2)\}.
\]
Both children of $f_2$ have higher priority at $d_2$ than $c_4$ and $c_5$,
and $f_2\rhd f_3$ and $f_2\rhd f_4$. Hence, $X_4$ is not fair by
replacement and master list.

Therefore, every individually rational matching is either wasteful or
violates fairness by replacement and master list. Hence, a matching
satisfying both properties need not exist.
\end{proof}

%-------------------------------------------
% SubsubSection : SS algorithm
%-------------------------------------------
\subsubsection{Serial Selection (SS) Algorithm}

A natural approach based on a master list is serial dictatorship, in which 
families are considered according to $\rhd$ and choose their most preferred
assignments that can be obtained without changing the assignments of other
families. Without initial enrollments, a single pass through the families is
sufficient to obtain a non-wasteful outcome even in the presence of siblings.

\begin{theorem}
\label{theo:SS:no_initial}
In the absence of initial enrollments, the serial dictatorship algorithm
terminates after a single round and satisfies both non-wastefulness and
fairness by replacement and master list.
\end{theorem}

\begin{proof}
Since there are no initial enrollments, the algorithm starts from the empty
matching. Consider a family $f$ when it is processed in the first round.
All families preceding $f$ in the master list have already selected their
assignments, whereas no family following $f$ has yet been assigned.

Once $f$ has selected its most preferred feasible assignment, no family
processed later can change the assignment of $f$, since the serial dictatorship algorithm
allows an update only if
\[
Y\setminus Y_{f'}\cup W
=
Ch_D(Y\setminus Y_{f'}\cup W),
\]
and hence no contract of another family can be rejected. Moreover, because
the matching is initially empty, a family processed later cannot release a
seat that was unavailable when $f$ was processed. Thus, no new opportunity
for $f$ can arise later in the round. It follows that no family can improve
its assignment in a second round. Hence, the algorithm terminates after a
single round. By the termination condition of the algorithm, the
resulting matching is non-wasteful.

It remains to show fairness by replacement and master list. Suppose, toward
a contradiction, that in the resulting matching $Y$, some family $f$ has
justified envy by replacement and master list through an assignment $W$.
Let
\[
\mathrm{Re}
=
(Y\setminus Y_f\cup W)
\setminus
Ch_D(Y\setminus Y_f\cup W)
\]
be the set of contracts rejected by this deviation. By
Condition~(4) of Definition~\ref{def:fairness_master}, every family
$f'\in F(\mathrm{Re})$ is ranked below $f$ in the master list.

Consider the moment when $f$ was processed in the first round. None of the
families in $F(\mathrm{Re})$ had yet been processed, and hence none of their
contracts in $\mathrm{Re}$ was present in the matching. By
Condition~(3) of Definition~\ref{def:fairness_master}, every newly proposed
contract in $W\setminus Y_f$ replaces a distinct contract in
$\mathrm{Re}$. Therefore, after removing these lower-ranked contracts,
$W$ could have been accommodated without changing the assignments of any
family preceding $f$ in the master list. Thus, $W$ was available to $f$
when $f$ was processed.

Since the serial dictatorship algorithm assigns $f$ its most preferred available assignment, $f$ must receive an assignment weakly preferred to $W$. This contradicts Condition~(1), which requires $W\succ_f Y_f$. Hence, no family has justified envy by replacement and master list, and the resulting matching is fair by replacement and master list.
\end{proof}

With initial enrollments, however, a family that changes its assignment may release seats at its initially assigned daycares, creating new opportunities for families considered earlier in the master list. Thus, a single pass of the serial dictatorship algorithm is no longer sufficient to guarantee non-wastefulness.

This motivates the repeated Serial Selection (SS) procedure described in
Algorithm~\ref{alg:Sequential}. The algorithm proceeds in rounds. In each
round, families are considered sequentially according to the fixed master
list $\rhd$. Whenever a family can obtain a strictly preferred assignment
without changing the assignments of other families, its assignment is
immediately updated and its previous contracts are withdrawn. If at least
one family changes its assignment during a round, the algorithm starts a new
round from the first family in $\rhd$. The procedure terminates when an
entire round is completed without any family changing its assignment.

Although repeated selection addresses the new opportunities created by
released seats and guarantees non-wastefulness, it does not restore fairness
by replacement and master list. Initial enrollments may create a discrepancy
between daycare priorities and the family-level master list: an initially
enrolled child has priority over non-incumbent children at its initial
daycare even when its family is ranked lower in the master list. As a result,
the outcome of the SS algorithm may still violate fairness by replacement
and master list.

\begin{algorithm}[tb]
\KwIn{An instance
$I=(C,F,D,Q,\mathcal{X},\omega,\succ_F,\succ_D)$ and a master list
$\rhd$ over families}
\KwOut{A matching $Y\subseteq\mathcal{X}$}

Initialize $Y\gets\omega$\;

\Repeat{$changed$ is $\mathrm{false}$}{
    $changed\gets\mathrm{false}$\;
    \ForEach{$f\in F$ according to $\rhd$}{
        Find $f$'s most preferred assignment $W\succ_f Y_f$ such that
        \[
        Y\setminus Y_f\cup W
        =
        Ch_D(Y\setminus Y_f\cup W).
        \]
        \If{such an assignment $W$ exists}{
            $Y\gets Y\setminus Y_f\cup W$\;
            $changed\gets\mathrm{true}$\;
        }
    }
}
\Return{$Y$}\;

\caption{Serial Selection (SS) Algorithm}
\label{alg:Sequential}
\end{algorithm}

\begin{theorem}
\label{algo:SS:NW}
The Serial Selection (SS) algorithm terminates after at most $\sum_{f\in F}|\mathcal{W}_f|$ assignment updates and satisfies non-wastefulness, but does not necessarily satisfy fairness by replacement and master list.
\end{theorem}
\begin{proof}
We first show termination. Each time a family $f$ changes its assignment,
it moves to an assignment that it strictly prefers to its current assignment.
Moreover, an update by $f$ does not change the assignment of any other
family, because the no-rejection condition
\[
Y\setminus Y_f\cup W
=
Ch_D(Y\setminus Y_f\cup W)
\]
requires every existing contract of the other families to remain selected.
Thus, once a family improves its assignment, it can never subsequently be
moved to a less-preferred assignment by another family's update.

Since family $f$ has at most $|\mathcal{W}_f|$ possible assignments and
each update strictly improves its assignment, $f$ can change its assignment
at most $|\mathcal{W}_f|-1$ times. Hence, the total number of assignment
updates is at most
\[
\sum_{f\in F}\bigl(|\mathcal{W}_f|-1\bigr),
\]
and the algorithm terminates.

We next show non-wastefulness. The algorithm terminates only when an entire
round is completed without any family changing its assignment. At
termination, therefore, there exists no family $f$ and assignment
$W\succ_f Y_f$ such that
\[
Y\setminus Y_f\cup W
=
Ch_D(Y\setminus Y_f\cup W).
\]
Thus, no family can strictly improve its assignment without changing the
assignments of other families. Hence, the resulting matching is
non-wasteful.

It remains to show that fairness by replacement and master list is not
guaranteed. Consider the instance constructed in the proof of
Theorem~\ref{theorem:non-existence:fairness_master&NW}. By that theorem,
there exists no matching satisfying both non-wastefulness and fairness by
replacement and master list. Since the SS algorithm always returns a
non-wasteful matching, its outcome on this instance cannot be fair by
replacement and master list. Hence, the SS algorithm does not necessarily
satisfy fairness by replacement and master list.
\end{proof}

%-------------------------------------------
% SubsubSection : Computational Complexity of FRM
%-------------------------------------------
\subsubsection{Computational Complexity of Problems Related to FRM}
We next examine the computational complexity associated with fairness by
replacement and master list. While verifying whether a given matching
satisfies this fairness notion can be done in polynomial time, deciding
whether there exists a non-wasteful matching satisfying it is NP-complete,
and computing such a matching is NP-hard.

\begin{theorem}
\label{theo:verification:FRM}
 Given an individually rational matching $Y$ and a strict master list
$\rhd$, whether $Y$ satisfies fairness by replacement and master list
can be verified in $O\left(|\mathcal X|\sum_{f\in F}|\mathcal W_f|
\right)$ time.
\end{theorem}

\begin{proof}
The proof follows the same argument as
Theorem~\ref{theo:check:envy:replacement}. We examine every assignment
$W\succ_fY_f$ in the explicitly represented preference ordering of
family $f$. For each such $W$, let
\[
Z=(Y\setminus Y_f)\cup W
\]
and compute
\[
Ch_D(Z)
\quad\text{and}\quad
\mathrm{Re}=Z\setminus Ch_D(Z).
\]
As shown in Theorem~\ref{theo:check:envy:replacement}, Conditions~(1)--(3)
of Definition~\ref{def:fairness_master} can be checked in
$O(|\mathcal X|)$ time.

It remains to check the master-list condition. Store the position of each family in the strict master list $\rhd$. Then, for every family 
$f'\in F(\mathrm{Re})$, we verify whether $f\rhd f'.$ 
Since $F(\mathrm{Re})$ can be obtained while computing the rejected
contracts and contains at most $|C|$ families, this additional check takes $O(|\mathcal X|)$ time.

Thus, for each assignment $W$, all conditions for justified envy by
replacement and master list can be checked in $O(|\mathcal X|)$ time.
Since at most $|\mathcal W_f|$ assignments need to be examined, whether
$f$ has justified envy by replacement and master list can be determined
in $O\bigl(|\mathcal X|\cdot|\mathcal W_f|\bigr)$ time.

Finally, to verify whether $Y$ satisfies fairness by replacement and
master list, we perform this check for every family $f\in F$. Hence, the
verification problem can be solved in polynomial time.
\end{proof}

\begin{restatable}{theorem}{nphardfrm}
\label{thm:decision:FRM_NW}
Deciding whether there exists an individually rational matching that is
both non-wasteful and fair by replacement and master list is NP-complete,
even if
\begin{itemize}
    \item each family has at most two children,
    \item every acceptable assignment places all siblings in the same daycare,
    \item each daycare has capacity at most two, and
    \item daycare priorities are master-list consistent up to initial
    enrollments.
\end{itemize}
Consequently, computing such a matching, while correctly reporting when
none exists, is NP-hard.
\end{restatable}

\begin{figure}[tb]
\centering

\begin{tikzpicture}

% ============================================================
% Styles
% ============================================================

\tikzstyle{family}=[
    rounded corners,
    draw=green!60!black,
    fill=green!10,
    minimum width=11mm,
    minimum height=7mm
]

\tikzstyle{repfamily}=[
    rounded corners,
    draw=purple!70!black,
    fill=purple!10,
    minimum width=12mm,
    minimum height=7mm
]

\tikzstyle{dummy}=[
    rounded corners,
    draw=gray!70!black,
    fill=gray!12,
    minimum width=10mm,
    minimum height=7mm
]

\tikzstyle{child}=[
    circle,
    draw=blue!70!black,
    fill=blue!10,
    minimum size=9mm,
    inner sep=1pt,
    font=\scriptsize
]

\tikzstyle{daycare}=[
    circle,
    draw=orange!80!black,
    fill=orange!15,
    minimum size=10mm,
    inner sep=0pt
]

\tikzstyle{womandaycare}=[
    circle,
    draw=purple!70!black,
    fill=purple!10,
    minimum size=10mm,
    inner sep=0pt
]

% Family -- child
\tikzstyle{member}=[
    -,
    thin,
    dashed,
    draw=gray!65
]

% Initial enrollment
\tikzstyle{initial}=[
    ->,
    very thick,
    draw=black
]

% Representation of an acceptable COM-SMTI pair
\tikzstyle{representation}=[
    ->,
    thick,
    dashed,
    draw=purple!75
]

% Dummy family's preferred transfer
\tikzstyle{prefer}=[
    ->,
    thick,
    draw=red!75
]

% ============================================================
% PANEL (a): Man gadget G_i
% ============================================================

\node[
    draw=gray!55,
    rounded corners,
    minimum width=8.6cm,
    minimum height=6.3cm
] at (-3.9, 1.0) {};

\node[
    fill=green!10,
    draw=green!35,
    rounded corners,
    minimum width=8.5cm,
    minimum height=9mm,
    font=\bfseries
] at (-3.9, 3.7)
{(a) Gadget \(G_i\) representing man \(m_i\)};

% ------------------------------------------------------------
% Families
% ------------------------------------------------------------

\node[family] (f1) at (-7.2,2.7)
    {\(f_{i,1}\)};

\node[repfamily] (f2) at (-7.2,1.2)
    {\(f_{i,2}\)};

\node[family] (f3) at (-7.2,-0.5)
    {\(f_{i,3}\)};

\node[family] (f4) at (-7.2,-1.5)
    {\(f_{i,4}\)};

% ------------------------------------------------------------
% Children
% ------------------------------------------------------------

\node[child] (a1) at (-5.5,2.7)
    {\(c_{i,1}\)};

\node[child] (a2) at (-5.5,1.65)
    {\(c_{i,2}^1\)};

\node[child] (a3) at (-5.5,0.65)
    {\(c_{i,2}^2\)};

\node[child] (a4) at (-5.5,-0.5)
    {\(c_{i,3}\)};

\node[child] (a5) at (-5.5,-1.5)
    {\(c_{i,4}\)};

% ------------------------------------------------------------
% Family membership
% ------------------------------------------------------------

\draw[member] (f1) -- (a1);

\draw[member] (f2) -- (a2);
\draw[member] (f2) -- (a3);

\draw[member] (f3) -- (a4);
\draw[member] (f4) -- (a5);

% ------------------------------------------------------------
% Internal daycares
% ------------------------------------------------------------

\node[daycare] (di1) at (-3.2,2.2)
    {\(d_{i,1}\)};

\node[daycare] (di2) at (-3.2,-0.9)
    {\(d_{i,2}\)};

\node[font=\small] at (-3.2,2.95)
    {\(Q=1\)};

\node[font=\small] at (-3.2,-1.65)
    {\(Q=2\)};

% ------------------------------------------------------------
% Initial enrollment
% c_{i,3} is initially enrolled at d_{i,1}
% ------------------------------------------------------------

\draw[initial] (a4) -- (di1);

% \node[
%     font=\scriptsize
% ] at (-4.15,0.2)
%     {initial};

% ------------------------------------------------------------
% Internal fallback of f_{i,2}
%
% These arrows only highlight the fallback assignment used
% in the reduction; the figure is not a complete
% acceptability graph.
% ------------------------------------------------------------

\draw[representation] (a2) -- (di2);
\draw[representation] (a3) -- (di2);

% ------------------------------------------------------------
% Woman daycare representing an acceptable woman w
% ------------------------------------------------------------

\node[womandaycare] (womanA) at (-1.1, 1.0)
    {\(d_w\)};

\node[
    font=\small,
    text=purple!70!black
] at (-1.1, 1.8)
    {woman daycare};

% If (m_i,w) is acceptable, f_{i,2} has the
% corresponding two-child assignment at d_w.

\draw[representation] (a2) -- (womanA);
\draw[representation] (a3) -- (womanA);

\node[
    font=\small,
    align=center,
    text=purple!70!black
] at (-1.2,-0.5)
    {\((m_i,w)\) acceptable\\
     \(\Longrightarrow f_{i,2}\) can use \(d_w\)};

% ============================================================
% PANEL (b): Woman with strict preferences
% ============================================================

\node[
    draw=gray!55,
    rounded corners,
    minimum width=6.5cm,
    minimum height=6.3cm
] at (3.8,1.0) {};

\node[
    fill=purple!10,
    draw=purple!35,
    rounded corners,
    minimum width=6.5cm,
    minimum height=9mm,
    font=\bfseries
] at (3.8,3.7)
{(b) Woman \(w\) with strict preferences};

% ------------------------------------------------------------
% Representative family f_{i,2}
% ------------------------------------------------------------

\node[repfamily] (fiB) at (1.4,2.1)
    {\(f_{i,2}\)};

\node[child] (bi1) at (3.2,2.6)
    {\(c_{i,2}^1\)};

\node[child] (bi2) at (3.2,1.6)
    {\(c_{i,2}^2\)};

\draw[member] (fiB) -- (bi1);
\draw[member] (fiB) -- (bi2);

% ------------------------------------------------------------
% Representative family f_{k,2}
% ------------------------------------------------------------

\node[repfamily] (fkB) at (1.4,-0)
    {\(f_{k,2}\)};

\node[child] (bk1) at (3.2,0.5)
    {\(c_{k,2}^1\)};

\node[child] (bk2) at (3.2,-0.5)
    {\(c_{k,2}^2\)};

\draw[member] (fkB) -- (bk1);
\draw[member] (fkB) -- (bk2);

% ------------------------------------------------------------
% Woman daycare
% ------------------------------------------------------------

\node[womandaycare] (dwB) at (5.8,1.05)
    {\(d_w\)};

\node[font=\small] at (5.8,1.75)
    {\(Q=2\)};

% ------------------------------------------------------------
% Representation of acceptable pairs
% ------------------------------------------------------------

\draw[representation] (bi1) -- (dwB);
\draw[representation] (bi2) -- (dwB);

\draw[representation] (bk1) -- (dwB);
\draw[representation] (bk2) -- (dwB);

\node[
    font=\small,
    align=center,
    text=purple!70!black
] at (3.8,-1.5)
    {priority at \(d_w\) inherits
     the \\ strict preference order of \(w\)};

% ============================================================
% PANEL (c): Woman with a tie
% ============================================================

\node[
    draw=gray!55,
    rounded corners,
    minimum width=15.3cm,
    minimum height=4.9cm
] at (-0.53,-4.8) {};

\node[
    fill=purple!10,
    draw=purple!35,
    rounded corners,
    minimum width=15.3cm,
    minimum height=9mm,
    font=\bfseries
] at (-0.53, -2.8)
{(c) Woman \(w\) with tie \((m_i,m_{i+1})\)};

% ------------------------------------------------------------
% Family f_{i,2}
% ------------------------------------------------------------

\node[repfamily] (fiC) at (-7.2,-4.7)
    {\(f_{i,2}\)};

\node[child] (ci1) at (-5.4,-4)
    {\(c_{i,2}^1\)};

\node[child] (ci2) at (-5.4,-5.4)
    {\(c_{i,2}^2\)};

\draw[member] (fiC) -- (ci1);
\draw[member] (fiC) -- (ci2);

% ------------------------------------------------------------
% Family f_{i+1,2}
% ------------------------------------------------------------

\node[repfamily] (fjC) at (6.2,-4.7)
    {\(f_{i+1,2}\)};

\node[child] (cj1) at (4.4,-4)
    {\(c_{i+1,2}^1\)};

\node[child] (cj2) at (4.4,-5.4)
    {\(c_{i+1,2}^2\)};

\draw[member] (fjC) -- (cj1);
\draw[member] (fjC) -- (cj2);

% ------------------------------------------------------------
% Two daycares representing tied woman w
% ------------------------------------------------------------

\node[womandaycare] (dwC) at (-2.2,-4.6)
    {\(d_w\)};

\node[womandaycare] (dwpC) at (1.2,-4.6)
    {\(d'_w\)};

\node[font=\small] at (-2.2,-3.7)
    {\(Q=2\)};

\node[font=\small] at (1.2,-3.7)
    {\(Q=2\)};

% ------------------------------------------------------------
% Representation of the tied acceptable pairs
%
% m_i is represented by d_w;
% m_{i+1} is represented by d'_w.
% ------------------------------------------------------------

\draw[representation] (ci1) -- (dwC);
\draw[representation] (ci2) -- (dwC);

\draw[representation] (cj1) -- (dwpC);
\draw[representation] (cj2) -- (dwpC);

\node[
    font=\small,
    text=purple!70!black
] at (-4.0,-3.65)
    {\((m_i,w)\)};

\node[
    font=\small,
    text=purple!70!black
] at (3,-3.65)
    {\((m_{i+1},w)\)};

% ------------------------------------------------------------
% Dummy families
% ------------------------------------------------------------

\node[dummy] (ew) at (-2.2,-6.6)
    {\(e_w\)};

\node[dummy] (ewp) at (1.2,-6.6)
    {\(e'_w\)};

% ------------------------------------------------------------
% Initial enrollments
% ------------------------------------------------------------

\draw[initial] (ew) -- (dwC);
\draw[initial] (ewp) -- (dwpC);

% ------------------------------------------------------------
% Dummy families prefer to transfer to the other daycare
% ------------------------------------------------------------

\draw[representation] (ew) -- (dwpC);
\draw[representation] (ewp) -- (dwC);

% ------------------------------------------------------------
% Priority lists
% ------------------------------------------------------------

\node[
    font=\small,
    text=purple!70!black
] at (-5.0,-6.6)
    {\(d_w:\ e_w\succ f_{i,2}\succ e'_w\)};

\node[
    font=\small,
    text=purple!70!black
] at (4.5,-6.6)
    {\(d'_w:\ e'_w\succ f_{i+1,2}\succ e_w\)};

% ============================================================
% Legend
% ============================================================

\draw[member] (-7.7,-7.7) -- (-6.8,-7.7);

\node[
    anchor=west,
    font=\small
] at (-6.75,-7.7)
    {family membership};

\draw[initial] (-2.8,-7.7) -- (-1.9,-7.7);

\node[
    anchor=west,
    font=\small
] at (-1.75,-7.7)
    {initial enrollment};

\draw[representation] (2.0,-7.7) -- (2.9,-7.7);

\node[
    anchor=west,
    font=\small
] at (3.0,-7.7)
    {possible assignment};

% \draw[prefer] (4.8,-7.7) -- (5.7,-7.7);

% \node[
%     anchor=west,
%     font=\small
% ] at (5.85,-7.7)
%     {dummy family's preferred transfer};

\end{tikzpicture}

\caption{
Illustration of the reduction from \textsc{COM-SMTI}.
Panel (a) shows the gadget \(G_i\) associated with man \(m_i\);
panel (b) shows the daycare representing a woman with a strict
preference list; and panel (c) shows the two-daycare gadget representing
a woman whose preference list consists of a tie.
}
\label{fig:COM-SMTI-reduction}
\end{figure}
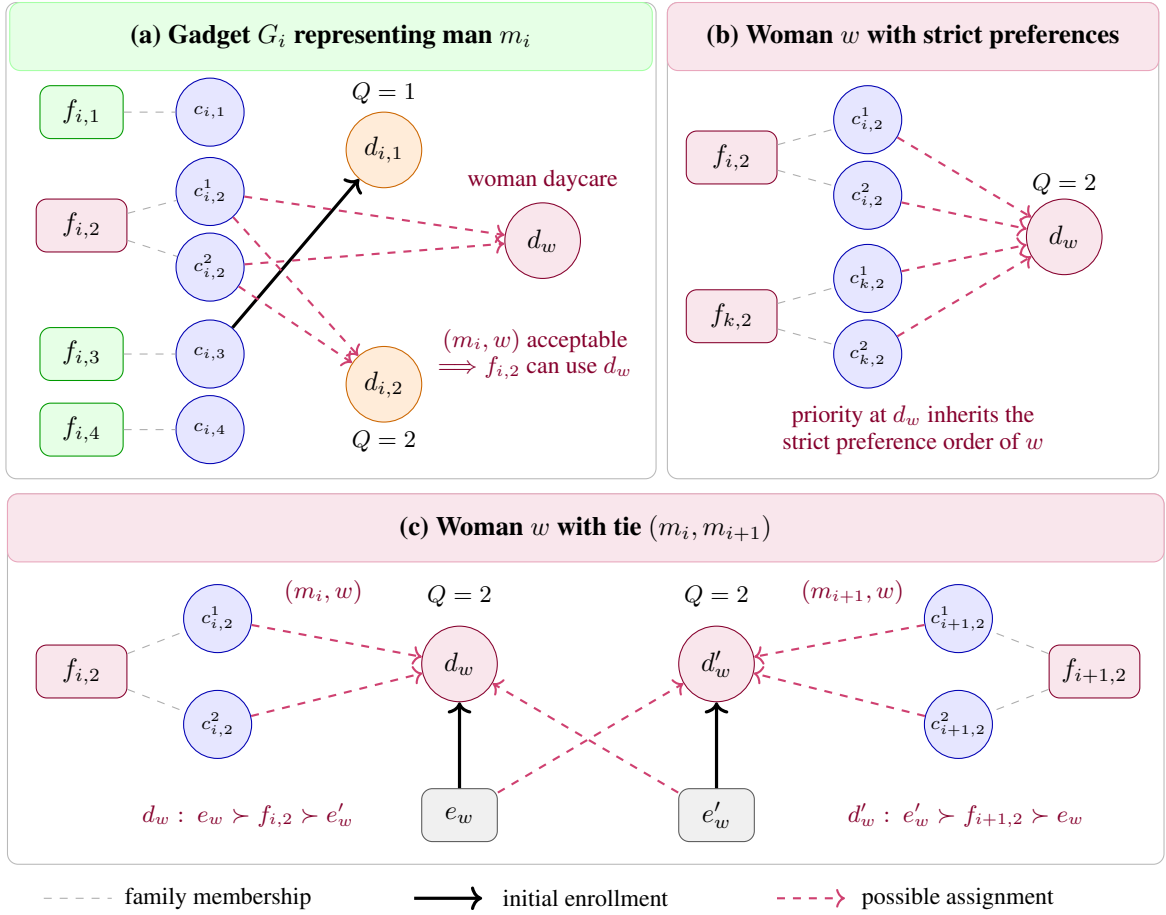

\begin{proof}
Inclusion in NP follows from the fact that verifying if a matching is individually rational, non-wasteful, and fair by replacement and master list can be done in polynomial time by Theorems~\ref{theo:verification:nw-fairness} and~\ref{theo:verification:FRM}.

We reduce from the problem of deciding whether a complete weakly stable
matching exists in the Stable Marriage with Ties and Incomplete Lists
(\textsc{COM-SMTI}) problem, which is NP-hard \citep{IMS08a}. An instance of
\textsc{COM-SMTI} consists of two sets
\[
M=\{m_1,\ldots,m_n\}
\quad\text{and}\quad
W=\{w_1,\ldots,w_n\}
\]
of men and women, respectively. Men have strict preferences over their
acceptable women. Women's preferences are derived from a master list with
ties of the form
\[
(m_1,m_2),(m_3,m_4),\ldots,(m_{n-1},m_n),
\]
where parentheses indicate ties. Each woman either has a single tie of
length two as her preference list or has a strict preference list
containing at most one acceptable man from each pair in the master list.
The NP-hardness of this restricted version follows from Theorem~3.2 \citep{IMS08a}, with the roles of men and women
interchanged.

\paragraph{Construction.}
For each man $m_i\in M$, we construct a gadget $G_i$ based on the ``no-instance'' in Theorem~\ref{theorem:non-existence:fairness_master&NW}. The gadget contains
two daycares, $d_{i,1}$ of capacity one and $d_{i,2}$ of capacity two,
and four families
\[
f_{i,1},f_{i,2},f_{i,3},f_{i,4}.
\]
Family $f_{i,2}$ represents man $m_i$ and consists of two children,
whereas each of the other three families is a singleton. The child of
$f_{i,3}$ is initially enrolled in $d_{i,1}$.

The preferences within $G_i$ follow the counterexample in
Theorem~\ref{theorem:non-existence:fairness_master&NW}:
$f_{i,1}$ prefers $d_{i,1}$ to $d_{i,2}$,
$f_{i,3}$ prefers $d_{i,2}$ to its initial assignment $d_{i,1}$,
and $f_{i,4}$ finds $d_{i,2}$ acceptable. Family $f_{i,2}$ inherits
the preference list of $m_i$, with each acceptable woman replaced by the
corresponding daycare constructed below, followed by $d_{i,2}$ as its
least-preferred acceptable assignment. Family $f_{i,2}$ always applies
with both children to the same daycare.

For each woman $w\in W$ with a strict preference list, we create a daycare
$d_w$ of capacity two. Its priority ordering is inherited from $w$'s
preference list, with each acceptable man $m_i$ replaced by the
corresponding family $f_{i,2}$.

For each woman $w$ whose preference list consists of a tie
$(m_i,m_{i+1})$, we create two daycares $d_w$ and $d'_w$, each of
capacity two. We also introduce two singleton dummy families $e_w$ and
$e'_w$. Family $e_w$ is initially enrolled in $d_w$ and prefers
$d'_w$ to $d_w$, whereas $e'_w$ is initially enrolled in $d'_w$ and
prefers $d_w$ to $d'_w$. The priority orders are
\[
d_w:\quad e_w\succ f_{i,2}\succ e'_w
\]
and
\[
d'_w:\quad e'_w\succ f_{i+1,2}\succ e_w.
\]
In the preference list of $f_{i,2}$, woman $w$ is represented by $d_w$,
whereas in that of $f_{i+1,2}$ she is represented by $d'_w$.

Finally, we construct a strict master list $\rhd$ over all families.
Break each tie in the \textsc{COM-SMTI} master list arbitrarily to obtain
the strict order
\[
m_1,m_2,\ldots,m_n.
\]
For each $m_i$, insert the four families of $G_i$ consecutively as
\[
f_{i,1}\rhd f_{i,2}\rhd f_{i,3}\rhd f_{i,4}.
\]
After all gadget families, append the dummy families $e_w,e'_w$ in an
arbitrary order. By construction, every daycare priority list follows
$\rhd$, except that a child initially enrolled at a daycare is moved
above all non-incumbent children at that daycare. Hence, all daycare
priorities are master-list consistent up to initial enrollments.

The construction is clearly polynomial, every family has size at most
two, every daycare has capacity at most two, and every multi-child family
applies only to same-daycare assignments.

\paragraph{Completeness.}
Suppose that the \textsc{COM-SMTI} instance admits a complete weakly
stable matching $\mu$. We construct a matching $X$ in our instance.

Since $\mu$ is complete, every man $m_i$ is matched to some woman $w$.
We assign the corresponding family $f_{i,2}$ to the daycare representing
$w$. If $w$ has a strict preference list, both children of $f_{i,2}$ are
assigned to $d_w$.

Now suppose that $w$ has the tied preference list $(m_i,m_{i+1})$.
If $\mu(m_i)=w$, assign $f_{i,2}$ to $d_w$ and assign both dummy families
$e_w$ and $e'_w$ to $d'_w$. If instead $\mu(m_{i+1})=w$, assign
$f_{i+1,2}$ to $d'_w$ and both dummy families to $d_w$.

For every $i$, family $f_{i,2}$ is now assigned outside $G_i$ to an
assignment it prefers to $d_{i,2}$. We complete the matching inside
$G_i$ by assigning $f_{i,1}$ to $d_{i,1}$ and assigning both
$f_{i,3}$ and $f_{i,4}$ to $d_{i,2}$.

We show that $X$ is non-wasteful and fair by replacement and master list.
Within every gadget $G_i$, all families other than $f_{i,2}$ receive
their most-preferred assignments, while $f_{i,2}$ receives an assignment
that it prefers to every assignment inside $G_i$. Thus, no relevant
deviation arises entirely within $G_i$.

Consider next a woman $w$ with a tied preference list. One of $d_w$ and
$d'_w$ is occupied by the two children of the family corresponding to
$\mu(w)$, while the other contains the two dummy families. The family
corresponding to the other man in the tie cannot replace the two dummy
families: it requires both seats, but one of the dummy children is
initially enrolled at that daycare and therefore has higher priority than
any non-incumbent child. The dummy family that is not at its preferred
daycare cannot generate justified envy by replacement and master list
either, since the family occupying its preferred daycare precedes it in
the master list.

Finally, consider a woman $w$ with a strict preference list. The
corresponding daycare $d_w$ is occupied by the family $f_{k,2}$
representing $\mu(w)=m_k$. If another family $f_{i,2}$ could obtain
justified envy by replacement and master list at $d_w$, then $m_i$ would
prefer $w$ to $\mu(m_i)$ and $w$ would prefer $m_i$ to $m_k$. Thus,
$(m_i,w)$ would be a blocking pair of $\mu$, contradicting its weak
stability.

No family can claim unused seats in the constructed matching by the same
case analysis: all improvements available within the gadgets have already
been exhausted, and the capacities associated with women are filled as
described above. Hence, $X$ is both non-wasteful and fair by replacement
and master list.

\paragraph{Soundness.}
Conversely, suppose that the constructed instance admits an individually
rational matching $X$ satisfying both non-wastefulness and fairness by
replacement and master list.

We first establish that every family $f_{i,2}$ must be assigned to a
daycare representing a woman. Fix $i$ and suppose, for contradiction,
that $f_{i,2}$ is not assigned to any such daycare. Since $f_{i,2}$
always applies with both children to the same daycare, its only remaining
acceptable assignment outside the woman daycares is $d_{i,2}$; it may
also be unmatched.

Consider the restriction of $X$ to the two internal daycares
$d_{i,1}$ and $d_{i,2}$. The only families other than $f_{i,2}$ that
find these daycares acceptable are
$f_{i,1},f_{i,3}$, and $f_{i,4}$. Moreover, none of these three families
has an acceptable assignment outside $G_i$. Hence, their possible
improvements and the priority relations at $d_{i,1}$ and $d_{i,2}$ are
exactly the same as in the no-instance of
Theorem~\ref{theorem:non-existence:fairness_master&NW}.

If $f_{i,2}$ is not assigned to a daycare representing a woman, then it is either assigned to $d_{i,2}$ or unmatched. In either case, the restriction of $X$ to $G_i$ reproduces the relevant structure of the no-instance in Theorem~\ref{theorem:non-existence:fairness_master&NW}: either $f_{i,2}$ can improve by taking its internal assignment at $d_{i,2}$, or there exists an unused-seat improvement or justified envy by replacement and master list entirely within $G_i$. Such a deviation remains valid in the full matching $X$, contradicting non-wastefulness or fairness by replacement and master list. Therefore, every $f_{i,2}$ must be assigned to a daycare representing a woman.

We next show that these assignments induce a one-to-one matching between
the families $f_{i,2}$ and the women. Each family $f_{i,2}$ consists of
two children and therefore occupies both seats of any daycare representing
a woman to which it is assigned. A woman with a strict preference list
corresponds to a single daycare $d_w$ of capacity two, and hence can
accommodate exactly one family of type $f_{i,2}$.

Now consider a woman $w$ with a tied preference list. The two associated
daycares $d_w$ and $d'_w$ have total capacity four. The dummy families
$e_w$ and $e'_w$ are initially assigned to $d_w$ and $d'_w$,
respectively, and each finds only these two daycares acceptable.
Individual rationality therefore requires both dummy families to remain
assigned to one of these two daycares. Together they occupy two seats,
leaving at most two seats for families of type $f_{i,2}$. Since each such
family requires two seats at the same daycare, at most one family
$f_{i,2}$ can be assigned to the daycares associated with $w$.

There are $n$ families $f_{i,2}$, and we have shown that every one of them
must be assigned to a daycare representing a woman. There are also exactly
$n$ women, and each woman can accommodate at most one such family.
Consequently, every woman must accommodate exactly one family
$f_{i,2}$. We can therefore define a complete matching $\mu$ in the
original \textsc{COM-SMTI} instance by setting
\[
\mu(m_i)=w
\]
whenever $f_{i,2}$ is assigned to a daycare associated with $w$.

It remains to show that $\mu$ is weakly stable. Suppose, for contradiction,
that $(m_i,w)$ is a blocking pair of $\mu$. Since a blocking pair under
weak stability requires both agents to strictly prefer each other to their
current partners, $w$ cannot have a preference list consisting of a single
tie. Hence, $w$ has a strict preference list.

Let $\mu(w)=m_k$. Since $(m_i,w)$ blocks $\mu$, we have
\[
w\succ_{m_i}\mu(m_i)
\qquad\text{and}\qquad
m_i\succ_w m_k.
\]
By construction, the first relation implies that $f_{i,2}$ strictly
prefers the assignment at $d_w$ to its current assignment. The second
relation implies that both children of $f_{i,2}$ have higher priority at
$d_w$ than the two children of $f_{k,2}$. Since $d_w$ has capacity two,
assigning $f_{i,2}$ to $d_w$ therefore rejects both contracts of
$f_{k,2}$.

Moreover, the priority ordering at $d_w$ follows the family master list.
Hence,
\[
m_i\succ_w m_k
\quad\Longrightarrow\quad
f_{i,2}\rhd f_{k,2}.
\]
Thus, the proposed assignment of $f_{i,2}$ to $d_w$ satisfies the
blocking-coalition, rejection, replacement-only, and master-list
conditions of justified envy by replacement and master list. Therefore,
$f_{i,2}$ has justified envy by replacement and master list toward
$f_{k,2}$, contradicting the assumption that $X$ satisfies FRM.

Hence, $\mu$ contains no blocking pair and is weakly stable. Since $\mu$
is also complete, the original \textsc{COM-SMTI} instance admits a
complete weakly stable matching.

We conclude that the \textsc{COM-SMTI} instance admits a complete weakly
stable matching if and only if the constructed instance admits a matching
satisfying both non-wastefulness and fairness by replacement and master
list. Hence, the decision problem is NP-hard.
\end{proof}

\subsection{Resettlement}

The impossibility result in
Theorem~\ref{theorem:non-existence:fairness_master&NW} shows that
incorporating a master list alone is insufficient to reconcile fairness
by replacement with non-wastefulness in the presence of initial
enrollments. A key difficulty is that a deviation may displace families
from their current assignments, potentially causing them to lose the
protection afforded by their initial enrollments.

To address this issue, we introduce the idea of \emph{resettlement}, which
protects the initial entitlements of families displaced by a justified
deviation. The key idea is to allow a higher-ranked family to displace
lower-ranked families only when the displaced families can be returned to
their initial assignments while preserving feasibility. Thus, displacement
does not force the affected families to give up the protection provided by
their initial enrollments. This leads to a fairness notion that combines a
master-list restriction with feasible resettlement.

\begin{definition}[Fairness by Master List and Resettlement]
\label{def:fairness_resettlement}
Given an individually rational matching $Y$, a new assignment $W$ for a
family $f$, and a strict master list $\rhd$ over families $F$, family $f$ is
said to have \emph{justified envy by master list and resettlement} if the
following conditions hold:
\begin{description}
    \item[(1) blocking coalition:]
    $W\succ_f Y_f$ and
    $
    W\subseteq Ch_D(Y\setminus Y_f\cup W).
    $

    \item[(2) at least one rejection:]
    \[
    |Y\setminus Y_f\cup W|
    >
    |Ch_D(Y\setminus Y_f\cup W)|.
    \]

    \item[(3) respect for the master list:]
    For every $f'\in F(\mathrm{Re})$, we have $f\rhd f'$, where
    \[
    \mathrm{Re}
    =
    (Y\setminus Y_f\cup W)
    \setminus
    Ch_D(Y\setminus Y_f\cup W)
    \]
    is the set of rejected contracts.

    \item[(4) feasible resettlement:]
    Let $F'=F(\mathrm{Re})$. The matching
    $
    Y\setminus(Y_f\cup Y_{F'})
    \cup(W\cup\omega_{F'})
    $
    is feasible.
\end{description}
An individually rational matching is said to be
\emph{fair by master list and resettlement} if no family has justified envy
by master list and resettlement.
\end{definition}

Compared with fairness in Definition~\ref{def:fairness}, fairness by master
list and resettlement imposes two additional requirements on justified
envy. Condition~(3) requires the envying family $f$ to be ranked above every
family whose contract is rejected as a result of the deviation.
Condition~(4) protects these rejected families through resettlement. Letting
$F'=F(\mathrm{Re})$, we remove the current assignments of $f$ and the rejected families $F'$, assign $f$ to $W$, and return every family in $F'$ to its initial assignment. The resulting matching
$
Y\setminus(Y_f\cup Y_{F'})
\cup(W\cup\omega_{F'})
$
is required to be feasible.

%-------------------------------------------
% SubsubSection : GI algorithm
%-------------------------------------------
\subsubsection{Greedy Improvement (GI) Algorithm}

We next introduce the Greedy Improvement (GI) algorithm, which repeatedly allows families to improve their assignments either without displacing other
families or through justified envy by master list and resettlement. Let
\[
f_1\rhd f_2\rhd\cdots\rhd f_m
\]
be a strict master list over families. The algorithm starts from the initial
matching $Y=\omega$ and examines families according to the master list. For
each family $f_i$, it searches for a strictly preferred assignment $W$ that
can be obtained in one of two ways. First, $f_i$ may improve without
displacing any other family. Second, $f_i$ may exercise justified envy by
master list and resettlement; in this case, $f_i$ is assigned to $W$, while
all rejected families $F'=F(\mathrm{Re})$ are returned to their initial
assignments $\omega_{F'}$.

Whenever either type of improvement is implemented, the algorithm
immediately restarts the scan from the highest-ranked family $f_1$. If no
such improvement is available to $f_i$, the algorithm proceeds to
$f_{i+1}$. The procedure terminates after a complete scan of the master list
in which no family can improve.

\begin{algorithm}[tb]
\KwIn{An instance
$I=(C,F,D,Q,\mathcal{X},\omega,\succ_F,\succ_D)$ and a strict master list
$f_1\rhd f_2\rhd\cdots\rhd f_m$}
\KwOut{A matching $Y$}

Initialize $Y\gets\omega$ and $i\gets1$\;

\While{$i\leq m$}{
    Let $f_i$ be the family currently examined\;

    Find the most preferred assignment $W\succ_{f_i}Y_{f_i}$ such that
    either
    \[
    Y\setminus Y_{f_i}\cup W
    =
    Ch_D(Y\setminus Y_{f_i}\cup W),
    \]
    or $f_i$ has justified envy by master list and resettlement
    through $W$\;

    \eIf{such an assignment $W$ exists}{
        \eIf{
        $Y\setminus Y_{f_i}\cup W
        =
        Ch_D(Y\setminus Y_{f_i}\cup W)$
        }{
            $Y\gets Y\setminus Y_{f_i}\cup W$\;
        }{
            Let
            \[
            \mathrm{Re}
            =
            (Y\setminus Y_{f_i}\cup W)
            \setminus
            Ch_D(Y\setminus Y_{f_i}\cup W)
            \]
            and $F'=F(\mathrm{Re})$\;
            
            $Y\gets
            Y\setminus(Y_{f_i}\cup Y_{F'})
            \cup(W\cup\omega_{F'})$\;
        }
        $i\gets1$\;
    }{
        $i\gets i+1$\;
    }
}
\Return{$Y$}\;

\caption{Greedy Improvement (GI) Algorithm}
\label{alg:greedy_improve}
\end{algorithm}

The GI algorithm follows a repeated-improvement structure similar to the SS
algorithm, but differs from it in two important respects. First, SS allows a
family to improve only without changing the assignments of other families,
whereas GI additionally permits improvements through justified envy by
master list and resettlement. Such an improvement may displace lower-ranked
families, which are then returned to their initial assignments. Second,
after each improvement, GI immediately restarts the scan from the
highest-ranked family. This restart is necessary because resettlement may
change the assignments of multiple families and thereby create new
improvement opportunities for families examined earlier in the master list.

The following theorem establishes that the GI algorithm restores
compatibility between non-wastefulness and fairness by master list and
resettlement.

\begin{theorem}
\label{theo:existence:FMR&NW}
Given any strict master list $\rhd$, the Greedy Improvement (GI) algorithm
terminates and returns an individually rational matching that satisfies both
non-wastefulness and fairness by master list and resettlement.
\end{theorem}

\begin{proof}
We first show that feasibility and individual rationality are preserved
throughout the algorithm.

Initially, $Y=\omega$, which is feasible and individually rational.
Suppose first that $f_i$ improves through an assignment $W$ without
displacing another family. By the update condition,
\[
Y\setminus Y_{f_i}\cup W
=
Ch_D(Y\setminus Y_{f_i}\cup W),
\]
so the updated matching is feasible. Family $f_i$ strictly improves, while
the assignments of all other families remain unchanged. Hence, individual
rationality is preserved.

Suppose instead that $f_i$ has justified envy by master list and
resettlement through $W$, and let $F'=F(\mathrm{Re})$. By Condition~(4) of
Definition~\ref{def:fairness_resettlement}, the updated matching
\[
Y\setminus(Y_{f_i}\cup Y_{F'})
\cup(W\cup\omega_{F'})
\]
is feasible. Family $f_i$ strictly improves, every family in $F'$ returns
to its initial assignment, and the assignments of all other families remain
unchanged. Since the matching before the update is individually rational,
and an initial assignment is individually rational for its family,
individual rationality is again preserved.

We next show termination. Consider the assignment profile
\[
(Y_{f_1},Y_{f_2},\ldots,Y_{f_m}),
\]
where $f_1\rhd f_2\rhd\cdots\rhd f_m$, and compare assignment profiles
lexicographically according to the families' preferences and the master
list. That is, among two distinct profiles, consider the highest-ranked
family whose assignments differ; the profile giving that family its more
preferred assignment is ranked higher.

If $f_i$ improves without displacement, the assignments of
$f_1,\ldots,f_{i-1}$ remain unchanged, while $f_i$ strictly improves.
Hence, the assignment profile strictly increases in the lexicographic
order.

If $f_i$ improves through justified envy by master list and resettlement,
Condition~(3) of Definition~\ref{def:fairness_resettlement} implies that
every family in $F'=F(\mathrm{Re})$ is ranked below $f_i$ in the master
list. Thus, the assignments of $f_1,\ldots,f_{i-1}$ remain unchanged,
$f_i$ strictly improves, and only families ranked below $f_i$ may be
resettled. Hence, the assignment profile again strictly increases in the
lexicographic order.

Since every family has finitely many possible assignments, there are only
finitely many assignment profiles. Therefore, the GI algorithm must
terminate.

It remains to establish non-wastefulness and fairness by master list and
resettlement. When the algorithm terminates, it has completed a scan of all
families without performing an update.

Suppose, toward a contradiction, that the final matching $Y$ is wasteful.
Then there exist a family $f_i$ and an assignment $W\succ_{f_i}Y_{f_i}$
such that
\[
Y\setminus Y_{f_i}\cup W
=
Ch_D(Y\setminus Y_{f_i}\cup W).
\]
When $f_i$ was examined during the final scan, this improvement would have
triggered an update, contradicting termination. Hence, $Y$ is
non-wasteful.

Similarly, suppose that $Y$ is not fair by master list and resettlement.
Then some family $f_i$ has justified envy by master list and resettlement
through an assignment $W\succ_{f_i}Y_{f_i}$. When $f_i$ was examined
during the final scan, this improvement would also have triggered an update,
again contradicting termination. Hence, $Y$ is fair by master list and
resettlement.

Therefore, the GI algorithm terminates at an individually rational matching
that satisfies both non-wastefulness and fairness by master list and
resettlement.
\end{proof}

Algorithm~\ref{alg:greedy_improve} is guaranteed to terminate, but the
number of updates need not be polynomially bounded. Although each
individual improvement can be identified and implemented in polynomial
time, resettlement may return lower-ranked families to their initial
assignments, allowing them to improve again later. The lexicographic
potential argument therefore guarantees finite termination but not a
polynomial running time. This is consistent with the PLS-hardness result
for computing a matching satisfying non-wastefulness and fairness by
master list and resettlement established below.

% \newpage
%-----------------------------------
% Subsection: Complexity of FMR
%-----------------------------------
\subsection{Computational Complexity of Problems Related to FMR }
\label{subsec:complexity:fmr}

Theorem~\ref{theo:existence:FMR&NW} establishes existence by a
finite improvement process.  We now show that finding such an outcome is
PLS-hard, even under a strong alignment between daycare priorities and the
family master list.  Throughout this subsection, the complete individually
rational part of every family preference list is explicitly given in the input: for every family $f$, all assignments $W\succeq_f\omega_f$ are listed in strictly decreasing order, and every listed assignment specifies the contract, or non-assignment, of every child in $C_f$.

For a feasible and individually rational matching $Y$, a family $f$, and an
alternative assignment $W$, write
\[
 Z(Y,f,W):=Y\setminus Y_f\cup W,
 \qquad
 \mathrm{Re}(Y,f,W):=Z(Y,f,W)\setminus
 \Ch_D\bigl(Z(Y,f,W)\bigr).
\]
Thus, if $W\succ_fY_f$ and $\mathrm{Re}(Y,f,W)=\varnothing$, then $f$ can
improve by utilizing vacant seats.  If $\mathrm{Re}(Y,f,W)\neq\varnothing$ and
the remaining conditions of Definition~\ref{def:fairness_resettlement} hold,
then $(f,W)$ is an FMR improvement.  Hence a matching is non-wasteful and fair
by master list and resettlement exactly when neither type of improvement exists.

%----------------------------------------------------------
% %  subsubsection -- Source problem
%----------------------------------------------------------
\subsubsection{Source Problem}

We reduce from \LLMS, which is PLS-complete~\citep{ScTa25a}. An instance consists of an ordered CNF formula
\[
\varphi=C_1\wedge\cdots\wedge C_m,
\qquad
1\le |C_i|\le 4,
\]
where the clauses are ordered by significance, with $C_1$ being the most
significant clause.

For a truth assignment $\alpha$, let
\[
\alpha(C_i)=
\begin{cases}
1,&\text{if $C_i$ is satisfied under $\alpha$},\\
0,&\text{otherwise},
\end{cases}
\]
and define its value by the clause-satisfaction vector
\[
\val(\alpha)
=
\bigl(\alpha(C_1),\ldots,\alpha(C_m)\bigr).
\]
These vectors are compared lexicographically. Thus, between two truth
assignments, the preferred one is determined by the first clause on which
their satisfaction statuses differ: satisfying an earlier clause is more
important than the satisfaction statuses of all later clauses.

The neighborhood of a truth assignment consists of the assignments
obtained by flipping a single variable. 
The objective of \LLMS  is to find
a \emph{locally optimal} truth assignment $\alpha$, i.e., one for which
no single-variable flip produces an assignment $\alpha'$ satisfying
\[
\val(\alpha')>_{\mathrm{lex}}\val(\alpha).
\]
In particular, the problem is not to find an assignment satisfying all
clauses; a locally optimal assignment may leave some clauses unsatisfied.

We may assume without loss of generality that no clause is tautological
and that no variable occurs more than once in a clause. For
$x\in\mathrm{var}(C_i)$, let $\ell_{i,x}$ denote the literal of $x$ in
$C_i$, and let
\[
\ell_{i,x}(b)\in\{0,1\}
\]
denote the truth value of this literal when $x=b$. Define
\[
k_i:=|C_i|,
\qquad
\Delta:=\max_x \occ(x),
\]
where $\occ(x)$ is the number of clauses containing either $x$ or
$\neg x$.

An occurrence $\tau=(i,x)$ with $x\in\mathrm{var}(C_i)$ is called
\emph{admissible} if no earlier singleton clause contains $x$. Let
\[
\mathcal A
=
\{(i,x):(i,x)\text{ is admissible}\},
\]
and, for each $\tau=(i,x)\in\mathcal A$,  let
\[
b_\tau\in\{0,1\}
\qquad\text{be defined by}\qquad
\ell_{i,x}(b_\tau)=0.
\]
Thus, $b_\tau$ is the value of $x$ for which its literal in $C_i$ is
false. Finally, define
\[
H_\tau
=
\{h<i:x\in\mathrm{var}(C_h)\},
\qquad
J_\tau
=
\{j>i:x\in\mathrm{var}(C_j)\}.
\]
Hence, $H_\tau$ contains the more significant clauses involving $x$,
whereas $J_\tau$ contains the less significant ones.

For an admissible occurrence $\tau=(i,y)$, we call $y$ the
\emph{pivot variable}. The transition associated with $\tau$ represents
a candidate local move that flips $y$, with $C_i$ being the first clause
whose truth value improves.

%----------------------------------------------------------
% %  old version -- Source problem
%----------------------------------------------------------
% \paragraph{Source problem.}
% We reduce from $\LLMS$, which is PLS-complete~\citep{SchederTantow}.  An
% instance consists of an ordered formula
% \[
%  \varphi=C_1\wedge\cdots\wedge C_m,
%  \qquad 1\le |C_i|\le4,
% \]
% where $C_1$ is the most significant clause.  The value of a truth assignment
% $\alpha$ is
% \[
%  \val(\alpha):=(\alpha(C_1),\ldots,\alpha(C_m)),
% \]
% ordered lexicographically, and a local move flips one variable.  We may assume
% that no clause is tautological and that no variable occurs twice in one clause.
% For $x\in\mathrm{var}(C_i)$, let $\ell_{i,x}$ denote the literal of $x$ in
% $C_i$, and let $\ell_{i,x}(b)\in\{0,1\}$ be its truth value when $x=b$.
% Put
% \[
%  k_i:=|C_i|,
%  \qquad
%  \Delta:=\max_x\occ(x),
% \]
% where $\occ(x)$ is the number of clauses containing $x$ or $\neg x$.

% An occurrence $\tau=(i,x)$ with $x\in\mathrm{var}(C_i)$ is called
% \emph{admissible} if no earlier singleton clause contains $x$.  Let
% \[
%  \mathcal A:=\{(i,x): (i,x)\text{ is admissible}\},
%  \qquad
%  b_\tau\in\{0,1\}:\ \ell_{i,x}(b_\tau)=0,
% \]
% and, for $\tau=(i,x)$, define
% \[
%  H_\tau:=\{h<i:x\in\mathrm{var}(C_h)\},
%  \qquad
%  J_\tau:=\{j>i:x\in\mathrm{var}(C_j)\}.
% \]

%----------------------------------------------------------
% %  end of old version -- Source problem
%----------------------------------------------------------

\subsubsection{Construction}

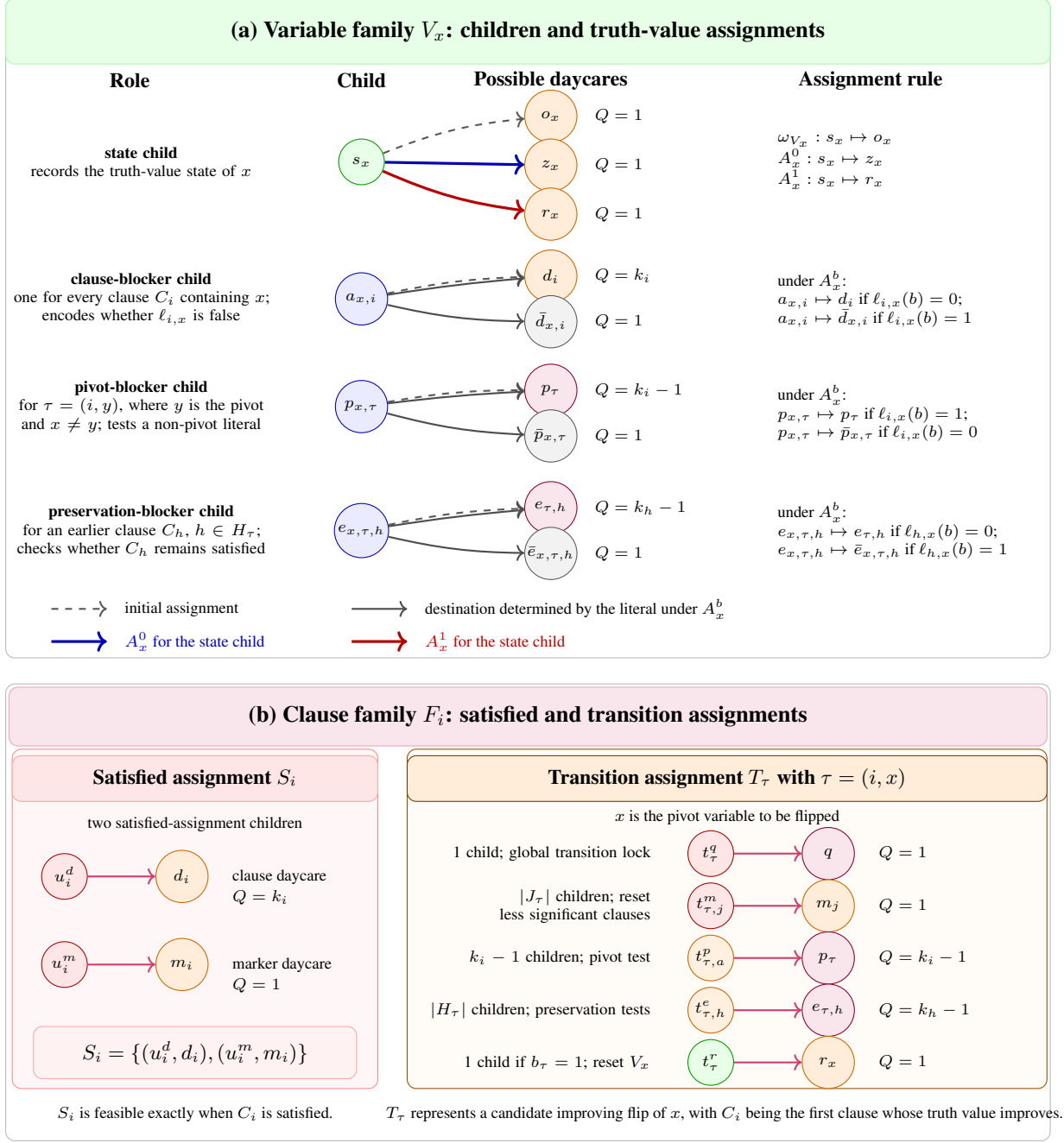
\begin{figure*}[t]
\centering

\begin{tikzpicture}

% ============================================================
% Styles
% ============================================================

\tikzstyle{statechild}=[
    circle,
    draw=green!60!black,
    fill=green!10,
    minimum size=7mm,
    inner sep=1pt,
    font=\scriptsize
]

\tikzstyle{blockchild}=[
    circle,
    draw=blue!70!black,
    fill=blue!8,
    minimum size=8mm,
    inner sep=1pt,
    font=\scriptsize
]

\tikzstyle{clausechild}=[
    circle,
    draw=red!65!black,
    fill=red!10,
    minimum size=7mm,
    inner sep=1pt,
    font=\scriptsize
]

\tikzstyle{testchild}=[
    circle,
    draw=orange!80!black,
    fill=orange!12,
    minimum size=7mm,
    inner sep=1pt,
    font=\scriptsize
]

\tikzstyle{resetchild}=[
    circle,
    draw=green!60!black,
    fill=green!10,
    minimum size=7mm,
    inner sep=1pt,
    font=\scriptsize
]

\tikzstyle{daycare}=[
    circle,
    draw=orange!80!black,
    fill=orange!15,
    minimum size=8mm,
    inner sep=0pt,
    font=\scriptsize
]

\tikzstyle{testdaycare}=[
    circle,
    draw=purple!70!black,
    fill=purple!10,
    minimum size=8mm,
    inner sep=0pt,
    font=\scriptsize
]

\tikzstyle{parking}=[
    circle,
    draw=gray!70!black,
    fill=gray!10,
    minimum size=8mm,
    inner sep=0pt,
    font=\scriptsize
]

% Initial assignment
\tikzstyle{initial}=[
    ->,
    thick,
    dashed,
    draw=gray!70!black
]

% Destination determined under A_x^b
\tikzstyle{possible}=[
    ->,
    thick,
    draw=black!70
]

% A_x^0 for state child
\tikzstyle{zeroarrow}=[
    ->,
    very thick,
    draw=blue!70!black
]

% A_x^1 for state child
\tikzstyle{onearrow}=[
    ->,
    very thick,
    draw=red!70!black
]

% Contracts in S_i and T_tau
\tikzstyle{assignment}=[
    ->,
    thick,
    draw=purple!75
]

% ============================================================
% PANEL (a): VARIABLE FAMILY V_x
% ============================================================

\node[
    draw=gray!55,
    rounded corners,
    minimum width=16.0cm,
    minimum height=10.1cm
] at (0,2.8) {};

\node[
    fill=green!10,
    draw=green!35,
    rounded corners,
    minimum width=16cm,
    minimum height=10mm,
    font=\bfseries
] at (0,7.35)
{(a) Variable family \(V_x\): children and truth-value assignments};

% ============================================================
% Column headings
% ============================================================

% Role column stays unchanged
\node[font=\small\bfseries] at (-6.10,6.60)
{Role};

% Child and daycare columns moved slightly right
\node[font=\small\bfseries] at (-2.55,6.60)
{Child};

\node[font=\small\bfseries] at (0.35,6.60)
{Possible daycares};

\node[font=\small\bfseries] at (5.25,6.60)
{Assignment rule};

% ============================================================
% STATE CHILD
% ============================================================

% Role stays in its original position
\node[
    anchor=east,
    font=\scriptsize,
    align=center,
    text width=3.9cm
] at (-3.85,5.35)
{
\textbf{state child}\\
records the truth-value state of \(x\)
};

\node[statechild] (sx) at (-2.55,5.35)
{\(s_x\)};

% ------------------------------------------------------------
% State daycares
% ------------------------------------------------------------

\node[daycare] (ox) at (0.35,6.05)
{\(o_x\)};

\node[daycare] (zx) at (0.35,5.3)
{\(z_x\)};

\node[daycare] (rx) at (0.35,4.55)
{\(r_x\)};

\node[anchor=west,font=\scriptsize] at (0.90,6.05)
{\(Q=1\)};

\node[anchor=west,font=\scriptsize] at (0.90,5.3)
{\(Q=1\)};

\node[anchor=west,font=\scriptsize] at (0.90,4.55)
{\(Q=1\)};

\draw[initial] (sx) to[bend left=8] (ox);
\draw[zeroarrow] (sx) -- (zx);
\draw[onearrow] (sx) to[bend right=8] (rx);

\node[
    anchor=west,
    font=\scriptsize,
    align=left
] at (3.70,5.35)
{
\(\omega_{V_x}: s_x\mapsto o_x\)\\
\(A_x^0: s_x\mapsto z_x\)\\
\(A_x^1: s_x\mapsto r_x\)
};

% ============================================================
% CLAUSE-BLOCKER CHILD
% ============================================================

\node[
    anchor=east,
    font=\scriptsize,
    align=center,
    text width=3.9cm
] at (-3.85,3.25)
{
\textbf{clause-blocker child}\\
one for every clause \(C_i\) containing \(x\);\\
encodes whether \(\ell_{i,x}\) is false
};

\node[blockchild] (ax) at (-2.55,3.25)
{\(a_{x,i}\)};

\node[daycare] (di) at (0.35,3.60)
{\(d_i\)};

\node[parking] (bard) at (0.35,2.90)
{\(\bar d_{x,i}\)};

\node[anchor=west,font=\scriptsize] at (0.90,3.60)
{\(Q=k_i\)};

\node[anchor=west,font=\scriptsize] at (0.90,2.90)
{\(Q=1\)};

\draw[initial] (ax) to[bend left=5] (di);
\draw[possible] (ax) to[bend left=2] (di);
\draw[possible] (ax) to[bend right=6] (bard);

\node[
    anchor=west,
    font=\scriptsize,
    align=left
] at (3.70,3.25)
{
under \(A_x^b\):\\
\(a_{x,i}\mapsto d_i\) if \(\ell_{i,x}(b)=0\);\\
\(a_{x,i}\mapsto\bar d_{x,i}\) if \(\ell_{i,x}(b)=1\)
};

% ============================================================
% PIVOT-BLOCKER CHILD
% ============================================================

\node[
    anchor=east,
    font=\scriptsize,
    align=center,
    text width=3.9cm
] at (-3.85,1.60)
{
\textbf{pivot-blocker child}\\
for \(\tau=(i,y)\), where \(y\) is the pivot\\
and \(x\neq y\); tests a non-pivot literal
};

\node[blockchild] (px) at (-2.55,1.60)
{\(p_{x,\tau}\)};

\node[testdaycare] (pt) at (0.35,1.85)
{\(p_\tau\)};

\node[parking] (barp) at (0.35,1.15)
{\(\bar p_{x,\tau}\)};

\node[anchor=west,font=\scriptsize] at (0.90,1.85)
{\(Q=k_i-1\)};

\node[anchor=west,font=\scriptsize] at (0.90,1.15)
{\(Q=1\)};

\draw[initial] (px) to[bend left=5] (pt);
\draw[possible] (px) to[bend left=2] (pt);
\draw[possible] (px) to[bend right=6] (barp);

\node[
    anchor=west,
    font=\scriptsize,
    align=left
] at (3.70,1.50)
{
under \(A_x^b\):\\
\(p_{x,\tau}\mapsto p_\tau\) if \(\ell_{i,x}(b)=1\);\\
\(p_{x,\tau}\mapsto\bar p_{x,\tau}\) if \(\ell_{i,x}(b)=0\)
};

% ============================================================
% PRESERVATION-BLOCKER CHILD
% ============================================================

\node[
    anchor=east,
    font=\scriptsize,
    align=center,
    text width=3.9cm
] at (-3.85,-0.3)
{
\textbf{preservation-blocker child}\\
for an earlier clause \(C_h\), \(h\in H_\tau\);\\
checks whether \(C_h\) remains satisfied
};

\node[blockchild] (ex) at (-2.55,-0.3)
{\(e_{x,\tau,h}\)};

\node[testdaycare] (et) at (0.35,0.05)
{\(e_{\tau,h}\)};

\node[parking] (bare) at (0.35,-0.65)
{\(\bar e_{x,\tau,h}\)};

\node[anchor=west,font=\scriptsize] at (0.90,0.05)
{\(Q=k_h-1\)};

\node[anchor=west,font=\scriptsize] at (0.90,-0.65)
{\(Q=1\)};

\draw[initial] (ex) to[bend left=5] (et);
\draw[possible] (ex) to[bend left=2] (et);
\draw[possible] (ex) to[bend right=6] (bare);

\node[
    anchor=west,
    font=\scriptsize,
    align=left
] at (3.70,-0.3)
{
under \(A_x^b\):\\
\(e_{x,\tau,h}\mapsto e_{\tau,h}\) if \(\ell_{h,x}(b)=0\);\\
\(e_{x,\tau,h}\mapsto\bar e_{x,\tau,h}\) if \(\ell_{h,x}(b)=1\)
};

% ============================================================
% Legend for Panel (a)
% Two rows
% ============================================================

% ------------------------------------------------------------
% Row 1: initial / literal-dependent destination
% ------------------------------------------------------------

\draw[initial] (-7.30,-1.5) -- (-6.45,-1.5);

\node[
    anchor=west,
    font=\scriptsize
] at (-6.30,-1.5)
{initial assignment};

\draw[possible] (-2.70,-1.5) -- (-1.85,-1.5);

\node[
    anchor=west,
    font=\scriptsize
] at (-1.70,-1.5)
{destination determined by the literal under \(A_x^b\)};

% ------------------------------------------------------------
% Row 2: A_x^0 / A_x^1
% ------------------------------------------------------------

\draw[zeroarrow] (-7.30,-2.0) -- (-6.45,-2.0);

\node[
    anchor=west,
    font=\scriptsize,
    text=blue!70!black
] at (-6.30,-2.0)
{\(A_x^0\) for the state child};

\draw[onearrow] (-2.70,-2.0) -- (-1.85,-2.0);

\node[
    anchor=west,
    font=\scriptsize,
    text=red!70!black
] at (-1.70,-2.0)
{\(A_x^1\) for the state child};

% ============================================================
% PANEL (b): CLAUSE FAMILY F_i
% Entire panel shifted upward by 0.3 cm
% ============================================================

\begin{scope}[yshift=3mm]

\node[
    draw=gray!55,
    rounded corners,
    minimum width=16.0cm,
    minimum height=7.0cm
] at (0,-6.45) {};

\node[
    fill=purple!10,
    draw=purple!35,
    rounded corners,
    minimum width=15.9cm,
    minimum height=9mm,
    font=\bfseries
] at (0,-3.45)
{(b) Clause family \(F_i\): satisfied and transition assignments};

% ============================================================
% LEFT: SATISFIED ASSIGNMENT S_i
% ============================================================

\node[
    draw=red!45,
    fill=red!3,
    rounded corners,
    minimum width=5.6cm,
    minimum height=5.2cm
] at (-5.10,-6.55) {};

\node[
    fill=red!10,
    draw=red!35,
    rounded corners,
    minimum width=5.6cm,
    minimum height=7mm,
    font=\small\bfseries
] at (-5.10,-4.40)
{Satisfied assignment \(S_i\)};

\node[
    font=\scriptsize
] at (-5.10,-5.10)
{two satisfied-assignment children};

% ------------------------------------------------------------
% u_i^d
% ------------------------------------------------------------

\node[clausechild] (uid) at (-7.10,-5.90)
{\(u_i^d\)};

\node[daycare] (diS) at (-5.30,-5.90)
{\(d_i\)};

\draw[assignment] (uid) -- (diS);

\node[
    anchor=west,
    font=\scriptsize
] at (-4.65,-5.90)
{clause daycare};

\node[
    anchor=west,
    font=\scriptsize
] at (-4.65,-6.22)
{\(Q=k_i\)};

% ------------------------------------------------------------
% u_i^m
% ------------------------------------------------------------

\node[clausechild] (uim) at (-7.10,-7.25)
{\(u_i^m\)};

\node[daycare] (miS) at (-5.30,-7.25)
{\(m_i\)};

\draw[assignment] (uim) -- (miS);

\node[
    anchor=west,
    font=\scriptsize
] at (-4.65,-7.25)
{marker daycare};

\node[
    anchor=west,
    font=\scriptsize
] at (-4.65,-7.57)
{\(Q=1\)};

\node[
    draw=red!35,
    fill=red!5,
    rounded corners,
    minimum width=4.95cm,
    minimum height=8mm,
    font=\small
] at (-5.10,-8.60)
{\(S_i=\{(u_i^d,d_i),(u_i^m,m_i)\}\)};

\node[
    font=\scriptsize,
    align=center
] at (-5.10,-9.55)
{\(S_i\) is feasible exactly when \(C_i\) is satisfied.};

% ============================================================
% RIGHT: TRANSITION ASSIGNMENT T_tau
% Background unchanged relative to the panel
% ============================================================

\node[
    draw=orange!60!black,
    fill=orange!3,
    rounded corners,
    minimum width=9.8cm,
    minimum height=5.2cm
] at (3.05,-6.55) {};

\node[
    fill=orange!15,
    draw=orange!45!black,
    rounded corners,
    minimum width=9.8cm,
    minimum height=7mm,
    font=\small\bfseries
] at (3.05,-4.40)
{Transition assignment \(T_\tau\) with \(\tau=(i,x)\)};

\node[
    font=\scriptsize
] at (3.05,-5.00)
{\(x\) is the pivot variable to be flipped};

% ============================================================
% Transition rows:
% role + number | child | daycare | capacity
% ============================================================

% ------------------------------------------------------------
% q row
% ------------------------------------------------------------

\node[
    anchor=east,
    font=\scriptsize,
    align=right,
    text width=3.9cm
] at (2.00,-5.55)
{1 child; global transition lock};

\node[clausechild] (tq) at (2.80,-5.55)
{\(t_\tau^q\)};

\node[testdaycare] (qT) at (4.60,-5.55)
{\(q\)};

\draw[assignment] (tq) -- (qT);

\node[
    anchor=west,
    font=\scriptsize
] at (5.25,-5.55)
{\(Q=1\)};

% ------------------------------------------------------------
% m_j row
% ------------------------------------------------------------

\node[
    anchor=east,
    font=\scriptsize,
    align=right,
    text width=3.9cm
] at (2.00,-6.35)
{\(|J_\tau|\) children; reset less significant clauses};

\node[clausechild] (tm) at (2.80,-6.35)
{\(t_{\tau,j}^m\)};

\node[daycare] (mjT) at (4.60,-6.35)
{\(m_j\)};

\draw[assignment] (tm) -- (mjT);

\node[
    anchor=west,
    font=\scriptsize
] at (5.25,-6.35)
{\(Q=1\)};

% ------------------------------------------------------------
% p_tau row
% ------------------------------------------------------------

\node[
    anchor=east,
    font=\scriptsize,
    align=right,
    text width=3.9cm
] at (2.00,-7.15)
{$k_i$ $-$ $1$ children; pivot test};

\node[testchild] (tp) at (2.80,-7.15)
{\(t_{\tau,a}^p\)};

\node[testdaycare] (ptT) at (4.60,-7.15)
{\(p_\tau\)};

\draw[assignment] (tp) -- (ptT);

\node[
    anchor=west,
    font=\scriptsize
] at (5.25,-7.15)
{\(Q=k_i-1\)};

% ------------------------------------------------------------
% e_tau,h row
% ------------------------------------------------------------

\node[
    anchor=east,
    font=\scriptsize,
    align=right,
    text width=3.9cm
] at (2.00,-7.95)
{\(|H_\tau|\) children; preservation tests};

\node[testchild] (te) at (2.80,-7.95)
{\(t_{\tau,h}^e\)};

\node[testdaycare] (etT) at (4.60,-7.95)
{\(e_{\tau,h}\)};

\draw[assignment] (te) -- (etT);

\node[
    anchor=west,
    font=\scriptsize
] at (5.25,-7.95)
{\(Q=k_h-1\)};

% ------------------------------------------------------------
% r_x row
% ------------------------------------------------------------

\node[
    anchor=east,
    font=\scriptsize,
    align=right,
    text width=3.9cm
] at (2.00,-8.75)
{1 child if $b_\tau$ $=$ $1$; reset \(V_x\)};

\node[resetchild] (tr) at (2.80,-8.75)
{\(t_\tau^r\)};

\node[daycare] (rxT) at (4.60,-8.75)
{\(r_x\)};

\draw[assignment] (tr) -- (rxT);

\node[
    anchor=west,
    font=\scriptsize
] at (5.25,-8.75)
{\(Q=1\)};

% ============================================================
% Transition explanation
% ============================================================

\node[
    font=\scriptsize,
    align=center
] at (3.05,-9.55)
{
\(T_\tau\) represents a candidate improving flip of \(x\), with
\(C_i\) being the first clause whose truth value improves.
};

\end{scope}

% ============================================================
% Caption
% ============================================================

\end{tikzpicture}

\caption{
Illustration of the construction.
Panel (a) shows the children of a variable family \(V_x\) and the
daycares involved in its initial and truth-value assignments.
Panel (b) shows the satisfied assignment \(S_i\) and a transition
assignment \(T_\tau\) of a clause family \(F_i\).
}
\label{fig:fmr-construction}

\end{figure*}

\paragraph{Overview.}

For every variable $x$, we create a \emph{variable family} $V_x$, whose assignment encodes the truth value of $x$. For every clause $C_i$, we
create a \emph{clause family} $F_i$, whose assignment represents the
status of $C_i$ and is also used to implement candidate local
improvements.
The strict master list over families is
\[
F_1\rhd F_2\rhd\cdots\rhd F_m
\rhd (V_x:x\text{ is a variable}),
\]
where the variable families are placed in an arbitrary strict order.

Figure~\ref{fig:fmr-construction} summarizes the correspondence between logical objects and the families, children, and daycares used in the reduction.

%----------------------------------------------------------
% %  paragraph : Daycares
%----------------------------------------------------------
\paragraph{Daycares.}

We create the following daycares before defining any family assignment.

\begin{enumerate}[label=(D\arabic*)]

\item Create one \emph{global transition daycare} $q$ with $Q_q=1$,
      which ensures that at most one transition assignment can be active
      at a time.

\item For every clause $C_i$, create a \emph{clause daycare} $d_i$ with
      $Q_{d_i}=k_i$ and a \emph{marker daycare} $m_i$ with $Q_{m_i}=1$.
      The former indicates whether $C_i$ is satisfied, while the latter
      controls the assignment of the corresponding clause family $F_i$.

\item For every admissible $\tau=(i,x)$ with $k_i\ge2$, create a
      \emph{pivot-test daycare} $p_\tau$ with
      $Q_{p_\tau}=k_i-1$. It checks whether all literals of $C_i$ other
      than the pivot literal $\ell_{i,x}$ are false.

\item For every admissible $\tau=(i,x)$ and every $h\in H_\tau$, create a
      \emph{preservation-test daycare} $e_{\tau,h}$ with
      $Q_{e_{\tau,h}}=k_h-1$. Admissibility guarantees $k_h\ge2$.
      It checks whether the earlier clause $C_h$ can remain satisfied
      when $x$ is flipped.

\item For every variable $x$, create three private state daycares
      $o_x,z_x,r_x$, each of capacity one. They encode and control the
      state of the variable family $V_x$.

\item Every blocker child introduced below has a private capacity-one
      parking daycare. These are denoted by $\bar d_{x,i}$,
      $\bar p_{x,\tau}$, and $\bar e_{x,\tau,h}$, according to the
      corresponding shared daycare from which the child may withdraw.

\end{enumerate}

We refer to $d_i$, $p_\tau$, and $e_{\tau,h}$ as \emph{shared
daycares}, since they may simultaneously contain blocker children from
different variable families and contracts belonging to clause families.
In contrast, each parking daycare is private to a single blocker child.

The priority list of every daycare follows the family master list $\rhd$, except that initially enrolled children are moved above all non-incumbent children at that daycare. Therefore, all daycare priorities are master-list consistent up to initial enrollments.

%----------------------------------------------------------
% %  paragraph : Variable families
%----------------------------------------------------------
\paragraph{Variable families.}

For every variable $x$, the family $V_x$ contains one \emph{state child}
$s_x$ and several types of \emph{blocker children}. 
\zhnew{
The state child records the truth-value state of $x$, while the blocker children translate the truth values of literals involving $x$ into occupancies of the corresponding shared daycares. All blocker children are initially enrolled at their corresponding shared daycares.
}
Formally, the blocker children of $V_x$ are defined as follows.

\begin{enumerate}[label=(V\arabic*)]

\item For every clause $C_i$ containing $x$, create a
      \emph{clause-blocker child} $a_{x,i}$, initially enrolled at $d_i$,
      with private parking daycare $\bar d_{x,i}$.
      \zhnew{
      Its role is to encode the truth value of the literal $\ell_{i,x}$:
      it occupies $d_i$ when $\ell_{i,x}$ is false and moves to its
      private parking daycare when $\ell_{i,x}$ is true. Consequently,
      the number of clause blockers occupying $d_i$ equals the number of
      false literals in $C_i$.
        }
        
\item For every admissible occurrence $\tau=(i,y)$ such that
      $x\neq y$ and $x$ appears in $C_i$, create a
      \emph{pivot-blocker child} $p_{x,\tau}$, initially enrolled at
      $p_\tau$, with private parking daycare $\bar p_{x,\tau}$.
      \zhnew{      
      Here $y$ is the pivot variable to be flipped, while $x$ corresponds
      to one of the other literals of $C_i$.
      The child $p_{x,\tau}$ occupies $p_\tau$ exactly when
      $\ell_{i,x}$ is true. Thus, the absence of all pivot blockers from
      $p_\tau$ indicates that all literals of $C_i$ other than the pivot
      literal $\ell_{i,y}$ are false.
        }
        
\item For every admissible occurrence $\tau=(i,y)$, every earlier clause
      $C_h$ with $h\in H_\tau$, and every variable $x\neq y$ appearing
      in $C_h$, create a \emph{preservation-blocker child}
      $e_{x,\tau,h}$, initially enrolled at $e_{\tau,h}$, with private
      parking daycare $\bar e_{x,\tau,h}$.
      \zhnew{          
      Its role is to determine whether $C_h$ has a true literal other
      than its $y$-literal. The child $e_{x,\tau,h}$ occupies
      $e_{\tau,h}$ exactly when $\ell_{h,x}$ is false. Hence, if at
      least one preservation blocker is absent from $e_{\tau,h}$, then
      $C_h$ has another true literal and can remain satisfied when $y$
      is flipped.
        }
\end{enumerate}

The state child $s_x$ is initially enrolled at $o_x$. Under the two
truth-value assignments defined below, it is assigned to $z_x$ when
$x=0$ and to $r_x$ when $x=1$. Thus, the initial assignment
$\omega_{V_x}$ keeps $s_x$ at $o_x$ and every blocker child at its
initial shared daycare.

For $b\in\{0,1\}$, the assignment $A_x^b$ represents the truth value
$x=b$ and simultaneously assigns all children of $V_x$ as follows:
\begin{align*}
s_x&\mapsto
\begin{cases}
z_x,&b=0,\\
r_x,&b=1;
\end{cases}\\
a_{x,i}&\mapsto
\begin{cases}
d_i,&\ell_{i,x}(b)=0,\\
\bar d_{x,i},&\ell_{i,x}(b)=1;
\end{cases}\\
p_{x,\tau}&\mapsto
\begin{cases}
p_\tau,&\ell_{i,x}(b)=1,\quad \tau=(i,y),\\
\bar p_{x,\tau},&\ell_{i,x}(b)=0;
\end{cases}\\
e_{x,\tau,h}&\mapsto
\begin{cases}
e_{\tau,h},&\ell_{h,x}(b)=0,\\
\bar e_{x,\tau,h},&\ell_{h,x}(b)=1.
\end{cases}
\end{align*}

Thus, a single joint assignment $A_x^b$ simultaneously records the truth
value of $x$ and updates every gadget involving $x$: clause blockers
encode false literals, pivot blockers encode true non-pivot literals,
and preservation blockers encode false non-pivot literals of earlier
clauses. The preference ordering of $V_x$ is
\[
A_x^1\succ_{V_x}A_x^0\succ_{V_x}\omega_{V_x}.
\]

%----------------------------------------------------------
% %  paragraph : Clause families
%----------------------------------------------------------

\paragraph{Clause families.}

For every clause $C_i$, the family $F_i$ has no initial assignment, i.e., $\omega_{F_i}=\varnothing.$
It contains two types of children: two \emph{satisfied-assignment children}, which are used when $C_i$ is represented as satisfied, and a separate collection of \emph{transition children} for each admissible
occurrence $\tau=(i,x)$, which are used to represent a candidate local move obtained by flipping the pivot variable $x$.

\medskip
\noindent
\emph{Satisfied-assignment children.}
The family $F_i$ has exactly two satisfied-assignment children,
$u_i^d$ and $u_i^m$. Their joint assignment is
\[
S_i:=\{(u_i^d,d_i),(u_i^m,m_i)\}.
\]
The child $u_i^d$ occupies one seat at the clause daycare $d_i$, while
$u_i^m$ occupies the marker daycare $m_i$.

\zhnew{
Under the truth-value assignments of the variable families, $d_i$
contains exactly one clause blocker for every false literal of $C_i$.
Since $Q_{d_i}=k_i$, the contract $(u_i^d,d_i)$ can be accommodated
exactly when at most $k_i-1$ literals of $C_i$ are false. Hence,
$S_i$ is feasible exactly when $C_i$ is satisfied.
}

\medskip
\noindent
\emph{Transition children.}
Fix an admissible occurrence $\tau=(i,x)$, where $x$ is the pivot
variable to be flipped. For this $\tau$, create a distinct collection of
transition children consisting of:

\begin{enumerate}[label=(T\arabic*)]

\item one \emph{global-transition child} $t_\tau^q$, assigned to $q$.
      Its role is to ensure that at most one transition assignment can be
      active at a time;

\item one \emph{marker child} $t_{\tau,j}^m$ for every
      $j\in J_\tau$, assigned to $m_j$.
      Thus, there are $|J_\tau|$ such children. Their role is to prevent
      the less significant clause families $F_j$ affected by the flip of
      $x$ from simultaneously using their satisfied assignments $S_j$;

\item $k_i-1$ \emph{pivot-test children}
      $t_{\tau,a}^p$, $a\in[k_i-1]$, all assigned to $p_\tau$.
      Their role is to test whether all literals of $C_i$ other than the
      pivot literal $\ell_{i,x}$ are false. This part is omitted when
      $k_i=1$;

\item one \emph{preservation-test child} $t_{\tau,h}^e$ for every
      $h\in H_\tau$, assigned to $e_{\tau,h}$.
      Thus, there are $|H_\tau|$ such children. Their role is to test
      whether every more significant clause $C_h$ containing $x$ has
      another true literal and can therefore remain satisfied after
      $x$ is flipped;

\item when $b_\tau=1$, one additional \emph{reset child} $t_\tau^r$,
      assigned to $r_x$. Its role is to interact with the state child of
      $V_x$ so that $V_x$ can be reset through resettlement. No reset
      child is created when $b_\tau=0$.

\end{enumerate}

Thus, for a fixed admissible $\tau=(i,x)$, the transition assignment
contains
\[
1+|J_\tau|+(k_i-1)+|H_\tau|
+\mathbf{1}[b_\tau=1]
\]
children, where the term $k_i-1$ is zero when $k_i=1$. These children
form the joint assignment
\begin{align}
T_\tau:=
&\ \{(t_\tau^q,q)\}
\cup\{(t_{\tau,j}^m,m_j):j\in J_\tau\}
\label{eq:fmr:transition-assignment}\\
&\ \cup\{(t_{\tau,a}^p,p_\tau):a\in[k_i-1]\}
\cup\{(t_{\tau,h}^e,e_{\tau,h}):h\in H_\tau\}
\nonumber\\
&\ \cup
\begin{cases}
\{(t_\tau^r,r_x)\},&b_\tau=1,\\
\varnothing,&b_\tau=0.
\end{cases}
\nonumber
\end{align}
where the pivot-test part is omitted when $k_i=1$.

The assignment $T_\tau$ represents a candidate local improvement in
which $x$ is flipped and $C_i$ is the first clause whose truth value
improves. Transition children belonging to transition assignments other
than the assignment currently chosen by $F_i$ are unmatched.

The preference ordering of $F_i$ is
\[
S_i\succ_{F_i}
\bigl(T_\tau:\tau=(i,x)\in\mathcal A\bigr)
\succ_{F_i}\varnothing,
\]
where the transition assignments are placed in an arbitrary strict
order.

\subsubsection{Correctness}

% The initial matching is feasible.  Indeed, every $d_i$ contains exactly its
% $k_i$ of its clause-blocker incumbents, every $p_\tau$ contains exactly its
% $k_i-1$ pivot-blocker incumbents, every $e_{\tau,h}$ contains exactly its
% $k_h-1$ preservation-blocker incumbents, and each state child $s_x$ is alone at
% $o_x$.  All other daycares are empty.

% Call a feasible and individually rational matching \emph{terminal} if it
% admits neither a non-wastefulness improvement nor an FMR improvement.

The initial matching is feasible.  For every clause $C_i$, the clause
daycare $d_i$ initially contains exactly one clause-blocker incumbent for
each literal of $C_i$, and hence exactly $k_i$ incumbents, matching its
capacity $Q_{d_i}=k_i$.  For every admissible $\tau=(i,x)$, the
pivot-test daycare $p_\tau$ contains exactly one pivot-blocker incumbent
for each non-pivot literal of $C_i$, and hence exactly $k_i-1$
incumbents, matching $Q_{p_\tau}=k_i-1$.  Similarly, for every
$h\in H_\tau$, the preservation-test daycare $e_{\tau,h}$ contains
exactly one preservation-blocker incumbent for each literal of $C_h$
other than the pivot literal, and hence exactly $k_h-1$ incumbents,
matching $Q_{e_{\tau,h}}=k_h-1$.  Finally, each state child $s_x$ is
initially enrolled alone at $o_x$, whose capacity is one.  All other
daycares are initially empty.  Therefore, the initial matching is
feasible.

Call a feasible and individually rational matching \emph{terminal} if it
admits neither a non-wastefulness improvement nor an FMR improvement.
We first characterize the structure of terminal matchings.  The next
lemma shows that every variable family must encode a Boolean value, while
a transition assignment cannot persist in a terminal matching.

\begin{lemma}
\label{lem:fmr:variables-and-transitions-consistent}
In every terminal matching, each variable family $V_x$ is assigned either
$A_x^0$ or $A_x^1$, and no clause family is assigned a transition assignment.
\end{lemma}

\begin{proof}
We first show that every variable family in a terminal matching must
encode a Boolean value. Suppose that $V_x$ is assigned
$\omega_{V_x}$. Moving from $\omega_{V_x}$ to $A_x^0$ moves the state
child $s_x$ from $o_x$ to the private state daycare $z_x$ and moves
some blocker children from their initially enrolled shared daycares to
their private parking daycares. No new contract is added to any shared
daycare. Hence
\[
\mathrm{Re}(Y,V_x,A_x^0)=\varnothing.
\]
Since $A_x^0\succ_{V_x}\omega_{V_x}$, this is a non-wastefulness
improvement, contradicting terminality. Therefore,
\[
Y_{V_x}\in\{A_x^0,A_x^1\}
\]
for every variable $x$.

We next show that no clause family can use a transition assignment in a
terminal matching. Suppose, for contradiction, that
\[
Y_{F_i}=T_\tau
\qquad\text{for some }\tau=(i,x).
\]

We first record several consequences of feasibility. Since $T_\tau$
contains $k_i-1$ contracts at $p_\tau$ and
\[
Q_{p_\tau}=k_i-1,
\]
no pivot-blocker child can remain at $p_\tau$. A pivot blocker is present
at $p_\tau$ exactly when the corresponding non-pivot literal of $C_i$
is true. Hence every literal of $C_i$ other than the pivot literal
$\ell_{i,x}$ is false.

Similarly, for every $h\in H_\tau$, the assignment $T_\tau$ contains
one contract at $e_{\tau,h}$, whose capacity is
\[
Q_{e_{\tau,h}}=k_h-1.
\]
Thus, at most $k_h-2$ preservation blockers can remain there. Since
there is one preservation blocker for every literal of $C_h$ other than
its $x$-literal, and such a blocker is present exactly when that literal
is false, at least one other literal of $C_h$ must be true.

Finally, the contract of $T_\tau$ at the unit-capacity daycare $q$
precludes every other clause family from using a transition assignment.
Moreover, for every $j\in J_\tau$, $T_\tau$ occupies $m_j$, so $F_j$
cannot use $S_j$. Hence every such $F_j$ must use its empty initial
assignment.

We now distinguish according to $b_\tau$.

\medskip
\noindent
\textbf{Case 1: $b_\tau=1$.}
Then $\ell_{i,x}$ is false when $x=1$, and hence
\[
\ell_{i,x}=\neg x.
\]
The transition assignment $T_\tau$ contains a contract at $r_x$.
Since $Q_{r_x}=1$ and $A_x^1$ also assigns the state child $s_x$ to
$r_x$, feasibility implies that $V_x$ cannot use $A_x^1$. By the first
part of the proof,
\[
Y_{V_x}=A_x^0.
\]
Thus $x=0$, so the pivot literal $\neg x$ is true.

We have already shown that every other literal of $C_i$ is false.
Consequently, $d_i$ contains exactly $k_i-1$ clause blockers, one for
each non-pivot literal, while the blocker corresponding to
$\ell_{i,x}$ is absent. Since $Q_{d_i}=k_i$, the contract
$(u_i^d,d_i)$ of $S_i$ can be added without rejecting any contract.

The marker daycare $m_i$ is also vacant. Family $F_i$ currently uses
$T_\tau$ rather than $S_i$, and the global lock $q$ ensures that no
other transition assignment is active. Therefore, after replacing
$T_\tau$ by $S_i$,
\[
\mathrm{Re}(Y,F_i,S_i)=\varnothing.
\]
Since
\[
S_i\succ_{F_i}T_\tau,
\]
this is a non-wastefulness improvement, contradicting terminality.

\medskip
\noindent
\textbf{Case 2: $b_\tau=0$.}
In this case the pivot literal is
\[
\ell_{i,x}=x.
\]

First suppose that $Y_{V_x}=A_x^1$. Then the pivot literal is true,
whereas every other literal of $C_i$ is false. Hence $d_i$ contains
exactly $k_i-1$ clause blockers and has one vacant seat. As in Case~1,
$m_i$ is vacant. Therefore $F_i$ can replace $T_\tau$ by $S_i$ without
rejecting any contract, again yielding a non-wastefulness improvement.

It remains to consider
\[
Y_{V_x}=A_x^0.
\]
We show that $V_x$ can move from $A_x^0$ to its preferred assignment
$A_x^1$ without rejecting any contract.

First, the state child $s_x$ moves from $z_x$ to $r_x$. Since
$b_\tau=0$, the active transition $T_\tau$ has no contract at $r_x$.
Moreover, the global lock $q$ excludes every other transition
assignment. Hence $r_x$ is vacant.

Next consider the pivot and preservation test daycares used by
$T_\tau$. Since $x$ itself is the pivot variable, the pivot daycare
$p_\tau$ has blocker children only for variables in
\[
\mathrm{var}(C_i)\setminus\{x\},
\]
and each preservation daycare $e_{\tau,h}$ has blocker children only for
variables in
\[
\mathrm{var}(C_h)\setminus\{x\}.
\]
Therefore no blocker child of $V_x$ is assigned to any test daycare
used by the active transition $T_\tau$. Changing $V_x$ from $A_x^0$
to $A_x^1$ thus does not affect the feasibility of these active tests.

Every other pivot or preservation daycare contains no transition
contract, because $q$ allows only one active transition. Such a daycare
contains only its blocker children, whose maximum possible number equals
its capacity. Therefore moving blocker children of $V_x$ between these
shared daycares and their private parking daycares preserves feasibility.

It remains to check the clause daycares. Consider first an earlier
clause $C_h$, $h<i$, containing $x$. The preservation test
$e_{\tau,h}$ implies that $C_h$ has some true literal other than its
$x$-literal. Since only $x$ changes, this other literal remains true
after changing $x$ from $0$ to $1$. Thus, after the change, at most
$k_h-1$ false-literal blockers occupy $d_h$. If $F_h$ uses $S_h$, these
blockers together with $(u_h^d,d_h)$ fit within
$Q_{d_h}=k_h$; if $F_h$ uses its empty initial assignment, feasibility
is immediate.

At $d_i$, changing $x$ from $0$ to $1$ makes the pivot literal
$\ell_{i,x}=x$ true, so the blocker $(a_{x,i},d_i)$ is removed.
Moreover, $F_i$ has no contract at $d_i$ while it uses $T_\tau$.

Finally, consider a later clause $C_j$, $j>i$, containing $x$. Since
$j\in J_\tau$, the transition assignment $T_\tau$ occupies $m_j$.
The global lock excludes a transition assignment of $F_j$, while the
occupation of $m_j$ excludes $S_j$. Hence $F_j$ uses its empty initial
assignment, and $d_j$ contains only variable blockers, of which there
are at most
\[
k_j=Q_{d_j}.
\]
Clauses not containing $x$ are unaffected.

All remaining changes move blocker children only between their shared
daycares and their private capacity-one parking daycares, each of which
is private to that blocker child. Consequently,
\[
\mathrm{Re}(Y,V_x,A_x^1)=\varnothing.
\]
Since
\[
A_x^1\succ_{V_x}A_x^0,
\]
this is a non-wastefulness improvement, again contradicting terminality.

Thus no transition assignment $T_\tau$ can occur in a terminal matching.
\end{proof}

We next show that, in a terminal matching, the assignment of each clause
family exactly reflects the truth value of the corresponding clause.

\begin{lemma}
\label{lem:fmr:clause-consistency-consistent}
Let $Y$ be terminal, and define the truth assignment $\alpha_Y$ by
\[
\alpha_Y(x)=b
\quad\Longleftrightarrow\quad
Y_{V_x}=A_x^b.
\]
Then, for every clause $C_i$,
\[
Y_{F_i}=S_i
\quad\Longleftrightarrow\quad
\alpha_Y(C_i)=1.
\]
Equivalently, $Y_{F_i}=\varnothing$ if and only if
$\alpha_Y(C_i)=0$.
\end{lemma}

\begin{proof}
By Lemma~\ref{lem:fmr:variables-and-transitions-consistent}, every
variable family uses one of its two truth-value assignments, and no
clause family uses a transition assignment. Hence each clause family
$F_i$ can use only $S_i$ or its empty initial assignment
$\omega_{F_i}=\varnothing$.

Under the truth assignment $\alpha_Y$, the clause daycare $d_i$
contains exactly one clause-blocker child for each false literal of
$C_i$. Thus, if
\[
\False_{\alpha_Y}(C_i)
\]
denotes the number of false literals in $C_i$, then $d_i$ contains
exactly $\False_{\alpha_Y}(C_i)$ blockers.

Suppose first that $Y_{F_i}=S_i$. Since $S_i$ contains the contract
$(u_i^d,d_i)$ and $Q_{d_i}=k_i$, feasibility requires
\[
\False_{\alpha_Y}(C_i)+1\le k_i.
\]
Hence at most $k_i-1$ literals of $C_i$ are false, so at least one
literal is true. Therefore,
\[
\alpha_Y(C_i)=1.
\]

Conversely, suppose that $\alpha_Y(C_i)=1$. Then at most $k_i-1$
literals of $C_i$ are false, so $d_i$ contains at most $k_i-1$
clause blockers and has sufficient capacity for $(u_i^d,d_i)$.

If $F_i$ were assigned its empty initial assignment, then $m_i$ would
also be vacant. Indeed, $F_i$ is not using $S_i$, and by
Lemma~\ref{lem:fmr:variables-and-transitions-consistent} no clause family
uses a transition assignment, so no transition contract occupies $m_i$.
Therefore, $F_i$ could move from $\varnothing$ to
\[
S_i=\{(u_i^d,d_i),(u_i^m,m_i)\}
\]
without rejecting any contract:
\[
\mathrm{Re}(Y,F_i,S_i)=\varnothing.
\]
Since
\[
S_i\succ_{F_i}\varnothing,
\]
this would be a non-wastefulness improvement, contradicting the
terminality of $Y$. Hence $Y_{F_i}=S_i$.

Therefore,
\[
Y_{F_i}=S_i
\quad\Longleftrightarrow\quad
\alpha_Y(C_i)=1.
\]
\end{proof}

% \begin{proof}
% By Lemma~\ref{lem:fmr:variables-and-transitions-consistent}, $F_i$ uses $S_i$ or the empty initial assignment. By Lemma~\ref{lem:fmr:variables-and-transitions-consistent}, $m_i$ can only be used by $F_i$. Under the truth-value assignments of the variable families,
% $d_i$ contains exactly one incumbent blocker for each false literal of $C_i$.
% Because $Q_{d_i}=k_i$ and the contract $(u_i^d,d_i)$ is below all these
% incumbents, feasibility of $S_i$ is equivalent to
% \[
%  \False_{\alpha_Y}(C_i)+1\le k_i,
% \]
% which holds exactly when $C_i$ is true.  Conversely, if $C_i$ is true and
% $F_i$ uses $\omega_{F_i}=\varnothing$, then $d_i$ contains at most $k_i-1$ blockers, $m_i$ is free,
% and $F_i$ can move to $S_i$ with
% $\mathrm{Re}(Y,F_i,S_i)=\varnothing$, contradicting terminality.
% \end{proof}

\begin{lemma}
\label{lem:fmr:flip-improvement-consistent}
Let $Y$ be terminal. Then no variable flip can strictly increase
$\val(\alpha_Y)$.
\end{lemma}

\begin{proof}
Suppose, for contradiction, that flipping a variable $x$ strictly
increases $\val(\alpha_Y)$. Let $C_i$ be the first clause, in the
significance order, whose truth value changes under this flip. Since the
change is lexicographically improving, $C_i$ changes from false to true.

Because only $x$ is flipped, the pivot literal $\ell_{i,x}$ is false
before the flip and becomes true after the flip. Hence
\[
\alpha_Y(x)=b_\tau,
\qquad
\tau:=(i,x),
\]
where $b_\tau$ is defined by $\ell_{i,x}(b_\tau)=0$.
Moreover, since $C_i$ is false before the flip, every other literal of
$C_i$ is also false.

Now consider any earlier clause $C_h$, $h<i$, containing $x$. Since
$C_i$ is the first clause whose truth value changes, $C_h$ must remain
unchanged after flipping $x$. Therefore, $C_h$ must contain some true
literal other than its $x$-literal. In particular, no earlier singleton
clause contains $x$, and hence $\tau=(i,x)$ is admissible.

By Lemma~\ref{lem:fmr:clause-consistency-consistent}, $C_i$ being false
implies
\[
Y_{F_i}=\omega_{F_i}=\varnothing.
\]
We show that the transition assignment $T_\tau$ gives an improvement for
$F_i$.

First consider the pivot-test daycare $p_\tau$. The assignment $T_\tau$
contains $k_i-1$ contracts at $p_\tau$, and
\[
Q_{p_\tau}=k_i-1.
\]
A pivot blocker corresponding to a non-pivot literal is present at
$p_\tau$ exactly when that literal is true. Since every non-pivot
literal of $C_i$ is false, no pivot blocker is present at $p_\tau$.
Hence all $k_i-1$ transition contracts at $p_\tau$ are selected.

Next consider a preservation-test daycare $e_{\tau,h}$ for
$h\in H_\tau$. Its capacity is $k_h-1$, and $T_\tau$ contains one
contract at $e_{\tau,h}$. Since $C_h$ has another true literal, at least
one of its $k_h-1$ preservation blockers is absent. Thus at most
$k_h-2$ preservation blockers are present, leaving a seat for the
transition contract. Hence the contract of $T_\tau$ at every
$e_{\tau,h}$ is selected.

The contract of $T_\tau$ at $q$ is also selected, because
Lemma~\ref{lem:fmr:variables-and-transitions-consistent} implies that
the terminal matching $Y$ contains no transition assignment, and hence
$q$ is vacant.

Now consider a marker daycare $m_j$ with $j\in J_\tau$. By
Lemma~\ref{lem:fmr:variables-and-transitions-consistent}, $F_j$ uses
either $S_j$ or its empty initial assignment. If $F_j$ uses $S_j$, then
its marker contract $(u_j^m,m_j)$ competes with
$(t_{\tau,j}^m,m_j)$. Since $i<j$, we have
\[
F_i\rhd F_j,
\]
so the transition contract of $F_i$ is selected and the marker contract
of $F_j$ is rejected.

Finally, suppose $b_\tau=1$. Then
\[
\alpha_Y(x)=1,
\]
so $Y_{V_x}=A_x^1$ and the state child $s_x$ occupies $r_x$.
The transition assignment $T_\tau$ also contains a contract at $r_x$.
Since $F_i\rhd V_x$ and $Q_{r_x}=1$, the transition contract is selected
and $(s_x,r_x)$ is rejected. If $b_\tau=0$, $T_\tau$ contains no
contract at $r_x$.

Consequently,
\[
T_\tau
\subseteq
Ch_D\bigl(Z(Y,F_i,T_\tau)\bigr),
\]
and the only families that can be rejected are
\[
F\bigl(\mathrm{Re}(Y,F_i,T_\tau)\bigr)
\subseteq
\{F_j:j\in J_\tau\}
\cup
\begin{cases}
\{V_x\},&b_\tau=1,\\
\varnothing,&b_\tau=0.
\end{cases}
\]
Every such family is ranked below $F_i$ in the master list.

It remains to verify that all rejected families can be feasibly
resettled to their initial assignments. Every rejected clause family
$F_j$ is returned from $S_j$ to
\[
\omega_{F_j}=\varnothing.
\]
If $b_\tau=1$, the rejected variable family $V_x$ is returned from
$A_x^1$ to $\omega_{V_x}$. We verify feasibility after these
resettlements.

Consider first the test daycares used by $T_\tau$. Since $x$ is the
pivot variable, the pivot-test daycare $p_\tau$ contains blocker children
only for variables in
\[
\mathrm{var}(C_i)\setminus\{x\},
\]
and each preservation-test daycare $e_{\tau,h}$ contains blocker children
only for variables in
\[
\mathrm{var}(C_h)\setminus\{x\}.
\]
Thus, $V_x$ has no blocker child associated with any active test daycare.
Resettling $V_x$ therefore does not alter their occupancies, and all test
contracts of $T_\tau$ remain feasible.

Any inactive pivot- or preservation-test daycare contains no transition
contract. It contains only its blocker children, and its capacity equals
the maximum possible number of such blockers. Hence returning blocker
children of $V_x$ to inactive test daycares cannot violate capacity.

Next consider an earlier clause daycare $d_h$, where $h<i$ and $C_h$
contains $x$. Resettling $V_x$ may return the clause blocker
$a_{x,h}$ to $d_h$. However, as shown above, $C_h$ has another true
literal. Hence at most $k_h-1$ clause blockers are present after
resettlement. If $F_h$ uses $S_h$, its contract $(u_h^d,d_h)$ therefore
still fits within
\[
Q_{d_h}=k_h.
\]

At the pivot clause daycare $d_i$, $F_i$ now uses $T_\tau$ rather than
$S_i$, so it has no contract at $d_i$. Even if resettling $V_x$ returns
the pivot clause blocker to $d_i$, at most all $k_i$ clause blockers are
present, which fits exactly within $Q_{d_i}=k_i$.

Finally, consider a later clause $C_j$, $j>i$, containing $x$. If $F_j$
is rejected at $m_j$, it is resettled to
$\omega_{F_j}=\varnothing$. Thus $d_j$ contains no clause-family
contract and can accommodate all of its at most $k_j$ clause blockers.
Clauses not containing $x$ are unaffected.

The marker, global-transition, state, and private parking daycares are
also feasible after resettlement. Therefore the resettled matching is
feasible.

If
\[
\mathrm{Re}(Y,F_i,T_\tau)=\varnothing,
\]
then $(F_i,T_\tau)$ is a non-wastefulness improvement. Otherwise, all
rejected families lie below $F_i$ in the master list and can be feasibly
resettled to their initial assignments, so $(F_i,T_\tau)$ is an FMR
improvement. Either case contradicts the terminality of $Y$.

Therefore no single-variable flip can strictly increase
$\val(\alpha_Y)$.
\end{proof}

% \begin{theorem}
% \label{theo:computation:FMR}
% Finding a feasible and individually rational matching that is
% non-wasteful and fair by master list and resettlement is PLS-hard.  The result
% holds even if
% \begin{enumerate}[label=(\roman*)]
% \item every daycare has capacity at most four; and
% \item daycare priorities are master-list consistent up to initial enrollments.
% \end{enumerate}
% \haris{PLS-hard or PLS-complete?}
% \end{theorem}

% \begin{proof}
% The construction is polynomial and all listed family assignments are explicit.
% Condition (ii) is also satisfied by construction. The
% largest capacities are $Q_{d_i}=k_i\le4$; all pivot and preservation tests have
% capacity at most three, and all other daycares have capacity one.

% Given any non-wasteful and FMR-fair matching $Y$, decode $\alpha_Y$ from the
% assignments of the variable families.  If $\alpha_Y$ were not locally optimal
% for $\LLMS$, a one-variable flip would strictly increase its lexicographic
% clause vector.  Lemma~\ref{lem:fmr:flip-improvement-consistent} would then give
% a non-wastefulness or FMR improvement of $Y$, a contradiction.  Hence every
% solution of the constructed matching instance maps in polynomial time to a
% local optimum of the source instance.  This is a PLS reduction from $\LLMS$.
% \end{proof}

\begin{theorem}
\label{theo:computation:FMR}
Finding a feasible and individually rational matching that is
non-wasteful and fair by master list and resettlement is PLS-complete.
The result holds even if
\begin{enumerate}[label=(\roman*)]
\item every daycare has capacity at most four; and
\item daycare priorities are master-list consistent up to initial
enrollments.
\end{enumerate}
\end{theorem}

\begin{proof}
We first show membership in PLS.  Let
\[
f_1\rhd f_2\rhd\cdots\rhd f_{|F|}
\]
be the family master list.  Since the complete individually rational
preference list of every family is explicitly given, assign to each
assignment $W$ of family $f$ a score $s_f(W)$ consistent with
$\succ_f$, with larger scores representing more preferred assignments.
Let
\[
B> \max_{f\in F} |\mathcal W_f|
\]
and define the potential of a feasible and individually rational matching
$Y$ by
\[
\Phi(Y)
=
\sum_{t=1}^{|F|}
s_{f_t}(Y_{f_t}) B^{|F|-t}.
\]
The binary representation of $\Phi(Y)$ has polynomial length.

Consider first a non-wastefulness improvement by a family $f_t$.
Only the assignment of $f_t$ changes, and it changes to a strictly
preferred assignment. Hence the lexicographic assignment profile, and
therefore $\Phi$, strictly increases.

Now consider an FMR improvement by $f_t$.  The assignment of $f_t$
strictly improves, while every family that is displaced and resettled
lies below $f_t$ in the master list.  All families above $f_t$ remain
unchanged.  Thus the first coordinate of the master-list-ordered
assignment profile that changes is that of $f_t$, and it changes
strictly for the better.  Hence $\Phi$ again strictly increases.

Moreover, given a feasible and individually rational matching $Y$, we
can find an improving neighbor or certify that none exists in polynomial
time.  For every family $f$ and every explicitly listed assignment
$W\succ_fY_f$, we compute
\[
Z(Y,f,W)
\quad\text{and}\quad
\mathrm{Re}(Y,f,W).
\]
If $\mathrm{Re}(Y,f,W)=\varnothing$, then $(f,W)$ is a
non-wastefulness improvement.  Otherwise, we check the master-list and
resettlement conditions of FMR, including feasibility after resettling
the rejected families to their initial assignments.  All these checks
are polynomial in the input size.  The initial matching $\omega$ is also
computable in polynomial time and is feasible and individually rational.
Therefore, the problem belongs to PLS.  Its local optima are exactly the
terminal matchings, i.e., the matchings that are both non-wasteful and
FMR-fair.

It remains to show PLS-hardness.  We use the reduction from \LLMS
described above.  The construction is polynomial, and all family
assignments are explicitly represented.  Condition~(ii) holds by
construction.  Moreover,
\[
Q_{d_i}=k_i\le4,
\]
every pivot-test and preservation-test daycare has capacity at most
three, and every other daycare has capacity one.  Hence
Condition~(i) also holds.

Given any non-wasteful and FMR-fair matching $Y$ of the constructed
instance, decode the truth assignment $\alpha_Y$ from the assignments of
the variable families.  If $\alpha_Y$ were not locally optimal for
\LLMS, then some single-variable flip would strictly increase
$\val(\alpha_Y)$.  By
Lemma~\ref{lem:fmr:flip-improvement-consistent}, this would induce either
a non-wastefulness improvement or an FMR improvement of $Y$,
contradicting the assumption that $Y$ is terminal.  Thus every solution
of the constructed matching instance can be mapped in polynomial time
to a locally optimal solution of the \LLMS instance.

Hence the reduction is a PLS reduction from the PLS-complete problem
\LLMS.  The matching problem is therefore PLS-hard, and together with
membership in PLS, it is PLS-complete.
\end{proof}

\paragraph{Dependence on family size.}
The reduction also preserves bounded occurrence.  A variable family has one
state child, at most $\Delta$ clause-blocker children, at most $3\Delta$
pivot-blocker children, and fewer than $3\Delta^2$ preservation-blocker
children.  Hence
\[
 |C_{V_x}|\le3\Delta^2+\Delta+1.
\]
A clause has at most four admissible transition assignments.  Using distinct
children for all their displayed contracts gives
\[
 |C_{F_i}|\le 2+4(2\Delta+3)=8\Delta+14.
\]
Therefore every family has size at most
\[
 3\Delta^2+8\Delta+14.
\]

\begin{corollary}
\label{cor:fmr:bounded-transfer-consistent}
If $\LLMS$ is PLS-hard for formulas in which every variable occurs at most a
fixed constant $\Delta_0$, then the hardness in
Theorem~\ref{theo:computation:FMR} holds with family size at
most $3\Delta_0^2+8\Delta_0+14$.
\end{corollary}

The qualification in
Corollary~\ref{cor:fmr:bounded-transfer-consistent} is necessary because
the complexity of bounded-occurrence \LLMS is currently left
open~\citep{ScTa25a}.

% \newpage
%----------------------------------------
% subsection: Dominance-based Displacement
%----------------------------------------
\subsection{Dominance-based Displacement}

We next strengthen fairness by master list and resettlement by introducing a new fairness notion based on \emph{dominance-based displacement}. Unlike the previous notion, where the families affected by a deviation are determined by the contracts rejected by the daycare choice functions, the present notion identifies potential displaced families directly through pairwise priority comparisons between the proposed contracts and the currently assigned contracts.

Given a proposed assignment $W$ of a family $f$, another family is
considered dominated if at least one of its currently assigned children has lower priority, at the same daycare, than some child of $f$ assigned there under $W$. Rather than requiring the deviation to act on the particular set of families induced by the daycare choice functions, we allow a nonempty subset of such dominated families to be displaced. Once a family is displaced, its entire current assignment is removed, even if only one of its contracts is dominated.

We first formalize the set of families that may be displaced.

\begin{definition}[Dominated Families]
\label{def:dominated_families}
Given a matching $Y$, a family $f$, and a proposed assignment $W$ for
$f$, define
\[
L_f(W,Y)
=
\left\{
g\in F\setminus\{f\}
\;\middle|\;
\exists (c,d)\in W,\;
\exists (c',d)\in Y_g
\text{ such that }
c\succ_d c'
\right\}.
\]
We refer to $L_f(W,Y)$ as the set of families \emph{dominated by $f$ through $W$ at $Y$}.
\end{definition}

Using this dominance relation, we define the following stronger fairness
notion.

\begin{definition}[Strong Fairness by Master List and Resettlement]
\label{def:SFMR}
Given an individually rational matching $Y$ and a strict master list $\rhd$ over families, family $f$ is said to have
\emph{weakly justified envy by master list and resettlement}
if there exist a proposed assignment $W$ for $f$ and a nonempty set $G\subseteq L_f(W,Y)$ such that the following conditions hold:

\begin{description}

\item[(1) Profitable deviation:] $W\succ_f Y_f.$

\item[(2) Acceptance after displacement:]
$W \subseteq Ch_D\bigl(Y\setminus(Y_f\cup Y_G)\cup W
\bigr),$ where $Y_G=\bigcup_{g\in G}Y_g.$

\item[(3) Master-list condition:]
\[f\rhd g \qquad \text{for every }g\in G.\]

\item[(4) Resettlement feasibility:]
\[Y\setminus(Y_f\cup Y_G) \cup (W\cup\omega_G)\] is feasible, where
$\omega_G=\bigcup_{g\in G}\omega_g.$
\end{description}

An individually rational matching is said to be
\emph{strongly fair by master list and resettlement} if no family
has weakly justified envy by master list and resettlement.
\end{definition}

The key difference between the two fairness notions FMR and SFMR lies in how the displaced families are determined. Under fairness by master list and
resettlement, adding $W$ to the current matching determines a particular
set of rejected contracts through the daycare choice functions, and all
families involved in these rejected contracts must satisfy the
master-list and resettlement conditions. Under strong fairness by master
list and resettlement, by contrast, the deviation may be supported by
any nonempty subset $G$ of dominated families, provided that removing
their entire assignments makes $W$ acceptable and satisfies the
master-list and resettlement requirements.

This additional flexibility strengthens the fairness requirement.
In particular, a deviation may fail to constitute justified envy under
fairness by master list and resettlement because some family induced by
the rejected contracts is ranked above $f$ in the master list or cannot
be feasibly resettled. Nevertheless, there may exist another suitable
subset of dominated families whose displacement makes $W$ acceptable
and for which both requirements are satisfied. Such a deviation is
captured by strong fairness by master list and resettlement.

Since the dominance-based notion allows a broader set of deviations to
constitute justified envy, the corresponding fairness requirement is
stronger than fairness by master list and resettlement. 
The following proposition formalizes the relationship between the two
notions.
\begin{proposition}
\label{prop:connection:sfmr-implies-fmr}
Strong fairness by master list and resettlement implies fairness by
master list and resettlement.
\end{proposition}

\begin{proof}
We prove the contrapositive. Suppose that a feasible and individually
rational matching $Y$ is not fair by master list and resettlement according
to Definition~\ref{def:fairness_resettlement}. Then there exist a family
$f$ and an assignment $W\succ_f Y_f$ such that, letting
$Z:=Y\setminus Y_f\cup W$
and $\mathrm{Re}:=Z\setminus Ch_D(Z),$
we have $\mathrm{Re}\neq\emptyset,$ and $W\subseteq Ch_D(Z),$
every family in $F(\mathrm{Re})$ is ranked below $f$ in the master list,
and $Y\setminus \bigl(Y_f\cup Y_{F(\mathrm{Re})}\bigr)
\cup \bigl(W\cup\omega_{F(\mathrm{Re})}\bigr)$ is feasible.

Set $G:=F(\mathrm{Re}).$
We first show that $G\subseteq L_f(W,Y).$ Take any family $g\in G$. By definition of $G$, there exists a rejected contract
$(c',d)\in Y_g$ such that
$(c',d)\notin Ch_d(Z_d).$
Since $W\subseteq Ch_D(Z)$, every rejected contract belongs to
$Y\setminus Y_f$.
Because $Y$ is feasible, $(c',d)$ could not become rejected after replacing $Y_f$ by $W$ unless some newly introduced contract $(c,d)\in W$ has higher priority at $d$, i.e.,
$c\succ_d c'.$
Indeed, suppose otherwise. Then every contract in $Z_d$ ranked above
$(c',d)$ would already belong to $(Y\setminus Y_f)_d$. Since $(c',d)$ is rejected by $Ch_d(Z_d)$, there must be at least $Q_d$ contracts in $Z_d$ ranked above it. These contracts, together with $(c',d)$, would therefore already belong to $Y_d$, implying $|Y_d|\geq Q_d+1,$
which contradicts the feasibility of $Y$.
Hence, there exists some $(c,d)\in W$ such that $c\succ_d c',$
and therefore $g\in L_f(W,Y).$
Since $g$ was arbitrary, we conclude that $G\subseteq L_f(W,Y).$

We now verify the conditions for weakly justified envy by master list and resettlement.
Condition~(1) holds because
$W\succ_f Y_f.$
For Condition~(2), we already have
$W\subseteq Ch_D(Z).$
The set $Y\setminus(Y_f\cup Y_G)\cup W$ is obtained from $Z$ only by removing additional contracts that do not belong to $W$. By Definition~\ref{def:choice_func}, each daycare chooses
the highest-priority contracts up to its quota. Hence, removing competing
contracts cannot cause a contract in $W$ that was previously selected to
become rejected. Therefore,
$W\subseteq Ch_D\bigl(
Y\setminus(Y_f\cup Y_G)\cup W
\bigr).$
Condition~(3) follows directly from the master-list condition of
Definition~\ref{def:fairness_resettlement}, because $G=F(\mathrm{Re}).$
Finally, Condition~(4) follows directly from the resettlement-feasibility condition of Definition~\ref{def:fairness_resettlement}, which requires
$Y\setminus(Y_f\cup Y_G)
\cup (W\cup\omega_G)$
to be feasible.
Thus, $W$ together with $G$ witnesses strongly justified envy by master
list and resettlement. Hence, $Y$ is not strongly fair by master list and
resettlement.

Therefore, by contraposition, strong fairness by master list and
resettlement implies fairness by master list and resettlement, which completes the proof of Proposition~\ref{prop:connection:sfmr-implies-fmr}.
\end{proof}

\begin{algorithm}[tb]
\KwIn{An instance
$I=(C,F,D,Q,\mathcal{X},\omega,\succ_F,\succ_D)$ and a strict master list
$f_1\rhd f_2\rhd\cdots\rhd f_m$}
\KwOut{A feasible and individually rational matching that is non-wasteful
and strongly fair by master list and resettlement}

$Y\leftarrow \omega$\;
$i\leftarrow 1$\;

\While{$i\le m$}{
    $\mathit{updated}\leftarrow \mathrm{False}$\;
    Let $f\leftarrow f_i$\;

    \ForEach{assignment $W\succ_f Y_f$, in decreasing order of $\succ_f$}{

        \If{$Y\setminus Y_f\cup W
        =\Ch_D(Y\setminus Y_f\cup W)$}{
            $Y\leftarrow Y\setminus Y_f\cup W$\;
            $\mathit{updated}\leftarrow \mathrm{True}$\;
            \textbf{break}\;
        }

        \ElseIf{there exists a nonempty $G\subseteq L_f(W,Y)$ satisfying
        Conditions~(2)--(4) of
        Definition~\ref{def:SFMR}}{
            $Y\leftarrow
            Y\setminus(Y_f\cup Y_G)\cup(W\cup\omega_G)$\;
            $\mathit{updated}\leftarrow \mathrm{True}$\;
            \textbf{break}\;
        }
    }

    \eIf{$\mathit{updated}$}{
        $i\leftarrow 1$\;
    }{
        $i\leftarrow i+1$\;
    }
}

\Return{$Y$}\;

\caption{Greedy Improvement for Non-wasteful and SFMR Matchings}
\label{alg:sfmr-greedy}
\end{algorithm}

Proposition~\ref{prop:connection:sfmr-implies-fmr} does not extend to ordinary fairness.
In fact, fairness and SFMR are incomparable, as shown next.

\begin{proposition}
\label{prop:F-SFMR-incomparable}
Fairness (F) and strong fairness by master list and resettlement (SFMR)
do not imply each other.
\end{proposition}

\begin{proof}
\textbf{(i) Fairness does not imply SFMR.}
Consider two families $f$ and $g$ and two capacity-one daycares
$d_1,d_2$, with no initial enrollment. The instance is given by
\[
\begin{array}{c@{\qquad}c}
C_f=(a_1,a_2)
&
W_f=\{(a_1,d_1),(a_2,d_2)\}\succ_f\emptyset
\\[1mm]
C_g=(b_1,b_2)
&
W_g=\{(b_1,d_1),(b_2,d_2)\}\succ_g\emptyset
\\[1mm]
a_1\succ_{d_1}b_1
&
b_2\succ_{d_2}a_2
\quad\text{and}\quad
f\rhd g.
\end{array}
\]
These are the only nonempty acceptable assignments of the two families.

Consider
\[
Y=W_g=\{(b_1,d_1),(b_2,d_2)\}.
\]
Family $g$ already receives its most-preferred assignment. The only
profitable deviation of $f$ is $W_f$, but
\[
W_f\not\subseteq Ch_D(Y\cup W_f),
\]
because $a_2$ is rejected at $d_2$ in favor of $b_2$. Hence, $f$ does
not form a blocking coalition, and $Y$ is fair.

However, $Y$ is not SFMR. Since
$a_1\succ_{d_1}b_1$, we have $g\in L_f(W_f,Y)$. Taking $G=\{g\}$,
\[
W_f\subseteq Ch_D(Y\setminus Y_G\cup W_f),
\qquad
Y\setminus Y_G\cup(W_f\cup\omega_G)=W_f,
\]
where the latter matching is feasible because $\omega_g=\emptyset$.
Moreover, $f\rhd g$. Thus, $f$ has weakly justified envy by master list
and resettlement, so $Y$ violates SFMR.

\medskip
\noindent
\textbf{(ii) SFMR does not imply fairness.}
Consider two singleton families and one capacity-one daycare $d$, again
with no initial enrollment:
\[
\begin{array}{c@{\qquad}c}
C_f=(a)
&
(a,d)\succ_f\emptyset
\\[1mm]
C_g=(b)
&
(b,d)\succ_g\emptyset
\\[1mm]
a\succ_d b
&
g\rhd f.
\end{array}
\]
Consider
\[
Y=\{(b,d)\}.
\]
Family $f$ can profitably deviate to $W=\{(a,d)\}$. Since
\[
W\subseteq Ch_D(Y\cup W)
\]
and $(b,d)$ is rejected, $f$ has justified envy. Hence, $Y$ is not fair.

On the other hand,
\[
L_f(W,Y)=\{g\}.
\]
Thus, the only possible nonempty displacement set is $G=\{g\}$, but the
master-list condition fails because $g\rhd f$. Family $g$ has no
profitable deviation. Hence, $Y$ satisfies SFMR.

Therefore, fairness and SFMR are incomparable.
\end{proof}

\begin{theorem}
% [Existence of NW and SFMR Matchings]
\label{theo:existence:SFMR}
For every instance with a strict master list, there always exists an individually rational matching that is both non-wasteful and strongly fair
by master list and resettlement.
\end{theorem}

\begin{proof}
We prove the theorem by showing that
Algorithm~\ref{alg:sfmr-greedy} preserves feasibility and individual
rationality, terminates after finitely many updates, and, upon termination,
returns a matching satisfying both non-wastefulness and strong fairness by
master list and resettlement.

\paragraph{Feasibility and individual rationality.}
The algorithm starts from the initial matching $Y=\omega$, which is
feasible and individually rational. We show that both types of updates
preserve these properties.

Suppose first that family $f$ claims unused seats through an assignment
$W$. By Definition~\ref{def:nonwastefulness},
\[
Y\setminus Y_f\cup W
=
Ch_D(Y\setminus Y_f\cup W).
\]
Hence, the updated matching is feasible. Moreover,
$W\succ_f Y_f$, while the assignments of all other families remain
unchanged. Since the current matching is individually rational before the
update, individual rationality is preserved.

Suppose instead that family $f$ has weakly justified envy by master list
and resettlement through an assignment $W$ and a nonempty set
$G\subseteq L_f(W,Y)$. By Condition~(4) of
Definition~\ref{def:SFMR}, the updated matching
\[
Y\setminus(Y_f\cup Y_G)\cup(W\cup\omega_G)
\]
is feasible. Family $f$ strictly improves from $Y_f$ to $W$, every family
in $G$ is returned to its initial assignment, and every other family
retains its current assignment. Thus, individual rationality is again
preserved.

\paragraph{Termination.} Consider the assignment profile $(Y_{f_1},Y_{f_2},\ldots,Y_{f_m}),$ where $f_1 \rhd  f_2 \rhd \cdots \rhd f_m.$ We compare two assignment profiles lexicographically according to the master list and the families' preferences: at the highest-ranked family whose assignments differ, the profile assigning that family its more preferred assignment is regarded as larger.

Consider any update performed by family $f_i$. If $f_i$ claims unused seats, then the assignments of all families $f_1, \ldots, f_{i-1}$ remain unchanged, while $f_i$ strictly improves.
Therefore, the assignment profile strictly increases in the
lexicographic order.
If $f_i$ instead exercises weakly justified envy by master list and resettlement through some nonempty set $G$, Condition~(3) of
Definition~\ref{def:SFMR} implies that
$f_i \rhd g$ for every $g\in G$. Hence, every family in $G$ appears strictly below $f_i$ in the master list. Consequently, the assignments of $f_1, \ldots, f_{i-1}$ remain unchanged, $f_i$ strictly improves, and only families ranked below $f_i$ may be returned to their initial assignments. Thus, the first coordinate of the assignment profile that changes is the one corresponding to $f_i$, and this coordinate strictly improves. Therefore, the assignment profile again strictly increases lexicographically.

Since every family has a finite preference ordering, there are only finitely many possible assignment profiles. As every update strictly increases the assignment profile in the lexicographic order, no assignment profile can be revisited. Hence, the algorithm must terminate after finitely many updates.

\paragraph{Non-wastefulness and strong fairness.}
Upon termination, the algorithm has completed a full scan of the master list without performing an update.
Suppose, for contradiction, that the final matching $Y$ is wasteful. Then there exist a family $f$ and an assignment $W\succ_f Y_f$ through which $f$ can claim unused seats. When $f$ was examined during the final scan, the first update condition of Algorithm~\ref{alg:sfmr-greedy} would have been satisfied, and the algorithm would have performed another update. This contradicts termination. Hence, $Y$ is non-wasteful.

Next, suppose that $Y$ is not strongly fair by master list and
resettlement. By 
Definition~\ref{def:SFMR}, there exist a family $f$, an assignment $W\succ_f Y_f$, and a nonempty set
$G\subseteq L_f(W,Y)$ satisfying Conditions~(2)--(4) of the definition. When $f$ was examined during the final scan, the second update condition of Algorithm~\ref{alg:sfmr-greedy} would therefore have been satisfied, again contradicting termination. Hence, no family has weakly justified envy by master list and resettlement, and $Y$ is strongly fair by master list and resettlement.

Therefore, Algorithm~\ref{alg:sfmr-greedy} terminates at a feasible and individually rational matching that is both non-wasteful and strongly fair by master list and resettlement.
\end{proof}

Algorithm~\ref{alg:sfmr-greedy} establishes the existence of a
non-wasteful and strongly fair matching, but does not necessarily yield a
polynomial-time algorithm. In particular, for a given family $f$ and
proposed assignment $W$, determining whether there exists a nonempty set
$G\subseteq L_f(W,Y)$ satisfying Conditions~(2)--(4) of
Definition~\ref{def:SFMR} may itself be
computationally hard, as shown by the hardness result below.

% \newpage
%-------------------------------------------
% SubsubSection : Computational Complexity of SFMR
%-------------------------------------------
\subsubsection{Computational Complexity of Problems Related to SFMR}
\label{subsection:complexity:SFMR}

Unlike the other fairness notions considered above, which can be verified
in polynomial time, verifying SFMR is computationally intractable.
In particular, we show that detecting a violation of SFMR is NP-complete,
and hence verifying whether a given matching satisfies SFMR is
coNP-complete.

\begin{theorem}
\label{theo:verification:SFMR_NW}
Given a matching and a family, deciding whether the family has weakly justified envy by master list and resettlement is NP-complete, even when
every family has at most three children and daycare priorities are master-list consistent up to initial enrollments. Moreover, the
potentially envious family has only one child.
Consequently, deciding whether a given matching satisfies Strong
Fairness by Master List and Resettlement (SFMR) is coNP-complete.
\end{theorem}

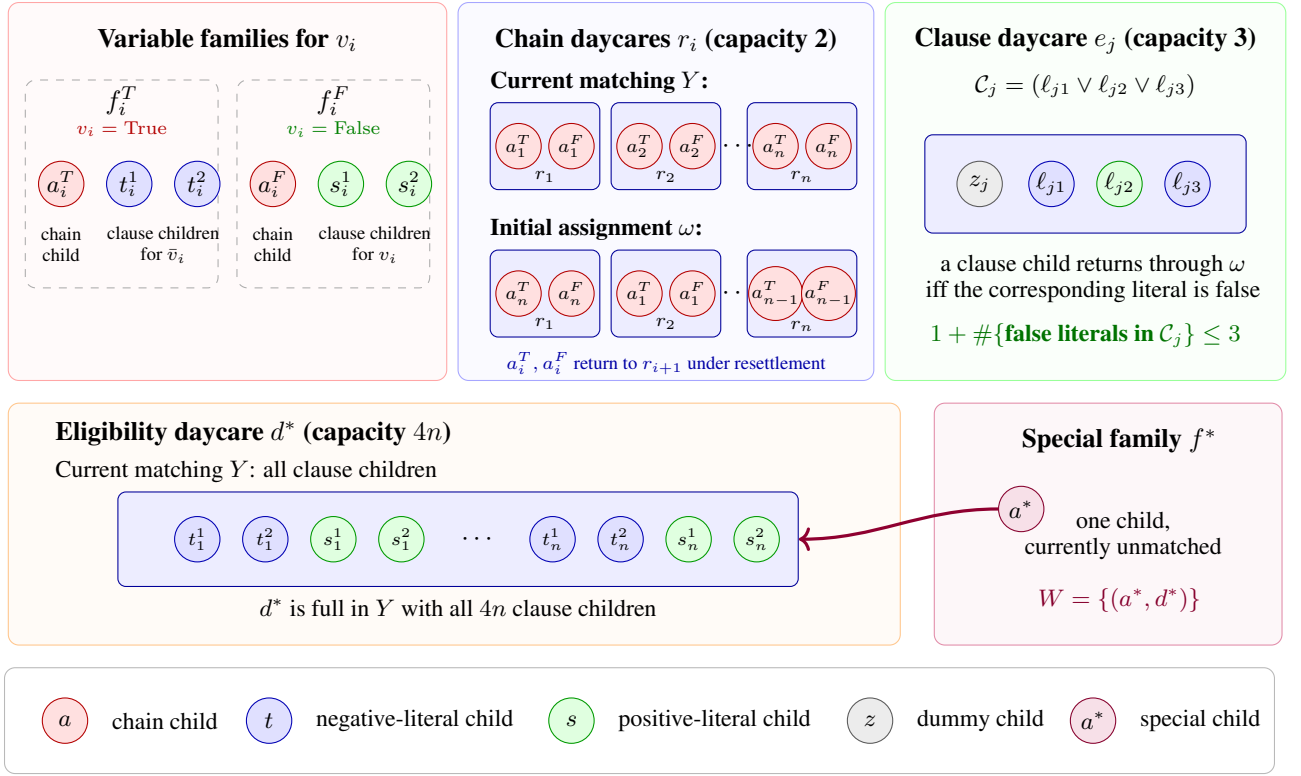
\begin{figure*}[tb]
\centering
\begin{tikzpicture}

% ---------- styles ----------
\tikzstyle{panel}=[draw=gray!55,rounded corners=4pt]
\tikzstyle{familybox}=[draw=gray!65,dashed,rounded corners=4pt]
\tikzstyle{chainchild}=[circle,draw=red!70!black,fill=red!12,minimum size=6mm,inner sep=0pt,font=\small]
\tikzstyle{negchild}=[circle,draw=blue!70!black,fill=blue!12,minimum size=6mm,inner sep=0pt,font=\small]
\tikzstyle{poschild}=[circle,draw=green!60!black,fill=green!12,minimum size=6mm,inner sep=0pt,font=\small]
\tikzstyle{dummychild}=[circle,draw=gray!70!black,fill=gray!15,minimum size=6mm,inner sep=0pt,font=\small]
\tikzstyle{specialchild}=[circle,draw=purple!70!black,fill=purple!12,minimum size=6mm,inner sep=0pt,font=\small]
\tikzstyle{daycare}=[draw=blue!60!black,fill=blue!7,rounded corners=3pt,minimum height=12mm,align=center]
\tikzstyle{proposal}=[->,very thick,draw=purple!75!black]
\tikzstyle{note}=[align=center,font=\small]

% ---------- title ----------
% \node[fill=blue!65!black,text=white,minimum width=16.8cm,minimum height=6mm,font=\large\bfseries]
% at (0,5.65) {Reduction from $(2,2)$-\textsc{E3-SAT}};

% ============================================================
% (a) VARIABLE FAMILIES
% ============================================================
\node[panel,draw=red!35,fill=red!2,minimum width=5.8cm,minimum height=5.0cm]
at (-5.65,2.65) {};
\node[font=\bfseries] at (-5.65,4.65) {Variable families for \(v_i\)};

\node[familybox,minimum width=2.55cm,minimum height=2.75cm] at (-7.05,2.75) {};
\node[font=\bfseries] at (-7.05,3.85) {\(f_i^T\)};
\node[font=\scriptsize,text=red!70!black] at (-7.05,3.48) {\(v_i=\mathrm{True}\)};
\node[chainchild] at (-7.85,2.75) {\(a_i^T\)};
\node[negchild] at (-6.95,2.75) {\(t_i^1\)};
\node[negchild] at (-6.05,2.75) {\(t_i^2\)};
\node[font=\scriptsize,align=center] at (-7.85,1.95) {chain\\child};
\node[font=\scriptsize,align=center] at (-6.50,1.95) {clause children\\for \(\bar v_i\)};

\node[familybox,minimum width=2.55cm,minimum height=2.75cm] at (-4.25,2.75) {};
\node[font=\bfseries] at (-4.25,3.85) {\(f_i^F\)};
\node[font=\scriptsize,text=green!50!black] at (-4.25,3.48) {\(v_i=\mathrm{False}\)};
\node[chainchild] at (-5.05,2.75) {\(a_i^F\)};
\node[poschild] at (-4.15,2.75) {\(s_i^1\)};
\node[poschild] at (-3.25,2.75) {\(s_i^2\)};
\node[font=\scriptsize,align=center] at (-5.05,1.95) {chain\\child};
\node[font=\scriptsize,align=center] at (-3.70,1.95) {clause children\\for \(v_i\)};

% ============================================================
% (b) CHAIN DAYCARES
% ============================================================
\node[panel,draw=blue!35,fill=blue!2,minimum width=5.5cm,minimum height=5.0cm]
at (0.15,2.65) {};
\node[font=\bfseries] at (0.15,4.65) {Chain daycares \(r_i\) (capacity 2)};
\node[anchor=west,font=\small\bfseries] at (-2.30,4.10) {Current matching \(Y\):};

% current row
\foreach \x/\lab/\one/\two in {
-1.45/r_1/a_1^T/a_1^F,
0.15/r_2/a_2^T/a_2^F,
1.95/r_n/a_n^T/a_n^F}{
  \node[daycare,minimum width=1.45cm,minimum height=1.15cm] at (\x,3.25) {};
  \node[chainchild] at ({\x-0.35},3.25) {\scriptsize \(\one\)};
  \node[chainchild] at ({\x+0.35},3.25) {\scriptsize \(\two\)};
  \node[font=\scriptsize] at (\x,2.85) {\(\lab\)};
}
\node at (1.10,3.25) {\(\cdots\)};

\node[anchor=west,font=\small\bfseries] at (-2.30,2.15) {Initial assignment \(\omega\):};

% initial row
\node[daycare,minimum width=1.45cm,minimum height=1.15cm] at (-1.45,1.30) {};
\node[chainchild] at (-1.80,1.30) {\scriptsize \(a_n^T\)};
\node[chainchild] at (-1.10,1.30) {\scriptsize \(a_n^F\)};
\node[font=\scriptsize] at (-1.45,0.9) {\(r_1\)};

\node[daycare,minimum width=1.45cm,minimum height=1.15cm] at (0.15,1.30) {};
\node[chainchild] at (-0.20,1.30) {\scriptsize \(a_1^T\)};
\node[chainchild] at (0.50,1.30) {\scriptsize \(a_1^F\)};
\node[font=\scriptsize] at (0.15,0.9) {\(r_2\)};

\node at (1.10,1.30) {\(\cdots\)};

\node[daycare,minimum width=1.45cm,minimum height=1.15cm] at (1.95,1.30) {};
\node[chainchild] at (1.60,1.30) {\scriptsize \(a_{n-1}^T\)};
\node[chainchild] at (2.30,1.30) {\scriptsize \(a_{n-1}^F\)};
\node[font=\scriptsize] at (1.95,0.87) {\(r_n\)};

\node[note,text=blue!65!black,font=\scriptsize] at (0.15,0.38)
{\(a_i^T,a_i^F\) return to \(r_{i+1}\) under resettlement};

% ============================================================
% (c) CLAUSE DAYCARE
% ============================================================
\node[panel,draw=green!40,fill=green!2,minimum width=5.3cm,minimum height=5.0cm]
at (5.70,2.65) {};
\node[font=\bfseries] at (5.70,4.65) {Clause daycare \(e_j\) (capacity 3)};
\node[font=\small] at (5.70,4.05) {\(\mathcal C_j=(\ell_{j1}\vee\ell_{j2}\vee\ell_{j3})\)};
\node[daycare,minimum width=4.25cm,minimum height=1.3cm] at (5.70,2.75) {};
\node[dummychild] at (4.30,2.75) {\(z_j\)};
% \node[font=\scriptsize,align=center] at (4.30,2) {dummy child\\in \(Y\)};
\node[negchild] at (5.25,2.75) {\(\ell_{j1}\)};
\node[poschild] at (6.15,2.75) {\(\ell_{j2}\)};
\node[negchild] at (7.05,2.75) {\(\ell_{j3}\)};
\node[font=\small,align=center] at (5.8,1.5)
{a clause child returns through \(\omega\)\\iff the corresponding literal is false};
\node[note,font=\small\bfseries,text=green!45!black] at (5.70,0.75)
{\(1+\#\{\text{false literals in }\mathcal C_j\}\leq3\)};

% ============================================================
% ELIGIBILITY DAYCARE d*
% ============================================================
\node[panel,draw=orange!45,fill=orange!3,minimum width=11.8cm,minimum height=3.2cm]
at (-2.65,-1.75) {};
\node[anchor=west,font=\bfseries] at (-8.05,-0.55)
{Eligibility daycare \(d^*\) (capacity \(4n\))};
\node[anchor=west,font=\small] at (-8.05,-1.05)
{Current matching \(Y\): all clause children};

\node[daycare,minimum width=9cm,minimum height=1.25cm] (dstar) at (-2.60,-1.95) {};
% \node[font=\bfseries] at (-2.10,-1.35) {\(d^*\)};

\node[negchild] at (-6.05,-1.95) {\scriptsize \(t_1^1\)};
\node[negchild] at (-5.15,-1.95) {\scriptsize \(t_1^2\)};
\node[poschild] at (-4.25,-1.95) {\scriptsize \(s_1^1\)};
\node[poschild] at (-3.35,-1.95) {\scriptsize \(s_1^2\)};
\node at (-2.35,-1.95) {\(\cdots\)};
\node[negchild] at (-1.35,-1.95) {\scriptsize \(t_n^1\)};
\node[negchild] at (-0.45,-1.95) {\scriptsize \(t_n^2\)};
\node[poschild] at (0.45,-1.95) {\scriptsize \(s_n^1\)};
\node[poschild] at (1.35,-1.95) {\scriptsize \(s_n^2\)};
\node[note,font=\small] at (-2.60,-2.85)
{\(d^*\) is full in \(Y\) with all \(4n\) clause children};

% ============================================================
% SPECIAL FAMILY f*
% ============================================================
\node[panel,draw=purple!45,fill=purple!3,minimum width=4.6cm,minimum height=3.2cm]
at (6.0,-1.75) {};
\node[font=\bfseries] at (6.15,-0.65) {Special family \(f^*\)};
\node[specialchild] (astar) at (4.85,-1.55) {\(a^*\)};
\node[note] at (6.2, -1.9) {one child,\\currently unmatched};
\node[note,text=purple!70!black] at (6.15,-2.75) {\(W=\{(a^*,d^*)\}\)};
\draw[proposal] (astar) to[out=180,in=0] (dstar);

% % ============================================================
% % MASTER LIST
% % ============================================================
% \node[panel,draw=cyan!45!black,fill=cyan!3,minimum width=9.6cm,minimum height=2.5cm]
% at (-3.65,-4.75) {};
% \node[anchor=west,font=\bfseries] at (-8.05,-3.95) {Master list};
% \node[note] at (-3.65,-4.60)
% {\(f^*\rhd f_1^T\rhd f_1^F\rhd f_2^T\rhd f_2^F
% \rhd\cdots\rhd f_n^T\rhd f_n^F\rhd h_1\rhd\cdots\rhd h_m\)};
% \node[note,font=\small,text width=8.7cm] at (-3.65,-5.35)
% {Initially enrolled children are ranked above all non-incumbents.
% Among non-incumbents, priorities follow the master list.};

% % ============================================================
% % KEY IDEAS
% % ============================================================
% \node[panel,draw=purple!40,fill=purple!3,minimum width=6.8cm,minimum height=2.5cm]
% at (4.85,-4.75) {};
% \node[anchor=west,font=\bfseries] at (1.8,-3.85) {Key ideas};
% \node[anchor=west,font=\small,align=left] at (1.80,-4.95)
% {1.\ \(G\) encodes a truth assignment.\\[1mm]
% 2.\ Chain daycares force one choice per variable.\\[1mm]
% 3.\ Clause capacities encode clause satisfaction.\\[1mm]
% 4.\ Displacing \(G\) makes \(a^*\) acceptable at \(d^*\).};

% ============================================================
% LEGEND
% ============================================================
\node[
    panel,
    minimum width=16.8cm,
    minimum height=1.4cm
] at (-0.2, -4.35) {};

% \node[
%     anchor=west,
%     font=\small\bfseries
% ] at (-7.65,-6.85)
% {Legend:};

\node[chainchild] at (-7.8,-4.35)
{\(a\)};

\node[
    anchor=west,
    font=\small
] at (-7.3,-4.35)
{chain child};

\node[negchild] at (-5.1,-4.35)
{\(t\)};

\node[
    anchor=west,
    font=\small
] at (-4.6,-4.35)
{negative-literal child};

\node[poschild] at (-1.1,-4.35)
{\(s\)};

\node[
    anchor=west,
    font=\small
] at (-0.6,-4.35)
{positive-literal child};

\node[dummychild] at (2.85,-4.35)
{\(z\)};

\node[
    anchor=west,
    font=\small
] at (3.35,-4.35)
{dummy child};

\node[specialchild] at (5.8,-4.35)
{\(a^*\)};

\node[
    anchor=west,
    font=\small
] at (6.3,-4.35)
{special child};

\end{tikzpicture}

\caption{
Illustration of the reduction from $(2,2)$-\textsc{E3-SAT}.
The variable families encode truth values; the cyclic chain daycares
enforce a consistent choice for every variable; the clause daycares
encode clause satisfaction through resettlement feasibility; and the
special family $f^*$ proposes $W=\{(a^*,d^*)\}$ at the eligibility
daycare $d^*$.
}
\label{fig:SFMR-hardness}

\end{figure*}

\begin{proof}
We first show that the problem belongs to NP.
Given a proposed assignment $W$ for a family $f^*$ and a nonempty set
of families $G$, we can compute $L_{f^*}(W,Y)$ and verify in polynomial
time whether
\[
G\subseteq L_{f^*}(W,Y)
\]
and whether Conditions~(1)--(4) of Definition~\ref{def:SFMR} are
satisfied. Hence, deciding whether a given family has weakly justified
envy by master list and resettlement belongs to NP.

We next prove NP-hardness by a polynomial-time reduction from
\emph{$(2,2)$-E3-SAT}, which is NP-complete~\citep{BKS03a}.
An instance consists of a Boolean formula
\[
B=\mathcal C_1\land\cdots\land\mathcal C_m
\]
over variables
\[
V=\{v_1,\ldots,v_n\},
\]
where every clause contains exactly three literals and each of
$v_i$ and $\bar v_i$ occurs exactly twice.

\paragraph{Variable families.}
For every variable $v_i$, create two families
$f_i^T$ and $f_i^F$.
Intuitively, selecting $f_i^T$ for displacement corresponds to setting
$v_i$ to true, whereas selecting $f_i^F$ corresponds to setting $v_i$
to false.

Family $f_i^T$ has three children
\[
C(f_i^T)=\{a_i^T,t_i^1,t_i^2\},
\]
and family $f_i^F$ has three children
\[
C(f_i^F)=\{a_i^F,s_i^1,s_i^2\}.
\]
We call $a_i^T$ and $a_i^F$ the \emph{chain children}, and
$t_i^1,t_i^2,s_i^1,s_i^2$ the \emph{clause children}.

The children $t_i^1,t_i^2$ correspond to the two occurrences of the
negative literal $\bar v_i$, while $s_i^1,s_i^2$ correspond to the two
occurrences of the positive literal $v_i$.
Thus, when $f_i^T$ is displaced and resettled, its clause children
return precisely to the clauses in which setting $v_i$ to true makes
the corresponding negative literal false.
Similarly, when $f_i^F$ is displaced and resettled, its clause children
return precisely to the clauses in which setting $v_i$ to false makes
the corresponding positive literal false.

\paragraph{Eligibility daycare.}
Create a daycare $d^*$ with capacity $4n$.
In the current matching $Y$, assign all clause children
\[
t_i^1,t_i^2,s_i^1,s_i^2
\qquad
(i=1,\ldots,n)
\]
to $d^*$.
Hence, $d^*$ is full in $Y$.

\paragraph{Variable-chain daycares.}
For every variable $v_i$, create a daycare $r_i$ of capacity $2$.
In the current matching $Y$, assign
\[
(a_i^T,r_i),\qquad (a_i^F,r_i).
\]
Thus, every $r_i$ is full.

For every $i=1,\ldots,n$, with indices taken modulo $n$, define
\[
(a_i^T,r_{i+1})\in\omega_{f_i^T},
\qquad
(a_i^F,r_{i+1})\in\omega_{f_i^F},
\]
where $r_{n+1}:=r_1$.
Thus, under the initial assignment, every daycare $r_i$ contains
exactly the two chain children associated with variable $v_{i-1}$,
whereas in the current matching it contains exactly $a_i^T$ and
$a_i^F$. Both assignments are feasible.

\paragraph{Clause daycares.}
For every clause $\mathcal C_j$, create a daycare $e_j$ of capacity $3$
and a dummy family $h_j$ with one child $z_j$.
In the current matching,
\[
(z_j,e_j)\in Y,
\]
and $z_j$ has no initial enrollment.

For every occurrence of a literal in clause $\mathcal C_j$, assign a
corresponding clause child to $e_j$ in the initial assignment as follows.
If the $r$-th occurrence of the positive literal $v_i$ appears in
$\mathcal C_j$, then
\[
(s_i^r,e_j)\in\omega_{f_i^F}.
\]
If the $r$-th occurrence of the negative literal $\bar v_i$ appears in
$\mathcal C_j$, then
\[
(t_i^r,e_j)\in\omega_{f_i^T}.
\]
Thus, for each literal in $\mathcal C_j$, the family selected when that
literal is false has one clause child initially assigned to $e_j$.
Since $\mathcal C_j$ contains exactly three literals, exactly three
clause children are initially assigned to $e_j$, and hence the initial
assignment is feasible at every clause daycare.

\paragraph{Special family.}
Create a special family $f^*$ with a single child $a^*$, currently
unmatched, and let
\[
W=\{(a^*,d^*)\}.
\]
At daycare $d^*$, child $a^*$ has strictly higher priority than every
clause child currently assigned there.
Consequently,
\[
f_i^T,f_i^F\in L_{f^*}(W,Y)
\qquad
\text{for every }i.
\]

\paragraph{Preferences, priorities, and master list.}
Let the strict master list be
\[
f^*
\rhd
f_1^T\rhd f_1^F
\rhd
f_2^T\rhd f_2^F
\rhd\cdots\rhd
f_n^T\rhd f_n^F
\rhd
h_1\rhd\cdots\rhd h_m.
\]
For every variable family $f_i^T$ and $f_i^F$, its current assignment
is its top choice, thus strictly preferred to its initial assignment, and every assignment
below its initial assignment is unacceptable.
Each dummy family $h_j$  prefers its current (and only acceptable) assignment
$(z_j,e_j)$ to being unmatched.
The special family $f^*$ has $W$ as its unique assignment strictly
preferred to its current and initial assignment $\emptyset$.

At every daycare, all initially enrolled children are ranked above all
non-incumbent children.
Among the remaining children, priorities follow the family master list
$\rhd$, with arbitrary strict tie-breaking between children of the same
family.
In particular, $d^*$ has no initially enrolled child, and therefore
$a^*$ is ranked above every clause child assigned there in $Y$.
At each clause daycare $e_j$, its three initially enrolled clause
children are ranked above the non-incumbent dummy child $z_j$.
At each chain daycare $r_i$, the two chain children initially enrolled
there, namely those associated with variable $v_{i-1}$, are ranked
above all non-incumbents, and the latter are ordered according to
$\rhd$.
Hence, all daycare priority lists are master-list consistent up to
initial enrollments.

The construction is polynomial in the size of $B$, every family has at
most three children, and the potentially envious family $f^*$ has only
one child.

\paragraph{Correctness.}
We show that $B$ is satisfiable if and only if $f^*$ has weakly
justified envy by master list and resettlement.

\smallskip
\noindent
$(\Rightarrow)$
Suppose that $B$ admits a satisfying truth assignment.
Define
\[
G
=
\{f_i^T : v_i=\mathrm{True}\}
\cup
\{f_i^F : v_i=\mathrm{False}\}.
\]
Thus, for every variable $v_i$, exactly one of $f_i^T$ and $f_i^F$
belongs to $G$.

By construction,
\[
W\succ_{f^*}Y_{f^*},
\]
so Condition~(1) holds.

At $d^*$, child $a^*$ has higher priority than every currently assigned
clause child.
Moreover, removing the assignments of the families in $G$ removes
clause children from $d^*$.
Hence $a^*$ is accepted, and
\[
W
\subseteq
Ch_D\bigl(
Y\setminus(Y_{f^*}\cup Y_G)\cup W
\bigr).
\]
Thus, Condition~(2) holds.

Every family in $G$ belongs to $L_{f^*}(W,Y)$ because it has clause
children assigned to $d^*$ that have lower priority than $a^*$.
Moreover,
\[
f^*\rhd g
\qquad
\text{for every }g\in G.
\]
Hence, Condition~(3) holds.

It remains to verify resettlement feasibility.
Consider first a chain daycare $r_i$.
Exactly one of the two families associated with variable $v_{i-1}$
belongs to $G$, so exactly one chain child of $v_{i-1}$ returns to
$r_i$ through resettlement.
At the same time, exactly one of $a_i^T$ and $a_i^F$ remains at $r_i$
in the current matching.
Hence, $r_i$ contains exactly two children after resettlement.

Next consider a clause daycare $e_j$.
The dummy child $z_j$ remains assigned there.
A clause child returns to $e_j$ under resettlement precisely when the
corresponding literal of $\mathcal C_j$ is false under the chosen truth
assignment.
Since the truth assignment satisfies $\mathcal C_j$, at most two of
its three literals are false.
Therefore, at most two clause children return to $e_j$.
Together with $z_j$, daycare $e_j$ contains at most three children,
which respects its capacity.

Finally, at $d^*$, every family in $G$ removes two clause children,
whereas only $a^*$ is added.
Since $G$ is nonempty, capacity is respected.

Therefore,
\[
Y\setminus(Y_{f^*}\cup Y_G)
\cup
(W\cup\omega_G)
\]
is feasible, and Condition~(4) holds.
Hence, $f^*$ has weakly justified envy by master list and resettlement.

\smallskip
\noindent
$(\Leftarrow)$
Conversely, suppose that $f^*$ has weakly justified envy by master list
and resettlement through $W$ and some nonempty set $G$.

By construction, the only families in $L_{f^*}(W,Y)$ are the variable
families $f_i^T$ and $f_i^F$.
Hence, $G$ contains only variable families.

For every variable $v_i$, define
\[
x_i
=
|G\cap\{f_i^T,f_i^F\}|.
\]
Thus,
\[
x_i\in\{0,1,2\}.
\]
Since $G$ is nonempty, there exists some $j$ with $x_j\geq 1$.

For every $i$, with indices taken modulo $n$, consider daycare
$r_{i+1}$.
In the resettled matching, $x_i$ chain children associated with
variable $v_i$ return to $r_{i+1}$, while
$2-x_{i+1}$ of the two currently assigned chain children
$a_{i+1}^T,a_{i+1}^F$ remain there.
Hence, feasibility gives
\[
x_i+(2-x_{i+1})\leq 2,
\]
and therefore
\[
x_{i+1}\geq x_i.
\]
Applying these inequalities cyclically yields
\[
x_1=x_2=\cdots=x_n.
\]
Since $G\neq\emptyset$, their common value is either $1$ or $2$.

The common value cannot be $2$.
Indeed, in that case every variable family belongs to $G$.
Consequently, at every clause daycare $e_j$, all three clause children
corresponding to the three literals of $\mathcal C_j$ return through
resettlement.
Together with the dummy child $z_j$, this would give four children at
the capacity-three daycare $e_j$, contradicting Condition~(4).
Hence,
\[
x_i=1
\qquad
\text{for every }i=1,\ldots,n.
\]

Therefore, exactly one of $f_i^T$ and $f_i^F$ belongs to $G$ for every
variable.
Define a truth assignment by
\[
v_i=\mathrm{True}
\quad\Longleftrightarrow\quad
f_i^T\in G.
\]

Consider any clause $\mathcal C_j$.
The dummy child $z_j$ remains assigned to $e_j$.
For every literal of $\mathcal C_j$ that is false under the induced
truth assignment, the corresponding selected variable family belongs
to $G$, and hence its clause child returns to $e_j$ through
resettlement.

If all three literals in $\mathcal C_j$ were false, then all three
corresponding clause children would return to $e_j$.
Together with $z_j$, daycare $e_j$ would contain four children despite
having capacity $3$, contradicting Condition~(4).
Hence, every clause contains at least one true literal.
Therefore, the induced truth assignment satisfies $B$.

We conclude that $B$ is satisfiable if and only if $f^*$ has weakly
justified envy by master list and resettlement.
Since the reduction is polynomial, the problem is NP-hard.
Together with membership in NP, the problem is NP-complete.
\end{proof}

\begin{theorem}
\label{theo:computation:SFMR}
Finding a feasible and individually rational matching that is both
non-wasteful and strongly fair by master list and resettlement is
PLS-hard.  The result holds even if
\begin{enumerate}[label=(\roman*)]
\item every daycare has capacity at most four; and
\item daycare priorities are master-list consistent up to initial
enrollments.
\end{enumerate}
\end{theorem}

\begin{proof}
We reduce from the PLS-hard computation problem for non-wasteful and
FMR-fair matchings established in
Theorem~\ref{theo:computation:FMR}.  Consider any instance used in the
PLS-hardness reduction of that theorem, and use exactly the same
daycare-matching instance and family master list.

By Theorem~\ref{theo:existence:SFMR}, the constructed instance
admits a feasible and individually rational matching that is both
non-wasteful and strongly fair by master list and resettlement.  Let
$Y$ be any such matching.  By
Proposition~\ref{prop:connection:sfmr-implies-fmr}, strong fairness by master
list and resettlement implies fairness by master list and resettlement.
Hence, $Y$ is also a feasible and individually rational matching that is
both non-wasteful and FMR-fair.  Therefore, $Y$ is a valid solution to
the source search problem of
Theorem~\ref{theo:computation:FMR}.

Thus, both the instance map and the solution map are the identity and
are computable in polynomial time.  The reduction also preserves all
structural restrictions of
Theorem~\ref{theo:computation:FMR}.  Consequently, finding a
non-wasteful and SFMR-fair matching is PLS-hard even when every daycare
has capacity at most four and daycare priorities are master-list
consistent up to initial enrollments.
\end{proof}

% \begin{theorem}
% \label{theo:computation:SFMR}
% Finding a feasible and individually rational matching that is
% non-wasteful and Weak-FMR-fair is PLS-hard.  The result holds even if
% \begin{enumerate}[label=(\roman*)]
% \item every daycare has capacity at most four; and
% \item daycare priorities are master-list consistent up to initial enrollments.
% \end{enumerate}
% \end{theorem}

% \begin{proof}
% We reduce from the PLS-hard search problem established in
% Theorem~\ref{theo:fmr:pls-hard-consistent-priorities}.  Given an instance
% produced by that reduction, use exactly the same daycare-matching instance.

% Let $Y$ be any feasible and individually rational matching that is
% non-wasteful and Weak-FMR-fair that always exist by Theorem~\ref{theo:weak-fmr:existence}.  By
% Proposition~\ref{prop:weak-fmr-implies-fmr}, $Y$ is also fair by master list
% and resettlement in the sense of
% Definition~\ref{def:fairness_resettlement}.  Hence $Y$ is a valid solution
% to the search problem of
% Theorem~\ref{theo:fmr:pls-hard-consistent-priorities}.  The solution map is
% the identity and is therefore polynomial-time computable.

% The reduction preserves all structural restrictions of
% Theorem~\ref{theo:fmr:pls-hard-consistent-priorities}.  Therefore finding a
% non-wasteful and Weak-FMR-fair matching is PLS-hard even when every daycare
% has capacity at most four and daycare priorities are master-list consistent
% up to initial enrollments.
% \end{proof}

% \clearpage
%----------------------------------------
% section: Weakening Non-wastefulness
%----------------------------------------
\section{Weakening Non-wastefulness}

In this section, we preserve fairness and explore which weaker variants of non-wastefulness can be reconciled with it. We begin by extending the cutoff algorithm~\citep{KaKo24a,ABB24a}, originally proposed to achieve fairness in matching problems with complex distributional constraints, to our setting with joint family preferences. We show, however, that the extended algorithm fails to satisfy fairness (Definition~\ref{def:fairness}), even in the absence of initial enrollments. Nevertheless, it satisfies the weaker notion of fairness by replacement (Definition~\ref{def:fairness_replacement}). 

To resolve this limitation, we introduce a new relaxation of non-wastefulness, termed \emph{fairness-preserving non-wastefulness}, which is compatible with fairness. Furthermore, this concept is formulated in a general manner and can be adapted to accommodate other notions of fairness.

%----------------------------------------
% subsection: Cutoff Scores
%----------------------------------------
\subsection{Cutoff Scores}
Before introducing the Generalized Cutoff Algorithm, we first define the cutoff scores used in the procedure. 
For each daycare $d$, assign each contract $(c,d)\in\mathcal X_d$ a
priority score $s_d(c)$ consistent with the priority ordering $\succ_d$,
where a higher score represents a higher priority. We normalize the scores
so that
\[
s_d(c)\in\{1,\ldots,|C|\},
\]
with higher-priority children receiving higher scores, and use $|C|+1$ as
an artificial cutoff above the score of every child.

For each daycare $d$, let
\[
Cut(d)\in\{1,\ldots,|C|+1\}
\]
denote its cutoff score. A newly obtained contract $(c,d)$ is available to
child $c$ if
\[
s_d(c)\ge Cut(d).
\]
Thus, lowering $Cut(d)$ weakly expands the set of contracts available at
daycare $d$.

Given a cutoff vector $Cut=(Cut(d))_{d\in D}$, define the set of assignments
available to family $f$ under $Cut$ as
\[
\mathcal W_f(Cut)
=
\left\{
W\in\mathcal W_f
\;\middle|\;
W\succeq_f\omega_f
\text{ and }
s_d(c)\ge Cut(d)
\text{ for every }(c,d)\in W\setminus\omega_f
\right\}.
\]
In other words, an assignment is available to a family if it is individually rational and every newly obtained contract meets the corresponding daycare cutoff. Contracts retained from a family's initial assignment remain available by virtue of the family's initial entitlement.

Since $\omega_f\in\mathcal W_f(Cut)$ for every cutoff vector $Cut$, the set $\mathcal W_f(Cut)$ is always nonempty. Let $Y_f(Cut)$ denote the most preferred assignment in $\mathcal W_f(Cut)$ according to $\succ_f$.

We define the matching induced by $Cut$ as
\[
Y(Cut)
=
\bigcup_{f\in F}Y_f(Cut).
\]
Since family preferences are strict, $Y(Cut)$ is uniquely determined by
$Cut$. Although $Y(Cut)$ is acceptable for both families and daycares by construction, it need not be feasible, because families choose their assignments independently under the same cutoff vector.

%----------------------------------------
% subsection: Generalized Cutoff Algorithm
%----------------------------------------

\subsection{Generalized Cutoff (GC) Algorithm}

We next introduce the Generalized Cutoff (GC) Algorithm, inspired by the Cutoff Algorithm~\citep{KaKo24a,ABB24a} and adapted to accommodate joint preferences among siblings. 

We now describe Algorithm~\ref{alg:GCA} in detail. The GC Algorithm starts with the cutoff of every daycare set to $|C|+1$. At these cutoffs, no newly obtained contract is available, so the induced matching coincides with the initial matching $\omega$. The algorithm then searches for a daycare whose cutoff can be decreased by one while keeping the induced matching feasible. If such a daycare exists, its cutoff is decreased by one and the induced matching is recomputed. The process continues until no daycare cutoff can
be decreased further without making the induced matching infeasible.

A key distinction between the Generalized Cutoff Algorithm and the original one lies in the construction of the induced matching. In our setting, preferences are defined at the family level rather than independently for individual children. As a result, lowering the cutoff of a single daycare may make an entire joint assignment available to a family and thereby change the assignments of several siblings simultaneously across different daycares. This interaction across daycares is the main challenge in extending the cutoff approach to matching markets with siblings.

\begin{algorithm}[t]
\caption{Generalized Cutoff (GC) Algorithm}
\label{alg:GCA}
\KwIn{An instance
$I=(C,F,D,Q,\mathcal X,\omega,\succ_F,\succ_D)$}
\KwOut{A feasible and individually rational matching $Y$}

Initialize $Cut(d)\gets |C|+1$ for every daycare $d\in D$\;
% \haris{If you want to satisfy IR, then you need to specify further...}

\BlankLine
\textbf{Induced matching:}
Given $Cut$, each family $f$ selects its most preferred assignment
$W\succeq_f\omega_f$ such that every newly obtained contract
$(c,d)\in W\setminus\omega_f$ satisfies $s_d(c)\ge Cut(d)$.
Let $Y$ be the union of these assignments over all families\;
% Construct the induced matching $Y$\;

\BlankLine
\While{there exists a daycare $d\in D$ such that decreasing $Cut(d)$
by one results in a feasible induced matching}{
    $Cut(d)\gets Cut(d)-1$\;
    Recompute the induced matching $Y$\;
}

\Return{$Y$}\;
\end{algorithm}

\begin{proposition}
\label{prop:GCA:running_time}
The Generalized Cutoff Algorithm returns a matching that is both individually rational and feasible. Moreover, the algorithm runs in $O(\ell\cdot |C|\cdot |D|)$ time, where $\ell$ denotes the total length of all families' preference lists.
\end{proposition}

\begin{proof}
Initially, all daycare cutoffs are set to $|C|+1$. Since no newly obtained contract is available at these cutoffs, the induced matching coincides with the initial matching $\omega$, which is feasible and individually rational.

Throughout the algorithm, each family selects its most preferred assignment among those that are individually rational and available under the current cutoffs. Hence, every induced matching is individually rational. Moreover, a cutoff is decreased only when the resulting induced matching is feasible. Thus, feasibility is preserved throughout the algorithm.

Each daycare cutoff starts at $|C|+1$ and can be decreased to $1$, so each cutoff can be reduced at most $|C|$ times. Therefore, the algorithm performs at most $|C||D|$ cutoff reductions and must terminate. For any fixed cutoff
vector, the induced matching can be computed in $O(\ell)$ time by scanning the families' preference lists, where $\ell$ denotes the total length of all families' preference lists. Hence, the overall running time is $O(\ell\cdot |C|\cdot |D|)$.
\end{proof}

\begin{proposition}
\label{prop:GCA:NOT:fair}
The Generalized Cutoff Algorithm does not necessarily satisfy fairness or
fairness by replacement and withdrawal, even when no child is initially
enrolled.
\end{proposition}

\begin{proof}
Consider three families: $f_1$ with two children
$C_{f_1}= (c_1,c_2)$, $f_2$ with one child $C_{f_2}=(c_3)$, and
$f_3$ with one child $C_{f_3}=(c_4)$. There are two daycares $d_1$ and
$d_2$, each with capacity one, and no child is initially enrolled. The
family preferences and daycare priorities are given by
\begin{align*}
\succ_{f_1}:&\quad
\{(c_1,d_1),(c_2,d_2)\}
\succ_{f_1}
\{(c_1,d_2)\}
\succ_{f_1}\emptyset,
\\
\succ_{f_2}:&\quad
(c_3,d_1)\succ_{f_2}\emptyset,
\qquad
\succ_{d_1}:\quad c_1\succ_{d_1}c_3,
\\
\succ_{f_3}:&\quad
(c_4,d_2)\succ_{f_3}\emptyset,
\qquad
\succ_{d_2}:\quad
c_1\succ_{d_2}c_2\succ_{d_2}c_4.
\end{align*}
Accordingly, assign the priority scores
\[
s_{d_1}(c_1)=4,\qquad s_{d_1}(c_3)=3,
\]
and
\[
s_{d_2}(c_1)=4,\qquad
s_{d_2}(c_2)=3,\qquad
s_{d_2}(c_4)=2.
\]

Initially,
\[
(Cut(d_1),Cut(d_2))=(5,5),
\]
and the induced matching is empty. Consider an execution in which the
cutoff of $d_1$ is lowered first. When $Cut(d_1)$ reaches $3$, child $c_3$
becomes available at $d_1$, and the induced matching becomes
\[
\{(c_3,d_1)\}.
\]
The cutoff of $d_1$ can be further lowered to $1$ without changing the
induced matching, because $f_1$'s preferred joint assignment is still
unavailable at $d_2$.

Next, lower the cutoff of $d_2$ from $5$ to $4$. Contract $(c_1,d_2)$
then becomes available to $f_1$, yielding the feasible induced matching
\[
Y=\{(c_3,d_1),(c_1,d_2)\}.
\]
Thus, the cutoff vector is now
\[
(Cut(d_1),Cut(d_2))=(1,4).
\]

The cutoff of $d_1$ cannot be lowered further. If the cutoff of $d_2$ is
lowered from $4$ to $3$, contract $(c_2,d_2)$ becomes available. Since
$(c_1,d_1)$ is already available under $Cut(d_1)=1$, family $f_1$ then
chooses its most preferred assignment
\[
W=\{(c_1,d_1),(c_2,d_2)\}.
\]
At the same time, $f_2$ continues to choose $(c_3,d_1)$. The resulting
induced matching therefore assigns both $c_1$ and $c_3$ to $d_1$ and is
infeasible. Hence, no cutoff can be decreased further, and the GC
Algorithm terminates at $Y$.

However, family $f_1$ strictly prefers $W$ to its current assignment
\[
Y_{f_1}=\{(c_1,d_2)\}.
\]
Under the deviation to $W$, contract $(c_1,d_1)$ replaces
$(c_3,d_1)$, where $c_1\succ_{d_1}c_3$, while the withdrawal of
$(c_1,d_2)$ makes the seat at $d_2$ available for $(c_2,d_2)$.
Therefore, $f_1$ has justified envy by replacement and withdrawal through
$W$.

Thus, $Y$ does not satisfy fairness by replacement and withdrawal. By
Proposition~\ref{prop:connection:F&FRW&FR}, it does not satisfy fairness
either. Hence, the Generalized Cutoff Algorithm does not necessarily
satisfy either property.
\end{proof}

\begin{proposition}
\label{prop:GC:FR}
Suppose that each daycare prioritizes its initially enrolled children over
children who are not initially enrolled there. Then the Generalized Cutoff
Algorithm satisfies fairness by replacement.
\end{proposition}

\begin{proof}
Let $Y=Y(Cut)$ be the matching returned by the GC Algorithm. Suppose, toward a contradiction, that some family $f$ has justified envy by replacement through an assignment $W\succ_f Y_f$.

By the definition of fairness by replacement, for every newly obtained contract $(c,d)\in W\setminus Y_f$, there exists a distinct rejected contract $(c',d)\in Y\setminus Y_f$ such that $c\succ_d c'.$ We first observe that $(c',d)$ cannot be an initial contract of $c'$. Indeed, an initially enrolled child has higher priority at its initial daycare than any child who is not initially enrolled there. Hence, such a contract cannot be replaced by $(c,d)$ with $c\succ_d c'$.

Therefore, $(c',d)$ is a newly obtained contract in the assignment of its family. Since $Y$ is induced by the final cutoff vector, $(c',d)$ satisfies the cutoff at $d$:
\[
s_d(c')\ge Cut(d).
\]
Because $c\succ_d c'$, we have
\[
s_d(c)>s_d(c')\ge Cut(d),
\]
and hence $(c,d)$ also satisfies the cutoff at $d$.

Thus, every contract in $W\setminus Y_f$ satisfies the corresponding daycare cutoff. Contracts in $W\cap Y_f$ are already held by $f$ and hence also satisfy the relevant availability requirements. Since $W\succ_f Y_f\succeq_f\omega_f$, the assignment $W$ is individually
rational. It follows that
\[
W\in\mathcal W_f(Cut).
\]

However, by definition of the induced matching, $Y_f=Y_f(Cut)$ is the most preferred assignment of $f$ in $\mathcal W_f(Cut)$. This contradicts $W\succ_f Y_f$.

Therefore, no family has justified envy by replacement, and the matching returned by the GC Algorithm satisfies fairness by replacement.
% \haris{This proof looks too informal}
\end{proof}

\begin{proposition}
\label{prop:GC:NOT:CNW}
The Generalized Cutoff Algorithm does not necessarily satisfy
non-wastefulness.
\end{proposition}

\begin{proof}
Consider three families: $f_1$ with two children
$C_{f_1}=(c_1,c_3)$, $f_2$ with one child $C_{f_2}=(c_2)$, and
$f_3$ with one child $C_{f_3}=(c_4)$. There are two daycares $d_1$ and
$d_2$, each with capacity one, and no child is initially enrolled.

Family $f_1$ has only one acceptable assignment,
\[
\{(c_1,d_1),(c_3,d_2)\},
\]
while the preferences of $f_2$ and $f_3$ are
\[
(c_2,d_1)\succ_{f_2}(c_2,d_2)\succ_{f_2}\emptyset
\]
and
\[
(c_4,d_2)\succ_{f_3}\emptyset,
\]
respectively. Both daycares have the same priority ordering
\[
c_1\succ c_2\succ c_3\succ c_4,
\]
with priority scores
\[
s_d(c_1)=4,\qquad
s_d(c_2)=3,\qquad
s_d(c_3)=2,\qquad
s_d(c_4)=1
\]
for each $d\in\{d_1,d_2\}$.

Initially, both cutoffs are equal to $5$. Consider an execution in which
the cutoff of $d_1$ is lowered first. When $Cut(d_1)$ reaches $3$,
$(c_2,d_1)$ becomes available to $f_2$, and the induced matching becomes
\[
Y=\{(c_2,d_1)\}.
\]
The cutoff of $d_1$ can subsequently be lowered to $1$ without changing
the induced matching.

Next, lower the cutoff of $d_2$. Its cutoff can be lowered from $5$ to
$3$ while preserving feasibility. However, lowering it from $3$ to $2$
makes $(c_3,d_2)$ available. Since $(c_1,d_1)$ is already available under
$Cut(d_1)=1$, the joint assignment
\[
\{(c_1,d_1),(c_3,d_2)\}
\]
becomes available to $f_1$. Family $f_1$ then selects this assignment,
while $f_2$ continues to select $(c_2,d_1)$, making the induced matching
infeasible at $d_1$. Hence, the GC Algorithm can terminate with
\[
(Cut(d_1),Cut(d_2))=(1,3)
\]
and
\[
Y=\{(c_2,d_1)\}.
\]

The matching $Y$ is wasteful. Indeed, daycare $d_2$ is vacant, and family
$f_3$ strictly prefers
\[
W=\{(c_4,d_2)\}
\]
to its current assignment. Moreover,
\[
Y\cup W=Ch_D(Y\cup W),
\]
so $f_3$ can improve its assignment without rejecting any existing
contract. Therefore, $Y$ does not satisfy non-wastefulness.
\end{proof}

%--------------------------------------------------
% % Subsection: Fairness-preserving Non-wastefulness
%--------------------------------------------------
\subsection{Fairness-preserving Non-wastefulness}

We now introduce a relaxation of non-wastefulness to address the inherent tension between fairness and non-wastefulness in matching with siblings and initial enrollments. The idea is to permit waste only when eliminating it would violate fairness. More precisely, whenever a family can improve its assignment without displacing any other family, we require such an improvement if the resulting matching continues to satisfy fairness.

Because eliminating waste corresponds to improving a family's assignment without harming other families, this notion provides a principled way to identify welfare improvements that are compatible with fairness. The definition is analogous to the \emph{cutoff non-wastefulness} notion introduced by \citet{ABB24a}. Importantly, it is not tied to a particular fairness criterion and can be defined with respect to any fairness property $A$.

\begin{definition}[Fairness-preserving Non-wastefulness]
\label{def:fairness_preserving_nonwastefulness}
Let $A$ be a fairness property. Given a feasible and individually rational matching $Y$ and a new assignment $W$, a family $f$ is said to \emph{conditionally utilize vacant seats with respect to $A$} if the
following conditions hold:
\begin{description}
    \item[(1) improvement:]
    $W\succ_f Y_f$.

    \item[(2) no rejection:]
    $Y\setminus Y_f\cup W = Ch_D(Y\setminus Y_f\cup W).$

    \item[(3) preservation of fairness:]
    The matching $Y\setminus Y_f\cup W$ satisfies property $A$.
\end{description}
A feasible and individually rational matching $Y$ 
is said to be
\emph{fairness-preserving non-wasteful with respect to $A$} if no family
can conditionally utilize vacant seats with respect to $A$.
\end{definition}

\begin{algorithm}[t]
\KwIn{An instance $I=(C,F,D,Q,\mathcal{X},\omega,\succ_F,\succ_D)$, a fairness property $A$, and a feasible and individually rational matching $Y$ satisfying $A$}
\KwOut{A matching $Y\subseteq\mathcal{X}$}
    \caption{Fairness-preserving Non-Wasteful (FN) Algorithm}
    \label{alg:FN}

    \While{there exists a family $f$ and a new assignment $W$ s.t.: 
        \begin{enumerate}
            \item family $f$ prefers $W$ over $Y_f$, and
        \item $Y\setminus Y_f\cup W = Ch_D(Y\setminus Y_f\cup W)$
            \item the new matching $Y \setminus Y_f \cup W$ satisfies $A$
        \end{enumerate}}{
        Update $Y \gets Y \setminus Y_f \cup W$
    }
    \Return$Y$
\end{algorithm}

We next introduce the Fairness-preserving Non-Wasteful (FN) Algorithm. Starting from a feasible and individually rational matching satisfying property $A$, the algorithm repeatedly allows a family to improve its assignment whenever the improvement does not displace any other family
and preserves property $A$. The process continues until no such improvement is possible.
\begin{theorem}
\label{theo:existence:constrained}
Let $A$ be a fairness property, and let $Y^0$ be a feasible and
individually rational matching satisfying $A$. Then the Fairness-preserving Non-Wasteful Algorithm, initialized at $Y^0$, terminates at a feasible
and individually rational matching that satisfies $A$ and is
fairness-preserving non-wasteful with respect to $A$.

Moreover, if property $A$ can be checked in polynomial time, then it runs in polynomial time given such an initial
matching $Y^0$. Consequently, if a feasible and individually rational
matching satisfying $A$ can also be found in polynomial time, then an
$A$-fairness-preserving non-wasteful matching can be computed in
polynomial time.
\end{theorem}

\begin{proof}
Start from the given feasible and individually rational matching
$Y^0$ satisfying property $A$, and apply
Algorithm~\ref{alg:FN}.

Consider an update by a family $f$ from its current assignment $Y_f$
to an assignment $W\succ_f Y_f$. By the update condition,
\[
Y\setminus Y_f\cup W
=
Ch_D(Y\setminus Y_f\cup W),
\]
so no contract is rejected and the updated matching is feasible.
The assignments of all families other than $f$ remain unchanged.
Moreover, since $W\succ_f Y_f$ and $Y_f\succeq_f\omega_f$, the new
assignment $W$ is individually rational for $f$. Hence individual
rationality is preserved. Finally, the algorithm performs the update
only if the resulting matching continues to satisfy property $A$.
Thus, every update preserves feasibility, individual rationality, and
property $A$.

Each update strictly improves the assignment of exactly one family and
does not change the assignments of any other family. Therefore, a family
$f$ can be updated at most $|\mathcal W_f|-1$ times. Hence the total
number of updates is at most
\[
\sum_{f\in F} |\mathcal W_f|,
\]
and the algorithm terminates.

Upon termination, there is no family $f$ and assignment
$W\succ_f Y_f$ such that
\[
Y\setminus Y_f\cup W
=
Ch_D(Y\setminus Y_f\cup W)
\]
and the resulting matching satisfies $A$. By
Definition~\ref{def:fairness_preserving_nonwastefulness}, the final
matching is fairness-preserving non-wasteful with respect to $A$.
Since every update preserves $A$, the final matching also satisfies
property $A$.

Finally, suppose that property $A$ can be checked in polynomial time.
For each family $f$, the explicitly represented assignments
$W\succ_fY_f$ can be enumerated. For each candidate assignment, the
no-rejection condition can be checked by computing the daycare choice
functions, and preservation of $A$ can be verified in polynomial time.
Since the number of updates and the total number of candidate
assignments are polynomial in the input size, Algorithm~\ref{alg:FN}
runs in polynomial time given the initial matching $Y^0$.
\end{proof}

\begin{corollary}
\label{cor:FN:SFMR}
There exists a feasible and individually rational matching that is
strongly fair by master list and resettlement and fairness-preserving
non-wasteful with respect to SFMR.
\end{corollary}

\begin{proof}
By Theorem~\ref{theo:existence:SFMR}, there exists a feasible and
individually rational matching satisfying SFMR. Apply
Theorem~\ref{theo:existence:constrained} with $A=\mathrm{SFMR}$.
\end{proof}

Corollary~\ref{cor:FN:SFMR} is an existence result. It does not imply a polynomial-time algorithm for the SFMR-based variant of FN, since verifying SFMR is coNP-complete and the existence algorithm for an SFMR matching is not
known to run in polynomial time.

% \clearpage
\section{Conclusion}

In this paper, we studied the tension between fairness and
non-wastefulness in matching markets with siblings and initial
enrollments. We first developed a unified framework for stability,
fairness, and non-wastefulness that accommodates joint family preferences
and seat withdrawal. We then explored two complementary approaches to
restoring compatibility when stability may fail. On the one hand, we
preserved non-wastefulness while progressively weakening fairness through
replacement, master-list restrictions, and resettlement. In particular,
resettlement restores universal existence when combined with a family
master list, and its dominance-based extension yields a stronger fairness
guarantee. On the other hand, we preserved fairness while relaxing
non-wastefulness through generalized cutoffs and fairness-preserving
non-wastefulness. The proposed GC, FN, SS, and GI algorithms illustrate
four distinct ways of navigating the fairness--non-wastefulness trade-off.

Our results also reveal substantial computational barriers behind these
compatibility guarantees. While several fairness notions can be verified
in polynomial time, deciding whether they can coexist with
non-wastefulness is NP-complete. Resettlement restores existence but does
not eliminate computational difficulty: computing a non-wasteful and
FMR-fair matching is PLS-complete, while verifying the stronger
dominance-based fairness notion is coNP-complete and computing a
non-wasteful matching satisfying it remains PLS-hard. These results
highlight that axiomatic compatibility and computational tractability
need not coincide.

Several directions remain for future research. One natural direction is
to identify preference restrictions or structural conditions under which
the computational barriers established here disappear, particularly for
resettlement-based fairness and bounded family size. Another is to study
incentive properties and mechanism-design questions associated with the
proposed solution concepts and algorithms. Finally, integer programming
formulations and large-scale empirical studies could help evaluate the
practical trade-offs among fairness, non-wastefulness, family welfare, and
the extent of resettlement in real daycare matching markets.

\section*{Acknowledgments}

This work was supported by the Japan Science and Technology Agency under ERATO Grant Number JPMJER2301, and by the Japan Society for the Promotion of Science under Grant Number 25K03186 and 25K21277. 

\bibliographystyle{plainnat}
\bibliography{reference}  

@article{APRS05a,
  author         = {A. Abdulkadiro{\u{g}}lu and P. A. Pathak and A. E. Roth and T. S{\"o}nmez},
  title          = {The {Boston} Public School Match},
  journal        = {American Economic Review},
  volume         = {95},
  number         = {2},
  pages          = {368--371},
  year           = {2005},
}

@article{AbSo99a,
  author         = {Abdulkadiro{\u{g}}lu, Atila and S{\"o}nmez, Tayfun},
  title          = {House allocation with existing tenants},
  journal        = {Journal of Economic Theory},
  volume         = {88},
  number         = {2},
  pages          = {233--260},
  publisher      = {Elsevier},
  year           = {1999},
}

@article{AbSo03b,
  author         = {A. Abdulkadiro{\u{g}}lu and T. S{\"o}nmez},
  title          = {School Choice: A Mechanism Design Approach},
  journal        = {American Economic Review},
  volume         = {93},
  number         = {3},
  pages          = {729--747},
  year           = {2003},
}

@article{ABH14a,
  author         = {I. Ashlagi and M. Braverman and A. Hassidim},
  title          = {Stability in large matching markets with complementarities},
  journal        = {Operations Research},
  volume         = {62},
  number         = {4},
  pages          = {713--732},
  year           = {2014},
}

@article{ABB24a,
  author         = {H. Aziz and A. Baychkov and P. Bir{\'o}},
  title          = {Cutoff stability under distributional constraints with an application to summer internship matching},
  journal        = {Mathematical Programming},
  volume         = {203},
  number         = {1},
  pages          = {247--269},
  publisher      = {Springer},
  year           = {2024},
  doi            = {10.1007/s10107-022-01917-1},
}

@inproceedings{AzSu21a,
  author         = {H. Aziz and Z. Sun},
  title          = {Multi-rank Smart Reserves},
  booktitle      = {Proceedings of the 22nd {ACM} Conference on Economics and Computation, Budapest, Hungary, July 18-23, {EC} 2021},
  pages          = {105--124},
  year           = {2021},
}

@article{BaSo99a,
  author         = {M. Balinski and T. S{\"o}nmez},
  title          = {A tale of two mechanisms: student placement},
  journal        = {Journal of Economic Theory},
  volume         = {84},
  number         = {1},
  pages          = {73--94},
  publisher      = {Elsevier},
  year           = {1999},
}

@techreport{BKS03a,
  author         = {P. Berman and M. Karpinski and A. D. Scott},
  title          = {Approximation Hardness of Short Symmetric Instances of {MAX-3SAT}},
  number         = {TR03-049},
  institution    = {Electronic Colloquium on Computational Complexity},
  type           = {Technical Report},
  year           = {2003},
  url            = {https://eccc.weizmann.ac.il/report/2003/049/},
}

@article{BFI16a,
  author         = {P. Bir{\'o} and T. Fleiner and R. W. Irving},
  title          = {Matching couples with {Scarf}'s algorithm},
  journal        = {Annals of Mathematics and Artificial Intelligence},
  volume         = {77},
  pages          = {303--316},
  publisher      = {Springer},
  year           = {2016},
}

@article{BIS11a,
  author         = {P. Bir{\'o} and R. W. Irving and I. Schlotter},
  title          = {Stable matching with couples: an empirical study},
  journal        = {Journal of Experimental Algorithmics},
  volume         = {16},
  articleno      = {1.2},
  numpages       = {27},
  publisher      = {ACM},
  year           = {2011},
  doi            = {10.1145/1963190.1970372},
  note           = {Article 1.2, 27 pages},
}

@article{BiKl13a,
  author         = {Bir{\'o}, P{\'e}ter and Klijn, Flip},
  title          = {Matching with couples: a multidisciplinary survey},
  journal        = {International Game Theory Review},
  volume         = {15},
  number         = {02},
  pages          = {1340008},
  publisher      = {World Scientific},
  year           = {2013},
}

@inproceedings{BMM14a,
  author         = {P. Bir{\'o} and D. F. Manlove and I. McBride},
  title          = {The hospitals/residents problem with couples: Complexity and integer programming models},
  booktitle      = {International Symposium on Experimental Algorithms},
  pages          = {10--21},
  year           = {2014},
  organization   = {Springer},
}

@article{CEE+22a,
  author         = {J. Correa and N. Epstein and R. Epstein and J. Escobar and I. Rios and N. Aramayo and B. Bahamondes and C. Bonet and M. Castillo and A. Cristi and B. Epstein and F. Subiabre},
  title          = {School choice in {Chile}},
  journal        = {Operations Research},
  volume         = {70},
  number         = {2},
  pages          = {1066--1087},
  publisher      = {INFORMS},
  year           = {2022},
}

@inproceedings{CMMT24a,
  author         = {G. Cs{\'{a}}ji and D. F. Manlove and I. McBride and J. Trimble},
  title          = {Couples Can Be Tractable: New Algorithms and Hardness Results for the Hospitals/Residents Problem with Couples},
  booktitle      = {Proceedings of the Thirty-Third International Joint Conference on Artificial Intelligence, {IJCAI} 2024, Jeju, South Korea, August 3-9, 2024},
  pages          = {2731--2739},
  year           = {2024},
}

@inproceedings{DPB15a,
  author         = {J. Drummond and A. Perrault and F. Bacchus},
  title          = {{SAT} Is an Effective and Complete Method for Solving Stable Matching Problems with Couples},
  booktitle      = {Proceedings of the Twenty-Fourth International Joint Conference on Artificial Intelligence, {IJCAI} 2015, Buenos Aires, Argentina, July 25-31, 2015},
  pages          = {518--525},
  publisher      = {{AAAI} Press},
  year           = {2015},
}

@article{DMP22a,
  author         = {U. Dur and T. Morrill and W. Phan},
  title          = {Family Ties: School Assignment with Siblings},
  journal        = {Theoretical Economics},
  volume         = {17},
  number         = {1},
  pages          = {89--120},
  year           = {2022},
}

@article{EHYY14a,
  author         = {L. Ehlers and I. E. Hafalir and M. B. Yenmez and M. A. Yildirim},
  title          = {School choice with controlled choice constraints: Hard bounds versus soft bounds},
  journal        = {Journal of Economic Theory},
  volume         = {153},
  pages          = {648--683},
  publisher      = {Elsevier},
  year           = {2014},
}

@article{GaSh62a,
  author         = {D. Gale and L. S. Shapley},
  title          = {College admissions and the stability of marriage},
  journal        = {The American Mathematical Monthly},
  volume         = {69},
  number         = {1},
  pages          = {9--15},
  publisher      = {Taylor \& Francis},
  year           = {1962},
}

@article{GJR+21a,
  author         = {A. H. Geitle and {\O}. K. Johnsen and H. F. E. Ruud and K. Fagerholt and C. A. Julsvoll},
  title          = {Kindergarten allocation in {Norway}: An integer programming approach},
  journal        = {Journal of the Operational Research Society},
  volume         = {72},
  number         = {7},
  pages          = {1664--1673},
  year           = {2021},
}

@article{HHKS+17a,
  author         = {N. Hamada and C. Hsu and R. Kurata and T. Suzuki and S. Ueda and M. Yokoo},
  title          = {Strategy-proof school choice mechanisms with minimum quotas and initial endowments},
  journal        = {Artificial Intelligence},
  volume         = {249},
  pages          = {47--71},
  publisher      = {Elsevier},
  year           = {2017},
}

@article{HaKo10a,
  author         = {J. Hatfield and F. Kojima},
  title          = {Substitutes and stability for matching with contracts},
  journal        = {Journal of Economic Theory},
  volume         = {145},
  number         = {5},
  pages          = {1704--1723},
  year           = {2010},
}

@article{HaMi05a,
  author         = {J. W. Hatfield and P. R. Milgrom},
  title          = {Matching with contracts},
  journal        = {American Economic Review},
  volume         = {95},
  number         = {4},
  pages          = {913--935},
  year           = {2005},
}

@article{IMS08a,
  author         = {R. W. Irving and D. F. Manlove and S. Scott},
  title          = {The stable marriage problem with master preference lists},
  journal        = {Discrete Applied Mathematics},
  volume         = {156},
  number         = {15},
  pages          = {2959--2977},
  publisher      = {Elsevier},
  year           = {2008},
}

@article{JPY88a,
  author         = {D. S. Johnson and C. H. Papadimitriou and M. Yannakakis},
  title          = {How easy is local search?},
  journal        = {Journal of Computer and System Sciences},
  volume         = {37},
  number         = {1},
  pages          = {79--100},
  publisher      = {Elsevier},
  year           = {1988},
}

@article{KaKo24a,
  author         = {Kamada, Yuichiro and Kojima, Fuhito},
  title          = {Fair Matching under Constraints: Theory and Applications},
  journal        = {The Review of Economic Studies},
  volume         = {91},
  number         = {2},
  pages          = {1162--1199},
  year           = {2024},
  month          = {March},
  doi            = {10.1093/restud/rdad046},
}

@article{KMT14a,
  author         = {J. Kennes and D. Monte and N. Tumennasan},
  title          = {The day care assignment: A dynamic matching problem},
  journal        = {American Economic Journal: Microeconomics},
  volume         = {6},
  number         = {4},
  pages          = {362--406},
  year           = {2014},
}

@article{KlKl05a,
  author         = {B. Klaus and F. Klijn},
  title          = {Stable matchings and preferences of couples},
  journal        = {Journal of Economic Theory},
  volume         = {121},
  number         = {1},
  pages          = {75--106},
  publisher      = {Elsevier},
  year           = {2005},
}

@article{Klkl07a,
  author         = {B. Klaus and F. Klijn},
  title          = {Paths to stability for matching markets with couples},
  journal        = {Games and Economic Behavior},
  volume         = {58},
  number         = {1},
  pages          = {154--171},
  publisher      = {Elsevier},
  year           = {2007},
}

@article{KPR13a,
  author         = {F. Kojima and P. A. Pathak and A. E. Roth},
  title          = {Matching with couples: Stability and incentives in large markets},
  journal        = {The Quarterly Journal of Economics},
  volume         = {128},
  number         = {4},
  pages          = {1585--1632},
  year           = {2013},
}

@article{KHIY17a,
  author         = {R. Kurata and N. Hamada and A. Iwasaki and M. Yokoo},
  title          = {Controlled school choice with soft bounds and overlapping types},
  journal        = {Journal of Artificial Intelligence Research},
  volume         = {58},
  pages          = {153--184},
  year           = {2017},
}

@article{MMT17a,
  author         = {D. F. Manlove and I. McBride and J. Trimble},
  title          = {``{A}lmost-stable'' matchings in the Hospitals/Residents problem with Couples},
  journal        = {Constraints},
  volume         = {22},
  number         = {1},
  pages          = {50--72},
  publisher      = {Springer},
  year           = {2017},
}

@inproceedings{MOPU07a,
  author         = {D. F. Manlove and G. O'Malley and P. Prosser and C. Unsworth},
  title          = {A Constraint Programming Approach to the Hospitals / Residents Problem},
  booktitle      = {Integration of AI and OR Techniques in Constraint Programming for Combinatorial Optimization Problems},
  pages          = {155--170},
  year           = {2007},
}

@article{MaSc11a,
  author         = {D. Marx and I. Schlotter},
  title          = {Stable assignment with couples: Parameterized complexity and local search},
  journal        = {Discrete Optimization},
  volume         = {8},
  number         = {1},
  pages          = {25--40},
  year           = {2011},
}

@article{McMa10a,
  author         = {E. J. McDermid and D. F. Manlove},
  title          = {Keeping partners together: algorithmic results for the hospitals/residents problem with couples},
  journal        = {Journal of Combinatorial Optimization},
  volume         = {19},
  number         = {3},
  pages          = {279--303},
  publisher      = {Springer},
  year           = {2010},
}

@article{NgVo18a,
  author         = {T. Nguyen and R. Vohra},
  title          = {Near-feasible stable matchings with couples},
  journal        = {American Economic Review},
  volume         = {108},
  number         = {11},
  pages          = {3154--3169},
  year           = {2018},
}

@inproceedings{PDB16a,
  author         = {A. Perrault and J. Drummond and F. Bacchus},
  title          = {Strategy-proofness in the stable matching problem with couples},
  booktitle      = {Proceedings of the 2016 International Conference on Autonomous Agents and Multiagent Systems},
  pages          = {132--140},
  year           = {2016},
}

@article{Ronn90a,
  author         = {E. Ronn},
  title          = {{NP}-complete stable matching problems},
  journal        = {Journal of Algorithms},
  volume         = {11},
  number         = {2},
  pages          = {285--304},
  publisher      = {Elsevier},
  year           = {1990},
}

@article{Roth84a,
  author         = {A. E. Roth},
  title          = {The evolution of the labor market for medical interns and residents: a case study in game theory},
  journal        = {Journal of Political Economy},
  volume         = {92},
  number         = {6},
  pages          = {991--1016},
  publisher      = {The University of Chicago Press},
  year           = {1984},
}

@article{Roth85a,
  author         = {A. E. Roth},
  title          = {The college admissions problem is not equivalent to the marriage problem},
  journal        = {Journal of Economic Theory},
  volume         = {36},
  number         = {2},
  pages          = {277--288},
  publisher      = {Elsevier},
  year           = {1985},
}

@article{Roth08a,
  author         = {A. E. Roth},
  title          = {Deferred acceptance algorithms: History, theory, practice, and open questions},
  journal        = {International Journal of Game Theory},
  volume         = {36},
  number         = {3},
  pages          = {537--569},
  publisher      = {Springer},
  year           = {2008},
}

@article{RoPe99a,
  author         = {A. E. Roth and E. Peranson},
  title          = {The redesign of the matching market for American physicians: Some engineering aspects of economic design},
  journal        = {American Economic Review},
  volume         = {89},
  number         = {4},
  pages          = {748--780},
  publisher      = {American Economic Association},
  year           = {1999},
}

@article{RoVa90a,
  author         = {Roth, Alvin E. and Vande Vate, John H.},
  title          = {Random paths to stability in two-sided matching},
  journal        = {Econometrica},
  volume         = {58},
  number         = {6},
  pages          = {1475--1480},
  year           = {1990},
  doi            = {10.2307/2938326},
}

@article{Scar67a,
  author         = {H. E. Scarf},
  title          = {The core of an {$N$} person game},
  journal        = {Econometrica},
  volume         = {35},
  number         = {1},
  pages          = {50--69},
  year           = {1967},
  doi            = {10.2307/1909383},
}

@article{ScTa25a,
  author         = {D. Scheder and J. Tantow},
  title          = {{PLS}-complete problems with lexicographic cost functions: Max-\emph{k}-{SAT} and {Abelian} Permutation Orbit Minimization},
  journal        = {CoRR},
  volume         = {abs/2510.15712},
  year           = {2025},
  doi            = {10.48550/ARXIV.2510.15712},
  url            = {https://arxiv.org/abs/2510.15712},
  archiveprefix  = {arXiv},
  eprint         = {2510.15712},
  eprinttype     = {arXiv},
}

@article{SoYe22a,
  author         = {T. S{\"o}nmez and M. B. Yenmez},
  title          = {Affirmative Action in {India} via Vertical, Horizontal, and Overlapping Reservations},
  journal        = {Econometrica},
  volume         = {90},
  number         = {3},
  pages          = {1143--1176},
  year           = {2022},
  doi            = {10.3982/ECTA17788},
}

@inproceedings{STM+23a,
  author         = {Z. Sun and Y. Takenami and D. Moriwaki and Y. Tomita and M. Yokoo},
  title          = {Daycare Matching in {J}apan: Transfers and Siblings},
  booktitle      = {Proceedings of the 37th AAAI Conference on Artificial Intelligence, AAAI 2023},
  pages          = {14487--14495},
  year           = {2023},
}

@inproceedings{SYT+24a,
  author         = {Z. Sun and N. Yamada and Y. Takenami and D. Moriwaki and M. Yokoo},
  title          = {Stable matchings in practice: A constraint programming approach},
  booktitle      = {Proceedings of the 38th AAAI Conference on Artificial Intelligence, AAAI 2024},
  pages          = {22377--22384},
  year           = {2024},
}

@inproceedings{SYY25a,
  author         = {Z. Sun and T. Yokoyama and M. Yokoo},
  title          = {Probabilistic analysis of stable matching in large markets with siblings},
  booktitle      = {Proceedings of the Thirty-Fourth International Joint Conference on Artificial Intelligence ({IJCAI})},
  pages          = {4065--4072},
  year           = {2025},
  doi            = {10.24963/ijcai.2025/453},
  url            = {https://www.ijcai.org/proceedings/2025/453},
  note           = {Extended version available at arXiv:2501.13043},
}

@article{STY+23a,
  author         = {T. Suzuki and A. Tamura and K. Yahiro and M. Yokoo and Y. Zhang},
  title          = {Strategyproof Allocation Mechanisms with Endowments and {M-convex} Distributional Constraints},
  journal        = {Artificial Intelligence},
  volume         = {315},
  pages          = {103825},
  year           = {2023},
  doi            = {10.1016/j.artint.2022.103825},
}

@inproceedings{STY18a,
  author         = {T. Suzuki and A. Tamura and M. Yokoo},
  title          = {Efficient allocation mechanism with endowments and distributional constraints},
  booktitle      = {Proceedings of the 17th International Conference on Autonomous Agents and MultiAgent Systems},
  pages          = {50--58},
  year           = {2018},
}

\end{document}